\documentclass[twoside,11pt]{article}
\usepackage{jmlr2e}

\usepackage{xr-hyper}
\usepackage{amsmath}
\usepackage{amsbsy}
\usepackage{mathtools}
\usepackage{mathbbol}
\usepackage{mathrsfs}
\usepackage{bbm} 
\usepackage{enumerate}
\usepackage{url} 
\usepackage{float}
\usepackage[skip=5pt]{caption}
\usepackage{subcaption}
\usepackage[normalem]{ulem}
\usepackage{comment}
\usepackage{hhline}
\usepackage{afterpage}
\usepackage{multirow}
\usepackage{adjustbox}
\usepackage{textcomp}
\usepackage{rotating}
\usepackage{tabularx}
\usepackage{xspace}

\usepackage{color,colortbl,soul}
\usepackage[dvipsnames]{xcolor}
\usepackage{blkarray}
\usepackage{tikz}
\usepackage[T1]{fontenc}
\usepackage{arydshln}
\usepackage{scalerel,stackengine}
\definecolor{Gray}{gray}{0.9}

\allowdisplaybreaks

\makeatletter
\renewcommand*\env@matrix[1][*\c@MaxMatrixCols c]{%
	\hskip -\arraycolsep
	\let\@ifnextchar\new@ifnextchar
	\array{#1}}
\makeatother

\RequirePackage{doi}

\newcommand{\given}{\,|\,}

\stackMath
\newcommand\reallywidehat[1]{%
	\savestack{\tmpbox}{\stretchto{%
			\scaleto{%
				\scalerel*[\widthof{\ensuremath{#1}}]{\kern-.6pt\bigwedge\kern-.6pt}%
				{\rule[-\textheight/2]{1ex}{\textheight}}%WIDTH-LIMITED BIG WEDGE
			}{\textheight}% 
		}{0.5ex}}%
	\stackon[1pt]{#1}{\tmpbox}%
}

\makeatletter
\newcommand*{\addFileDependency}[1]{% argument=file name and extension
  \typeout{(#1)}
  \@addtofilelist{#1}
  \IfFileExists{#1}{}{\typeout{No file #1.}}
}
\makeatother

\def\mathcolor#1#{\@mathcolor{#1}}
\def\@mathcolor#1#2#3{%
	\protect\leavevmode
	\begingroup
	\color#1{#2}#3%
	\endgroup
}

\renewcommand{\tilde}[1]{\widetilde{#1}}

\newcommand{\bolds}[1]{\boldsymbol{#1}}
\newcommand{\calD}{{\cal D}}
\newcommand{\calG}{{\cal G}}

\newcommand{\calL}{{\cal L}}

\newcommand{\calS}{{\cal S}}
\newcommand{\calT}{{\cal T}}
\newcommand{\barcalT}{\overline{\calT}}
\newcommand{\calU}{{\cal U}}

\newcommand{\Cov}{\bolds{C}}

\newcommand{\ba}{\bolds{a}}
\newcommand{\bA}{\bolds{A}}
\newcommand{\bb}{\bolds{b}}

\newcommand{\be}{\bolds{e}}
\newcommand{\bE}{\bolds{E}}
\newcommand{\bbf}{\bolds{f}}
\newcommand{\bF}{\bolds{F}}
\newcommand{\bg}{\bolds{g}}
\newcommand{\bG}{\bolds{G}}

\newcommand{\bH}{\bolds{H}}

\newcommand{\bI}{\bolds{I}}

\newcommand{\bl}{\bolds{\ell}}

\newcommand{\bm}{\bolds{m}}
\newcommand{\bM}{\bolds{M}}

\newcommand{\br}{\bolds{r}}
\newcommand{\bR}{\bolds{R}}

\newcommand{\bu}{\bolds{u}}

\newcommand{\bv}{\bolds{v}}
\newcommand{\bV}{\bolds{V}}
\newcommand{\bw}{\bolds{w}}
\newcommand{\bW}{\bolds{W}}
\newcommand{\bx}{\bolds{x}}

\newcommand{\by}{\bolds{y}}

\newcommand{\bz}{\bolds{z}}

\newcommand{\bzero}{\mathbf{0}}

\newcommand{\balpha}{\bolds{\alpha}}
\newcommand{\bbeta}{\bolds{\beta}}

\newcommand{\btheta}{\bolds{\theta}}
\newcommand{\bphi}{\bolds{\phi}}
\newcommand{\bPhi}{\bolds{\Phi}}

\newcommand{\bSigma}{\bolds{\Sigma}}

\newcommand{\bGamma}{\bolds{\Gamma}}
\newcommand{\bLambda}{\bolds{\Lambda}}
\newcommand{\blambda}{\bolds{\lambda}}

\newcommand{\bmu}{\bolds{\mu}}

\newcommand{\bxi}{\bolds{\xi}}

\newcommand{\bOmega}{\bolds{\Omega}}
\newcommand{\brho}{\bolds{\rho}}

\usepackage{algorithmicx}
\usepackage{algpseudocode}

\usepackage{tcolorbox}
\newtcolorbox[auto counter]{reviewcommentinside}[1][]{box align=center,
    width=0.9\textwidth,
    colframe = teal,
    colback=teal!10,
    code={\spacingset{0.9}},
    #1}

\newtcolorbox[auto counter]{reviewanswerinside}[1][]{box align=center,
    width=0.98\textwidth,
    colframe = orange!10,
    colback= orange!10,
    code={\spacingset{0.9}},
    #1}

\usepackage{tikz}
\usetikzlibrary{bayesnet}

\newcommand{\others}{\textrm{---}}

\definecolor{airforceblue}{rgb}{0.66, 0.84, 0.96}
\definecolor{amber}{rgb}{1.0, 0.85,0.55}
\definecolor{antiquefuchsia}{rgb}{0.77,0.56, 0.71}
\definecolor{antiquewhite}{rgb}{0.88,0.82,0.74}

\usepackage{setspace}

\usepackage{algpseudocode}% http://ctan.org/pkg/algorithmicx
\usepackage{algorithm}% http://ctan.org/pkg/algorithm

\newcommand{\foritem}{\State\hspace{\algorithmicindent}}
\newcommand{\algindent}{\hspace{0.5cm}}
\newcommand{\plusequal}{{\ \displaystyle\mathrel{+}=\ }}
\algnewcommand{\algorithmicgoto}{\textbf{go to}}%
\algnewcommand{\Goto}[1]{\algorithmicgoto~\ref{#1}}%

\newcommand{\blind}{0}
\newcommand{\melange}{\textsc{melange}}

\newcommand{\Par}{\text{Par}}
\newcommand{\Chi}{\text{Chi}}
\newcommand{\Copar}{\text{Copar}}
\newcommand{\mb}{\texttt{mb}}

\date{}

\usepackage{lastpage}
\jmlrheading{25}{2024}{1-\pageref{LastPage}}{1/22; Revised
2/24}{3/24}{22-0083}{Michele Peruzzi and David B. Dunson}
\ShortHeadings{Spatial meshing for general Bayesian multivariate models}{Peruzzi and Dunson}

\firstpageno{1}

\begin{document} 

\title{Spatial meshing for general \\Bayesian multivariate models}

\author{\name Michele Peruzzi \email peruzzi@umich.edu \\
        \addr Department of Biostatistics\\
        University of Michigan\\
        Ann Arbor, MI 48109-2029, USA
        \AND
        \name David B. Dunson \email dunson@duke.edu \\
        \addr Department of Statistical Science\\
        Duke University\\
        Durham, NC 27708-0251, USA}

\editor{Debdeep Pati}
\maketitle 

\begin{abstract}
Quantifying spatial and/or temporal associations in multivariate geolocated data of different types is achievable via spatial random effects in a Bayesian hierarchical model, but severe computational bottlenecks arise when spatial dependence is encoded as a latent Gaussian process (GP) in the increasingly common large scale data settings on which we focus. The scenario worsens in non-Gaussian models because the reduced analytical tractability leads to additional hurdles to computational efficiency. In this article, we introduce Bayesian models of spatially referenced data in which the likelihood or the latent process (or both) are not Gaussian. First, we exploit the advantages of spatial processes built via directed acyclic graphs, in which case the spatial nodes enter the Bayesian hierarchy and lead to posterior sampling via routine Markov chain Monte Carlo (MCMC) methods. Second, motivated by the possible inefficiencies of popular gradient-based sampling approaches in the multivariate contexts on which we focus, we introduce the simplified manifold preconditioner adaptation (SiMPA) algorithm which uses second order information about the target but avoids expensive matrix operations. We demostrate the performance and efficiency improvements of our methods relative to alternatives in extensive synthetic and real world remote sensing and community ecology applications with large scale data at up to hundreds of thousands of spatial locations and up to tens of outcomes. Software for the proposed methods is part of \texttt{R} package \texttt{meshed}, available on CRAN.
\end{abstract}
	
\noindent%
	
% !TEX root = article_jmlr_style_texshop.tex
\section{Introduction} \label{section:intro}
Geolocated data are routinely collected in many fields and motivate the development of geostatistical models based on Gaussian processes (GPs).  GPs are appealing due to their analytical tractability, their flexibility via a multitude of covariance or kernel choices, and their ability to effectively represent and quantify uncertainty.
When Gaussian distributional assumptions are appropriate, GPs may be used directly as correlation models for the multivariate response. Otherwise, flexible models of multivariate spatial association can in principle be built via assumptions of conditional independence of the outcomes on a latent GP encoding space- and/or time-variability, regardless of data type.
The poor scalability of na\"ive implementations of GPs to large scale data is addressed in a growing body of literature. \cite{sunligenton}, \cite{sudipto_ba17} and \cite{Heaton2019} review and compare methods for big data geostatistics. Methods include low-rank approaches \citep{gp_predictive_process, frk}, covariance tapering \citep{taper1, taper2}, domain partitioning \citep{fsa, stein2014}, local approximations \citep{lagp}, and composite likelihood approximations \citep{steinetal2004}. In particular, a popular strategy is to assume sparsity in the Gaussian precision matrix via Gaussian random Markov fields \citep[GRMF;][]{grmfields} which can be represented as sparse undirected graphical models. Proper joint densities are a result of using directed acyclic graphs (DAG), leading to Vecchia's approximation \citep{vecchia88}, nearest-neighbor GPs \citep[NNGPs;][]{nngp}, and generalizations \citep[see e.g.][]{katzfuss_jasa17, katzfuss_vecchia}. DAGs can be designed by taking a small number of ``past'' neighbors after choosing an arbitrary ordering of the data. In models of the response and in the conditionally-conjugate latent Gaussian case, posterior computations rely on sparse-matrix routines for scalability \citep{nngp_algos, jurekkatzfuss2020}, enabling fast cross-validation \citep{conjnngp, sudipto_ss20}. Alternatives to sparse-matrix algorithms involve Gibbs samplers whose efficiency improves by prespecifying a DAG defined on domain partitions, resulting in spatially meshed GPs \citep[MGPs;][]{meshedgp}. These perspectives are reinforced when considering multivariate outcomes (see e.g. \citealt{zhangbanerjee20, deyetal20, spamtrees}).

The literature on scalable GPs predominantly relies on Gaussian assumptions on the outcomes, but in many applied contexts these assumptions are restrictive, inflexible, or inappropriate. For example, vegetation phenologists may wish to characterize the life cycle of plants in mountainous regions using remotely sensed Leaf Area Index (LAI, a count variable) and relate it to snow cover during 8 day periods (SC, a discrete variable whose values range from 0 to 8---see e.g., Figure \ref{fig:sat:data}). Similarly, community ecologists are faced with spatial patterns when considering counts or dichotomous presence/absence data of several animal species (Figure \ref{fig:nabbs:data}). In this article, we address this key gap in the literature, which is how to construct arbitrary Bayesian multivariate geostatistical models which (1) may include non-Gaussian components, (2) lead to efficient computation for massive datasets.

\begin{figure}
    \centering
    \includegraphics[width=.75\textwidth]{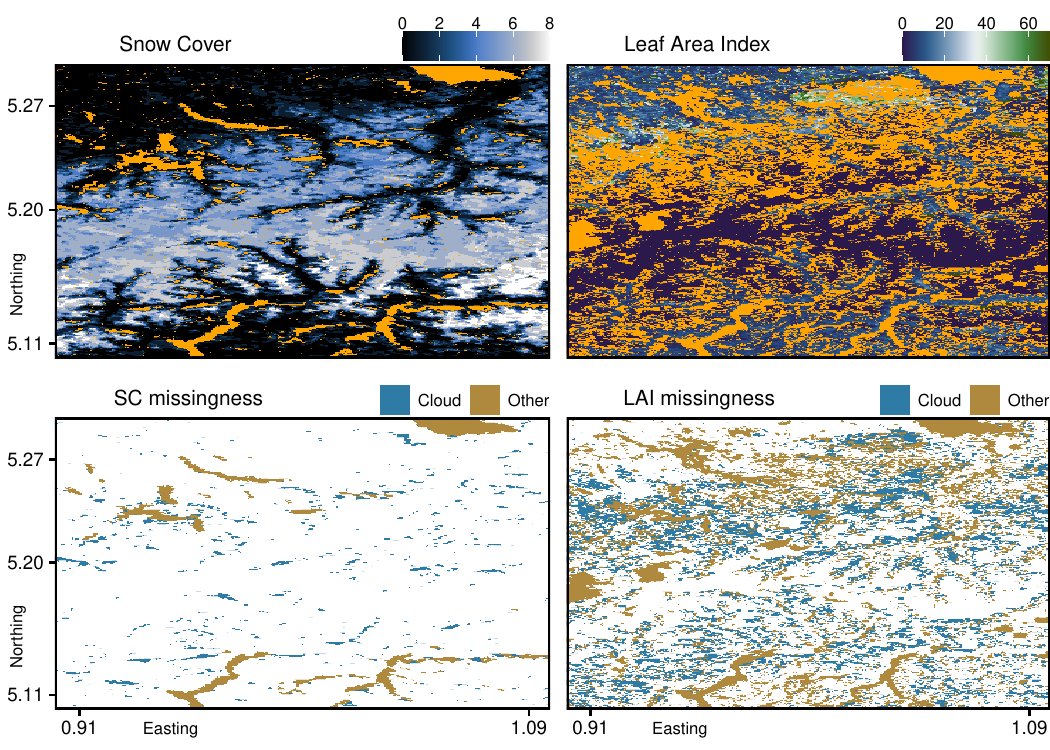}
    \caption{Snow cover (left) and Leaf Area Index, as measured by the MODIS-TERRA satellite. Missing data are in orange. Bottom maps detail the extents of cloud cover and other phenomena negatively impacting data quality.}
    \label{fig:sat:data}
\end{figure}

There are considerable challenges in these contexts for efficient Bayesian computation when avoiding Gaussian distributional assumptions on the outcomes. General purpose Markov chain Monte Carlo (MCMC) methods can in principle be used to draw samples from the posterior distribution of the latent process by making local proposals within accept/reject schemes. However, due to the huge dimensionality of the parameter space, poor mixing and slow convergence are likely. 
For instance, random-walk Metropolis proposals are cheaply computed but lack in efficiency as they overlook the local geometry of the high dimensional posterior.
Alternatively, one may consider gradient-based MCMC methods such as the Metropolis-adjusted Langevin algorithm (MALA; \citealt{robertsstramer02}), Hamiltonian Monte Carlo (HMC; \citealt{hmc_duane, hmc_neal, hmc_conceptual}) and others such as MALA and HMC on the Riemannian manifold \citep{girolamicalderhead11} or the no-U-turn sampler \citep[NUTS;][]{nuts} used in the \texttt{Stan} probabilistic programming language \citep{stan}. These methods are appealing because they modulate proposal step sizes using local gradient and/or higher order information of the target density. Unfortunately, their performance very rapidly drops with parameter dimension \citep{hastings50}. Although it is common in other contexts to rely on subsamples to cheaply approximate gradients, \cite{johndrowetal20}
show that such approximate MCMC algorithms are either slow or have large approximation error.
\begin{figure}
    \centering
    \includegraphics[width=.99\textwidth]{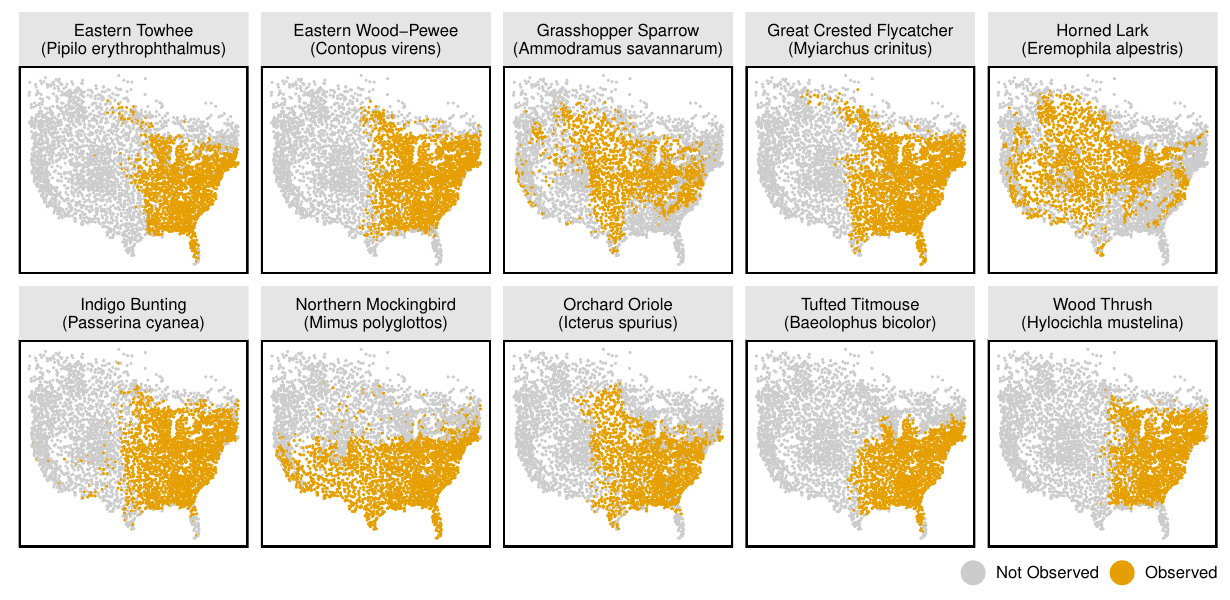}
    \caption{An extract of dichotomized North American Breeding Bird Survey data. Orange points correspond to locations at which at least 1 individual has been observed.}
    \label{fig:nabbs:data}
\end{figure}
Such issues can be tackled by considering low-rank models, which facilitate the design of more efficient proposals as they involve parameters of greatly reduced dimension. Certain low-rank models endowed with conjugate full conditional distributions \citep{bradley_ba, bradley_jasa} lead to always-accepted Gibbs proposals. However, excessive dimension reduction---which may be necessary for acceptable MCMC performance---may lead to oversmoothing of the spatial surface, overlooking the small-range variability that frequently occurs in big spatial data  \citep{gp_pp_biasadj}. 
Alternative dimension reduction strategies via divide-and-conquer methods that combine posterior samples obtained via MCMC from data subsets typically rely on assumptions of independence that are inappropriate in the highly correlated data settings in which we are interested \citep{neiswangeretal2014, wangdunson2014, wangetal2015, nemeth2018, blomstedtetal2019, mesquitaetal2020} or have only considered univariate Gaussian likelihoods \citep{metakriging}.

The poor practical performance of MCMC in high dimensional settings has motivated the development of MCMC-free methods for posterior computation that take advantage of Laplace approximations \citep{sengupta_cressie, vecchialaplace}. In particular, the integrated nested Laplace approximation \citep[INLA;][]{inla} iterates between Gaussian approximations of the conditional posterior of the latent effects, and numerical integrations over the hyperparameters. INLAs are accurate because of the non-negligible impact the Gaussian prior on the latent process has on its posterior; they achieve scalability to big spatial data by forcing sparsity on the Gaussian precision matrix via a GMRF assumption \citep{spde}. INLAs are reliable alternatives to MCMC methods in several settings, but may be outperformed by carefully-designed MCMC methods in terms of accuracy or uncertainty quantification \citep{taylordiggle14}. Furthermore, the practical reliance of INLAs on Matérn covariance models with small dimensional hyperparameters for fast numerical integration makes them less flexible than MCMC methods in multivariate contexts or whenever special-purpose parametric covariance functions are required. 

In this article, we introduce methodological and computational innovations for scalable posterior computations for general non-Gaussian spatial models. Our contributions include a class of Bayesian hierarchical models of multivariate outcomes of possibly different types based on spatial meshing of a latent multivariate process. In our treatment, outcomes can be misaligned---i.e., not all measured at all spatial locations---and relatively large in number, and there is no Gaussian assumption on the latent process. We maintain this perspective when developing posterior sampling methods. In particular, we develop a new Langevin algorithm which, based on ideas related to manifold MALA, adaptively builds a preconditioner but also avoids cubic-cost operations, leading to efficiency improvements in the contexts in which we focus. Our methods enable computations on data of size $10^5$ or more. Unlike low-rank methods, we do not require restrictive dimensionality reduction at the level of the latent process. Unlike INLA, our computational methods are exact (upon convergence) for a class of valid spatial processes which is not restricted to latent GPs with Mat\'ern covariances; furthermore, our methods are hit by a smaller computational penalty in higher-dimensional multivariate settings.
Our methods are generally applicable to models of spatially referenced data, but we highlight the connections between Langevin methods and the Gibbs sampler available for Gaussian outcomes, and we develop new results for latent coregionalization models using MGPs. In applications, we consider Student-t processes, HMC and NUTS, and other cross-covariance models as methodological and computational alternatives to latent GPs, Langevin algorithms, and coregionalization models, respectively. Software for the proposed methods and the related posterior sampling algorithms is available as part of the \texttt{meshed} package for \texttt{R}, available on CRAN.

The article proceeds as follows. Section \ref{sec:setup} outlines our model for spatially-referenced multivariate outcomes of different types and introduces general purpose methods and algorithms for scaling computations to high dimensional spatial data. Section \ref{sec:melange} outlines Langevin methods for posterior sampling of the latent process and introduces a novel algorithm for multivariate spatial models. Section \ref{sec:gaussianlmc} translates the proposed methodologies for the latent Gaussian model of coregionalization. The remaining sections highlight algorithmic efficiency in applications on large synthetic and real world datasets motivated by remote sensing and spatial community ecology. The supplementary material includes alternative constructions of our proposed methods based on latent grids, Student-t processes, and NUTS for posterior computations, in addition to proofs, practical guidelines, additional simulations, and a real world application of our methods in the context of spatial multi-species N-mixture models.

\section{Meshed Bayesian multivariate models for non-Gaussian data} \label{sec:setup}
We introduce our model for multivariate outcomes of possibly different types (e.g. continuous and counts) which also allows for misalignment.
Let $\calG = \{\bA, \bE \}$ be a DAG with nodes $\bA = \{ a_1, \dots, a_M \}$ and edges $\bE = \{ \Par(a) : a \in \bA \}$, where $\Par(a) \subset A$ is referred to as the parent set of $a$. Let $\calD$ be the input domain and $\calS \subset \calD$ denote a user-specified set of ``knots'' or ``reference locations.'' We partition $\calS$ into subsets $\calS_i \subset \calS$ such that $\calS_i \cap \calS_j = \emptyset$ if $i\neq j$ and $\cup_{i=1}^M \calS_i = \calS$. Then, we set up our hierarchical model for multivariate outcomes as:
\begin{equation} \label{eq:meshed_hierarchy}
    \begin{aligned}
    y_j(\bl) \given \eta_j(\bx_j(\bl),  w_j(\bl)), \gamma_j \sim &F_j(\eta_j(\bx_j(\bl), w_j(\bl)), \gamma_j),\\ %&\qquad \eta_j(\bl) = \bx_j(\bl)^\top \bbeta_j + w_j(\bl), \\
    \bbeta_j, \gamma_j \sim \pi(\bbeta_j, \gamma_j) \quad &\btheta \sim \pi(\btheta), \quad
    \bw(\cdot) \sim \Pi_{\calG, \btheta} %= \Pi(\cdot \mid \bw_{[\bl]}), \text{ for } \bl \in \calD \setminus \calS, \\
    %\bw_{\calS_i} = \bw_i &\sim \Pi_{\calG} = \Pi(\cdot \mid \bw_{[i]}), \text{ for } i=1, \dots, M,
    \end{aligned}
\end{equation}
where $F_j$ is the probability distribution of the $j$th outcome, parametrized by an unknown constant $\gamma_j$ and spatially-varying function $\eta_j(\bx_j(\bl), w_j(\bl))$, which includes a $p_j$-dimensional vector of covariates specific for the $j$th outcome, denoted by $\bx_j(\bl)$, whereas $w_j(\bl)$ is the $j$th element of the random vector $\bw(\bl)$, for $j=1,\dots,q$. A common linear assumption leads to $\eta_j(\bl) = \bx_j(\bl)^\top\bbeta_j + w_j(\bl)$. Given a set of locations $\calL \subset \calD$ of size $n_{\calL}$ we denote $\bw_{\calL} = (\bw(\bl_1)^{\top},\bw(\bl_2)^{\top},\ldots,\bw(\bl_{n_{\calL}})^{\top})^{\top}$. We assume $\bw_{\calL}$ is the finite realization at $\calL$ of an infinite-dimensional latent process $\bw(\cdot)$, with law $\Pi_{\calG}$ and density $\pi_{\calG}$, which characterizes spatial/temporal dependence between outcomes. We construct such a process by enforcing conditional independence assumptions encoded in $\calG$ onto the law $\Pi$ of a $q$-variate spatial process (also referred to as the \textit{base} or \textit{parent} process). 
For locations $\bl \in \calS$, we make the assumption that $\pi_{\calG}$ factorizes according to $\calG$. This means $\pi_{\calG}(\bw_\calS \given \btheta) = \prod_{a_i \in \bA} \pi(\bw_{i} \given \bw_{[i]}, \btheta)$, where we denote $\bw_i = \bw_{\calS_i}$ and $\bw_{[i]}$ is the vector of $\bw(\cdot)$ at locations $\bl \in \cup_{a_j \in \Par(a_i)} S_j$ -- i.e. the set of locations mapped to parents of $a_i$. For locations $\bl \in \calU=\calD \setminus \calS$, we assume conditional independence given a set of parents $[\bl] \subset \bA$, which means $\pi_{\calG}(\bw_{\calU} \mid \bw_{\calS}, \btheta) = \prod_{\bl \in \calU} \pi(\bw(\bl) \mid \bw_{[\bl]}, \btheta)$ where $\bw_{[\bl]}$ is a vector collecting realizations of $\bw(\cdot)$ at locations $\calS_{[\bl]} = \cup_{\ba_i \in [\bl]} \calS_i$.

\subsection{DAG and partition choice}
We refer to the method of building spatial processes via sparse DAGs associated to domain partitioning as spatial meshing. Several options for constructing $\calG$ and populating and partitioning $\calS$ are available, but
sparsity assumptions on $\calG$ are necessary to avoid computational bottlenecks in using $\Pi_{\calG}$. Specifically, we restrict our focus on sparse DAGs such that $|\mb(a)| \leq \overline{m}$ for all $a \in \bA$, where $\mb(a)$ is the Markov blanket of $a$, and $\overline{m}$ is a small number. The Markov blanket of a node in a DAG is the set $\mb(a) = \Par(a) \cup \Chi(a) \cup \Copar(a)$ which enumerates the parents of $a$ along with the set of children of $a$, $\Chi(a) = \{ b \in \bA : a \in \Par(b)\}$, and the set of co-parents of $a$, $\Copar(a) = \{ c \in \bA : c \neq a \text{ and } \{ a, c \} \subset \Par(b) \text{ for some } b \in \Chi(a) \}$---this is the set of $a$'s children's other parents. 
We additionally assume that the undirected moral graph $\bar{\calG}$ obtained by adding pairwise edges between co-parents has a small number of colors; if node $a$ has color $c$, then no elements of $\mb(a)$ have the same color. Because our assumptions on the size of the Markov blanket lead to large scale conditional independence, the spatially meshed process $\Pi_{\calG}$ has a simpler dependence structure than the parent process $\Pi$ from which it originates. The ``screening'' effect \citep{stein_screening} makes these assumptions appealing in geostatistical contexts. Furthermore, if the Markov blanket of nodes in $\calG$ can be built to cover their spatial neighbors, then $\Pi_{\calG}$ can provably accurately approximate $\Pi$ in some settings \citep{radgp}. If $\Pi$ is a GP, the $i,j$ entry of the resulting precision matrix is nonzero if the corresponding nodes are in their respective Markov blankets. In the context of Gibbs-like samplers that visit each node of $\calG$, a small Markov blanket bounds the compute time for each step of the algorithm; we take advantage of our assumptions on step \ref{alg:lmc_meshed_posterior:step4} of Algorithm \ref{algorithm:meshed_posterior}. Refer to Algorithm \ref{algorithm:lmc_meshed_posterior} and Section \ref{appx:complexity_lmcqmgp} in the supplement for an account of computational complexity in the coregionalized GP setting.

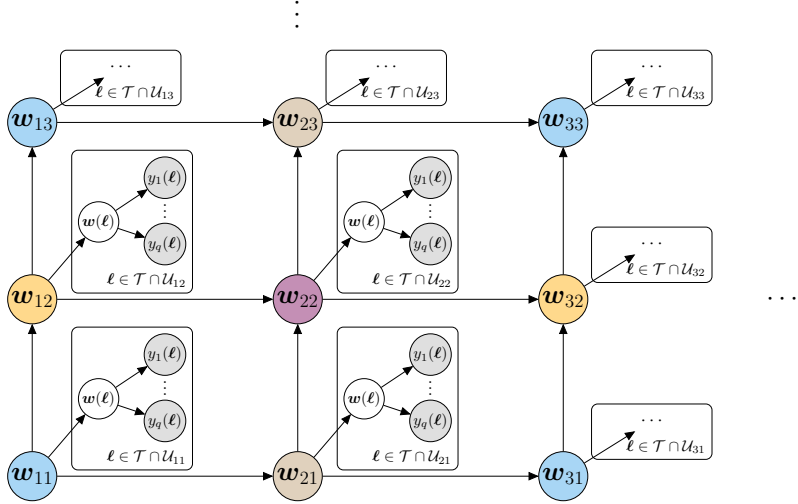
\begin{figure}  
    \centering
\resizebox{.65\columnwidth}{!}{%
\begin{tikzpicture} % airforceblue amber antiquefuchsia
		\node [latent, fill=amber] (6) at (-9, 1) {\Large $\bw_{11}$};
		\node [latent] (9) at (-7.5, 2.75) {$\bw(\bl)$};
		\node [obs] (13) at (-6, 3.75) {$y_1(\bl)$};
		\node  (91) at (-6, 3.1) {$\vdots$};
		\node [obs] (14) at (-6, 2.25) {$y_q(\bl)$};
		
		\node [latent, fill=airforceblue] (0) at (-9, 5) {\Large $\bw_{12}$};
		\node [latent] (10) at (-7.5, 6.75) {$\bw(\bl)$};
		\node [obs] (15) at (-6, 7.75) {$y_1(\bl)$};
		\node  (93) at (-6, 7.1) {$\vdots$};
		\node [obs] (16) at (-6, 6.25) {$y_q(\bl)$};
		
		\node [latent, fill=amber] (2) at (-3, 5) {{\Large $\bw_{22}$}};
		\node [latent] (11) at (-1.5, 6.75) {$\bw(\bl)$};
		\node [obs] (17) at (0, 7.75) {$y_1(\bl)$};
		\node  (94) at (0, 7.1) {$\vdots$};
		\node [obs] (18) at (0, 6.25) {$y_q(\bl)$};
		
		\node [latent, fill=antiquefuchsia] (8) at (-3, 1) {\Large $\bw_{21}$};
		\node [latent] (12) at (-1.5, 2.75) {$\bw(\bl)$};
		\node [obs] (19) at (0, 3.75) {$y_1(\bl)$};
		\node  (92) at (0, 3.1) {$\vdots$};
		\node [obs] (20) at (0, 2.25) {$y_q(\bl)$};
		
		\node [latent, fill=antiquefuchsia] (29) at (-9, 9) {\Large $\bw_{13}$};
		\node  (32) at (-7, 10.25) {$\qquad \cdots \qquad $};
		
		\node [latent, fill=airforceblue] (30) at (-3, 9) {\Large $\bw_{23}$};
		\node  (33) at (-1, 10.25) {$\qquad \cdots \qquad $};
		
		\node [latent, fill=amber] (31) at (3, 9) {\Large $\bw_{33}$};
		\node  (34) at (5, 10.25) {$\qquad \cdots \qquad $};
		
		\node [latent, fill=antiquefuchsia] (22) at (3, 5) {\Large $\bw_{32}$};
		\node  (35) at (5, 6.25) {$\qquad \cdots \qquad $};
		
		\node [latent, fill=airforceblue] (21) at (3, 1) {\Large $\bw_{31}$};
		\node  (36) at (5, 2.25) {$\qquad \cdots \qquad $};
		
		\node  (37) at (-3, 11.5) {{\Large \rotatebox{90}{$\cdots$}}};
		\node  (38) at (8, 5) {{\Large $\cdots$ }};
		
		\edge {6} {9};
		\edge {9} {13};
		\edge {9} {14};
		\edge {6} {8};
		\edge {8} {2};
		\edge {6} {0};
		\edge {0} {2};
		\edge {0} {10};
		\edge {10} {16};
		\edge {10} {15};
		\edge {2} {11};
		\edge {11} {18};
		\edge {11} {17};
		\edge {8} {12};
		\edge {12} {19};
		\edge {12} {20};
		\edge {0} {29};
		\edge {29} {30};
		\edge {2} {30};
		\edge {8} {21};
		\edge {2} {22};
		\edge {21} {22};
		\edge {22} {31};
		\edge {30} {31};
		\edge {29} {32};
		\edge {30} {33};
		\edge {31} {34};
		\edge {22} {35};
		\edge {21} {36};
		
		\plate {data1} {(9)(13)(14)} {$\bl \in \calT\cap \calU_{11}$} ;
		\plate {data2} {(10)(15)(16)} {$\bl \in \calT\cap \calU_{12}$} ;
		\plate {data4} {(11)(17)(18)} {$\bl \in \calT\cap \calU_{22}$} ;
		\plate {data5} {(12)(19)(20)} {$\bl \in \calT\cap \calU_{21}$} ;
		
		\plate {dataplus1} {(32)} {$\bl \in \calT\cap \calU_{13}$} ;
		\plate {dataplus2} {(33)} {$\bl \in \calT\cap \calU_{23}$} ;
		\plate {dataplus3} {(34)} {$\bl \in \calT\cap \calU_{33}$} ;
		\plate {dataplus4} {(35)} {$\bl \in \calT\cap \calU_{32}$} ;
		\plate {dataplus5} {(36)} {$\bl \in \calT\cap \calU_{31}$} ;
\end{tikzpicture}
}
    \caption{Directed acyclic graph representing a special case of model (\ref{eq:meshed_hierarchy}). For simplicity, we omit the directed edges from $(\bbeta_j, \gamma_j)$ to each $y_j(\bl)$, $\bl \in \calT$. If $y_j(\bl)$ is unobserved and therefore $\bl \notin \calT_j$, the corresponding node is missing.}
    \label{fig:qmgpdag}
\end{figure}

Figure \ref{fig:qmgpdag} visualizes (\ref{eq:meshed_hierarchy}) when implemented on a ``cubic'' spatial DAG using row-column indexing of the nodes resulting in $M=M_{\text{row}} \cdot M_{\text{col}}$ and $\calS = \cup_{i=1}^{M_{\text{row}}} \cup_{j=1}^{M_{\text{col}}} S_{ji}$. Even though DAGs are abstract representations of conditional independence assumptions, nodes of the DAG in Figure \ref{fig:qmgpdag} conform to a single pattern (i.e., edges from left and bottom nodes, and to right and top nodes). As a consequence, the moral graph $\bar{\calG}$ only adds undirected edges between $a_{i+1, j}$ and $a_{i, j+1}$ for all $i=1, \dots, M_{\text{row}}-1$ and $j=1, \dots, M_{\text{col}}-1$, leading to cliques of size 3 and 3 colors, irrespective of input data. % (for node $a_{ij}$, we assign one of three color labels $\{ ee, eo, oe, oo \}$ depending on whether $i$ and $j$ are \textit{e}ven or \textit{o}dd, e.g. $\text{Color}(a_{32}) = oe$). 
We refer to this kind of DAG as a cubic DAG as it naturally extends to a hypercube structure in $d>2$ dimensions. 

Once a sparse DAG has been set, one needs to associate each node to a partition $\calS_i$ of $\calS$. With cubic DAGs, the $i$th node of $\calG$ can be associated to the $i$th domain partition found via axis-parallel tiling, or via Voronoi tessellations using a grid of centroids. These two partitioning strategies are equivalent when data have no gaps; otherwise, the latter strategy simplifies the proposal in \cite{meshedgp} and can be used to guarantee that every domain partition includes observations, see e.g. Figure \ref{fig:explain_dagpartition}. 
\begin{figure}
    \centering
    \includegraphics[width=.9\textwidth]{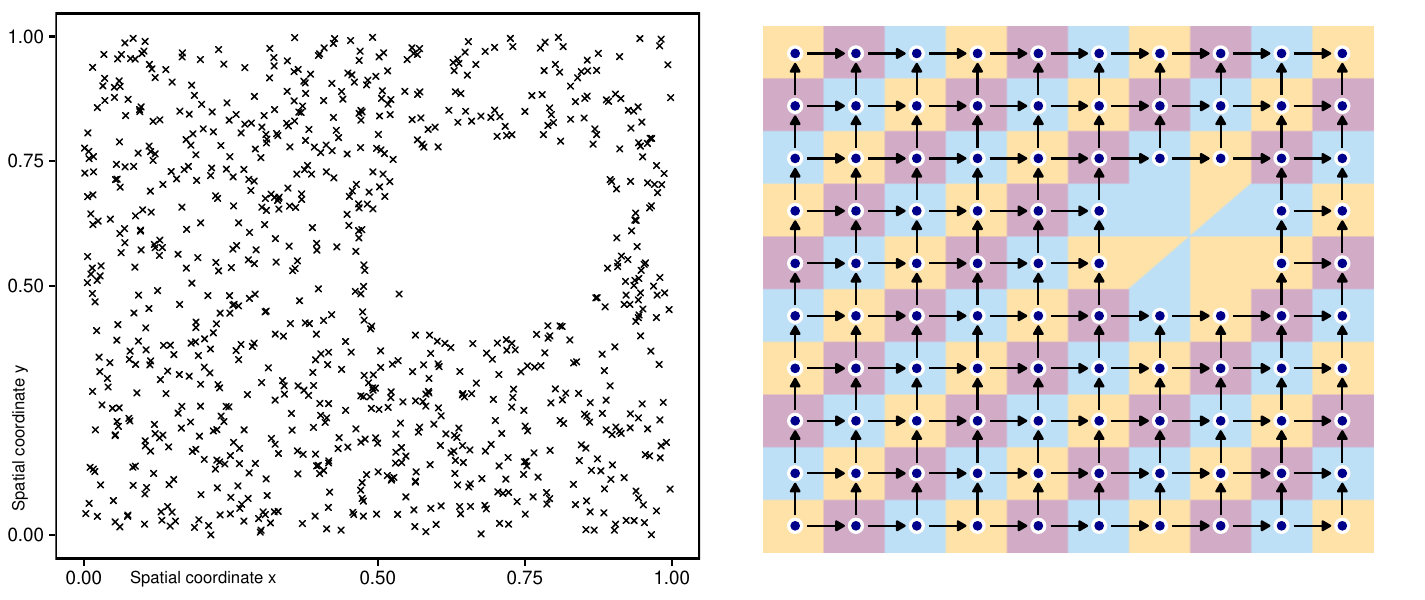}
    \caption{Visualizing cubic DAG and associated domain partitioning. Left: scatter of $\calS$ locations. Right: $\calG$ overlaid to partitions of the domain with colors matching those of $\bar{\calG}$.}
    \label{fig:explain_dagpartition}
\end{figure}
Suppose $\calD_i$, $i=1, \dots, M$ is the chosen domain tessellation. Then, the parent set $[\bl]$ for a location $\bl \in \calU$ can be as simple as letting $[\bl] = \calS_i$ if $\bl \in \calU_i = \calD_i \setminus \calS_i$. 

This general methodology can be used to construct other processes. For instance, dropping the sparsity assumptions on $\calG$, one can recover the base process itself.
\begin{proposition}
If $\calG$ is such that for all $a_{i_j} \in \bA$, $Par(a_{i}) = \{ a_{1}, \dots, a_{i-1} \}$ then $\Pi_{\calG} = \Pi$ at $\calS$, i.e. $\pi_{\calG}(\bw_{\calS}) = \pi(\bw_{\calS})$. The same result holds if $M=1$.
\end{proposition}
\begin{proof} 
Omitting $\btheta$ for clarity, $\pi_{\calG}(\bw_{\calS})$ $= \prod_{a_i \in \bA} \pi(\bw_{i} \given \bw_{[i]})$ $= \pi(\bw_1) \prod_{i=2}^M \pi(\bw_{i} \given \bw_1, \dots, \bw_{i-1}) = \pi(\bw_1, \dots, \bw_M) = \pi(\bw_{\calS})$. If $M=1$ then $\bA = \{a_1 \}$ and $\calS = \calS_1$, $\bE = \{ \emptyset \}$, and the result is immediate.
\end{proof}
Several other spatial process models based on Vecchia's approximation can be derived similarly \citep[][and others]{vecchia88, gp_predictive_process, nngp, katzfuss_jasa17, katzfuss_vecchia, spamtrees} and any of these can be used in place of $\Pi_{\calG}$. 
For example, a Vecchia approximation can be obtained by partitioning $\calS = \{\bl_1, \dots, \bl_{n_{\calS}} \}$ into sets of size $1$; the sparse DAG is then generated by finding the $m$ nearest neighbors of $\bl_i$ from the set $\{\bl_1, \dots, \bl_{i-1}\}$. Heuristic graph coloring algorithms can be used to ensure a degree of parallelization in Algorithm \ref{algorithm:meshed_posterior}.  Unlike in the cubic DAG setting, the number of colors cannot be determined in advance because it is bounded below by clique size, which depends on the order of elements in $\calS$ and their values, and $m$. A larger number of colors corresponds to smaller sampling blocks and may correspond to lower MCMC efficiency when sampling latent surfaces with strong spatial correlations.

DAG and partition choice both relate to the restrictiveness of spatial conditional independence assumptions. Relative to the same partition, adding edges to a DAG brings $\Pi_{\calG}$ closer to $\Pi$ in a Kullback-Leibler (KL) sense \citep[][Section 2]{meshedgp}, and similar reasoning informs placement of knots in recursive treed DAGs \citep{spamtrees}. Here, we consider a cubic DAG and alternative nested partitions. Proposition \ref{prop:kl_partitioning} shows that coarser partitions lead to smaller KL divergence of $\Pi_{\calG}$ from the base process $\Pi$.
\begin{proposition}\label{prop:kl_partitioning}
Consider a $2\times 1$ domain partition $\bw = (\bw_1^\top, \bw_2^\top)^\top$ and suppose $\calG_1$ is a DAG with nodes $\bA_1 = \{a_1, a_2\}$ and the edge $a_1 \to a_2$. Take a finer $3\times 1$ partition nested in the first, i.e. we write $\bw_2 = (\bw_{21}^\top, \bw_{22}^\top)^\top$, and DAG $\calG_2$ such that $\bA_2 = \{a_1, a_{21}, a_{22}\}$, edges $a_1 \to a_{21}$ and $a_{21} \to a_{22}$. Then, $KL(\pi \| \pi_{\calG_1}) \leq KL(\pi \| \pi_{\calG_2})$.
\end{proposition}
\begin{proof}
Since $\pi_{\calG_1} = \pi(\bw_1) \pi(\bw_2 \given \bw_1) = \pi(\bw_1) \pi(\bw_{21} \given \bw_1) \pi(\bw_{22} \given \bw_{21}, \bw_1)$, the coarser partition model can be equivalently written in terms of the finer partition using the DAG $\calG_1^*$ with nodes $\bA_1^* = \bA_2$ and the additional edge $a_1 \to a_{22}$. Then, $\calG_{2}$ is sparser than $\calG_{1}^*$ and therefore $KL(\pi \| \pi_{\calG_1}) \leq KL(\pi \| \pi_{\calG_2})$.
\end{proof}
We provide a discussion in the supplement relating to KL comparisons between non-nested partitioning schemes.

\subsection{Posterior distribution and sampling}
After introducing the set $\calT_j = \{ \bl \in \calT : y_j(\bl) \text{ is observed} \}$, we obtain $\calT_1 \cup \cdots \cup \calT_q = \calT = \{ \bl_1, \dots, \bl_n \}$ as the set of locations at which at least one outcome is observed. Then, we denote as $\barcalT = \calT \setminus \calS$ the set of non-reference locations with at least one observed outcome. The posterior distribution of (\ref{eq:meshed_hierarchy}) is
\begin{equation} \label{eq:meshedposterior}
\begin{aligned}
    &\pi( \{ \bbeta_j, \gamma_j \}_{j=1}^q, \bw_{\calS}, \bw_{\barcalT}, \btheta \given \by_{\calT})  \propto \\
    & \qquad \qquad \pi(\btheta) \pi_{\calG}(\bw_\calS \given \btheta) \pi_{\calG}(\bw_{\barcalT} \given \bw_\calS \btheta) \prod_{j=1}^q \pi(\bbeta_j, \gamma_j) \prod_{\bl \in \calT_j}  dF_j(y_j(\bl) \given w_j(\bl), \bbeta_j, \gamma_j).
\end{aligned}
\end{equation}
Sampling (\ref{eq:meshedposterior}) may proceed via Algorithm \ref{algorithm:meshed_posterior}, where we denote as $\by_i$ the vector of observed outcomes at $\calS_i$ and as $\bw_{\mb(i)}$ the vector of latent effects at the Markov blanket of $\bw_i$, which includes parents, children, coparents of $a_i \in A$, and all locations $\bl \in \calU$ such that $\bw_i$ is part of $\bw_{[\bl]}$. Algorithm \ref{algorithm:meshed_posterior} has the structure of a Gibbs sampler, as the Bayesian hierarchy is expanded to include the spatial DAG $\calG$: at each step of the MCMC loop, the goal is to sample from a full conditional distribution of one random component, conditioning on the most recent value of all the others. Upon convergence, one obtains correlated samples from the target joint posterior density. The lack of conditional conjugacy at steps 1--5, which is expected given our avoidance of simplifying assumptions on $F_j$'s and the base process $\Pi$, implies that 1--5 will require accept/reject steps in which updating parameter $\bz$ proceeds by generating a move to $\bz^*$ via a proposal distribution $q(\cdot \mid \bz)$ and then accepting such move with probability $\min \{1, \frac{ p(\bz^* \mid -) q(\bz \mid \bz^*) }{p(\bz \mid -) q(\bz^* \mid \bz)} \}$ where $p(\bz \mid -)$ is the target distribution to be sampled from. Steps \ref{alg:meshed_posterior:step1} and \ref{alg:meshed_posterior:step2} are generally not a concern in the setting on which we focus due to the independence of $(\bbeta_j, \gamma_j)$ on $(\bbeta_{i}, \gamma_i)$ for $i \neq j$ given the latent process and the fact that the number of covariates for each outcomes is typically small relative to the data size. 

\begin{algorithm}
{ \small
  \caption{Posterior sampling of spatially meshed model (\ref{eq:meshed_hierarchy}) and predictions.}\label{algorithm:meshed_posterior}
  \begin{algorithmic}[1]
  \Statex Initialize $\bbeta_j^{(0)}$ and $\gamma_j^{(0)}$ for $j=1, \dots, q$, $\bw_{\calS}^{(0)}$ $\bw_{\barcalT}^{(0)}$, and $\btheta^{(0)}$
  \Statex \textbf{for} $t \in \{1, \dots, T^*, T^* + 1, \dots, T^* + T\}$ \textbf{do} \Comment{{\footnotesize sequential MCMC loop}}
    \foritem for $j=1, \dots, q$, sample $\bbeta_j^{(t)} \given \by_\calT, \bw_{\calT}^{(t-1)}, \gamma_j^{(t-1)}$ \label{alg:meshed_posterior:step1} 
    \foritem for $j=1, \dots, q$, sample $\gamma_j^{(t)} \given \by_\calT, \bw_{\calT}^{(t-1)}, \bbeta_j^{(t)}$ \label{alg:meshed_posterior:step2} 
    \foritem sample $\btheta^{(t)} \given \bw^{(t-1)}_{\barcalT}, \bw^{(t-1)}_\calS $ \label{alg:meshed_posterior:step3} 
    \Statex \algindent \textbf{for} $c \in \text{Colors}(\calG)$ \textbf{do} \Comment{{\footnotesize sequential}}
    \Statex \algindent \algindent \textbf{for} $i \in \{ i: \text{Color}(a_i) = c \}$ \textbf{do \underline{in parallel}} 
    \foritem \algindent \algindent \algindent sample $\bw_{i}^{(t)} \given \bw_{\mb(i)}^{(t)}, \by_i, \btheta^{(t)}, \{ \bbeta_j^{(t)}, \gamma_j^{(t)} \}_{j=1}^q$ \label{alg:meshed_posterior:step4} \Comment{{\footnotesize reference sampling}}
    %\Statex \algindent \algindent \textbf{end for} 
    %\Statex \algindent \textbf{end for} 
    \Statex \algindent \textbf{for} $\bl \in \barcalT$ \textbf{do \underline{in parallel}} 
    \foritem \algindent sample $\bw(\bl)^{(t)} \given \bw_{[\bl]}^{(t-1)}, \by(\bl), \btheta^{(t)}, \{ \bbeta_j^{(t)}, \gamma_j^{(t)} \}_{j=1}^q$ \label{alg:meshed_posterior:step5} \Comment{{\footnotesize non-reference sampling}}
    %\Statex \algindent \textbf{end for}
  %\Statex \textbf{end for}
  \Statex Assuming convergence has been attained after $T^*$ iterations:
\Statex discard $\{ \bbeta_j^{(t)}, \gamma_j^{(t)} \}_{j=1}^q, \bw_{\calS}^{(t)}, \bw_{\barcalT}^{(t)}, \btheta^{(t)}$ for $t = 1, \dots, T^*$
\Statex \textbf{Output:} Correlated sample of size $T$ with density \[ \{ \bbeta_j^{(t)}, \gamma_j^{(t)} \}_{j=1}^q, \bw_{\calS}^{(t)}, \bw_{\barcalT}^{(t)}, \btheta^{(t)} \sim \pi_\calG(\{ \bbeta_j, \gamma_j \}_{j=1}^q, \bw_{\calS}^{(t)}, \bw_{\barcalT}^{(t)}, \btheta \mid \by_{\calT}).\]
\Statex \textbf{Predict at $\bl^* \in \calU$}: for $t=1, \dots, T$ and $j=1, \dots, q$, sample from $\pi(\bw_{\bl^*}^{(t)} \given \bw_{[\bl^*]}^{(t)}, \btheta^{(t)})$, then from $F_j(w_j(\bl^*)^{(t)}, \bbeta_j^{(t)}, \gamma_j^{(t)})$ 
\end{algorithmic} 
}
\end{algorithm}

It is also typical in these settings to choose a reference set $\calS$ which includes all locations with at least one observed outcome, implying that $\barcalT = \emptyset$; when this is the case, step \ref{alg:meshed_posterior:step5} is not performed in Algorithm \ref{algorithm:meshed_posterior}. We consider alternative strategies to restore flexibility in choosing $\calS$ in the supplementary material. 
Our sparsity assumptions encoded in $\Pi_{\calG}$ via $\calG$ facilitate computations at steps \ref{alg:meshed_posterior:step3} and \ref{alg:meshed_posterior:step4}, which would otherwise be the two major computational bottlenecks. 
Specifically, in step 3 and assuming $\barcalT = \emptyset$, a proposal $\btheta^*$ generated from a distribution $q(\cdot \mid \btheta)$ is accepted with probability $\alpha$
\begin{align}\label{eq:meshed_mhratio}
\alpha = \min \left\{ 1, \frac{ \pi(\btheta^*) %\prod_{\bl \in \barcalT} \pi(\bw(\bl) \given \bw_{[\bl]}, \btheta^*) 
\prod_{i=1}^M \pi(\bw_i \given \bw_{[i]}, \btheta^*) q(\btheta \mid \btheta^*) }{ \pi(\btheta) %\prod_{\bl \in \barcalT} \pi(\bw(\bl) \given \bw_{[\bl]}, \btheta) 
\prod_{i=1}^M \pi(\bw_i \given \bw_{[i]}, \btheta) q(\btheta^* \mid \btheta) } \right\},
\end{align}
whose computation is likely expensive when $\bw_i$ and $\bw_{[i]}$ are high dimensional because the base law $\Pi$ models pairwise dependence of elements of $\bw_i$ based on their spatial distance. As an example, a GP assumption on $\Pi$ leads to $\pi(\bw_i \given \bw_{[i]}, \btheta) = N(\bw_i; \bH_i, \bR_i)$ where $\bH_i = \Cov_{i, [i]} \Cov_{[i]}^{-1}$ and $\bR_i = \Cov_{i} - \bH_i \Cov_{[i], i}$, whose computation has complexity $O(\min\{ n_i^3 q^3, n_{[i]}^3 q^3 \})$. If $n_i$ or the number of parent locations $n_{[i]}$ are large, such density evaluation is computationally prohibitive. Partitioning of $\calS$ ensures that $n_i$ is small for all $i$, and the assumed small Markov blankets of nodes in $\calG$ ensure that the number of parents, and thus $n_{[i]}$, is small.

Step \ref{alg:meshed_posterior:step4} updates the latent process at each partition and is performed in two loops. The outer loop is sequential with a number of sequential steps equalling the number of colors of $\bar{\calG}$, which is small by construction. The inner loop can be performed in parallel or, equivalently, all partitions of the same color can be updated as a single block. In step \ref{alg:meshed_posterior:step4}, the lack of conditional conjugacy implies that proposals for $\bw_i^*$ for all $i=1, \dots, M$ need to be designed and then accepted with probability $\alpha_{i}$
\begin{align} \label{eq:meshed_wratio}
    \alpha_{i} =& \min \left\{1, \frac{ \pi(\bw^*_{i} \given \others) dF(\by_i \given \bw^*_i, \others) q( \bw_i\mid \bw^*_i ) }{ \pi(\bw_{i} \given \others) dF(\by_i \given \bw_i, \others) q( \bw_i^* \mid \bw_i ) } \right\},
\end{align} 
where we denote the full conditional distribution of $\bw_i$ as $\pi(\bw_i \given \others)$ and the outcome densities $dF(\by_i \given \bw^*_i, \others) = \prod_{j=1}^q \prod_{\bl_i \in \calS_i \cap \calT_j} dF_j(y_j(\bl_i) \given w_j(\bl), \bbeta_j, \gamma_j)$. Here, it is desirable to increase the size of each $\bw_i$: in proposition \ref{prop:kl_partitioning} we showed that a coarser partitioning of $\calS_i$ leads to less restrictive spatial conditional independence assumptions. Furthermore, we may expect a smaller number of larger blocks to lead to improved sampling efficiency at step \ref{alg:meshed_posterior:step4}. However, several roadblocks appear when $\bw_i$ is high dimensional. Firstly, evaluating $\pi(\bw^*_i \given \others)/\pi(\bw_i \given \others)$ becomes expensive. Secondly, it is difficult to design an efficient proposal distribution $q(\cdot \mid \bw_i)$ in high dimensions. A random-walk Metropolis (RWM) proposal proceeds by letting $\bw^*_i = \bw_i + \bg_i$ where we let $\bg_i \sim N(\bzero, \bG_i)$, but the $n_iq\times n_iq$ matrix $\bG_i$ must be specified by the user for all $i$, making a RWM proposal unlikely to achieve acceptable performance in practice if $n_i$ is large, especially if one were to take $\bG_i$ as diagonal matrices. Manual specification of $\bG_i$'s can be circumvented via Adaptive Metropolis (AM) methods, which build $\bG_i$ dynamically based on past acceptances and rejections \citep[see e.g.,][]{haario2001, andrieuthoms2008, vihola2012}, or via gradient-based schemes such as HMC, which use local information about the target distribution. However, when the dimension of $\bw_i$ is large the Markov chain will only make small steps and thus negatively impact overall efficiency and convergence regardless of the proposal scheme. The above mentioned issues worsen when $q$ is larger, because spatial meshing via partitioning and a sparse DAG only operates at the level of the spatial domain. 

Finally, while it is easier to specify smaller dimensional proposals, reducing the size of each $\bw_i$ will lead to more restrictive spatial conditional independence assumptions and poorer sampling performance due to high posterior correlations in the spatial nodes.
Therefore, proposal mechanisms for updating $\bw_i$ should (1) be inexpensive to compute and allow for the number of outcomes to increase without overly restrictive spatial conditional independence assumptions, and (2) use local target information with minimal or no user input or tuning.

We begin detailing novel computational approaches in the next section, maintaining a general perspective. We implement our proposals on Gaussian coregionalized meshed process models and detail Algorithm \ref{algorithm:lmc_meshed_posterior} with an account of computational cost in terms of flops and clock time.

\section{Gradient-based sampling of spatially meshed models}\label{sec:melange}
Algorithm \ref{algorithm:meshed_posterior} is essentially a Metropolis-within-Gibbs sampler for updating the latent effects $\bw_{\calT}$ in $M + |\barcalT|$ small dimensional substeps. The setup and tuning of efficient proposals for updating $\bw_i$ remains a challenge and we consider several update schemes below. 
Given our assumption that $\barcalT = \emptyset$, we only need to sample all $\bw_i$'s conditional on their Markov blanket (step \ref{alg:meshed_posterior:step4}). The target full conditional density, for $i=1, \dots, M$, is
\begin{equation}\label{eq:target_fullcond}
    \begin{aligned}
    & p(\bw_{i} \given \others) \propto \pi(\bw_{i} \given \bw_{[i]}, \btheta) \prod_{j \in \{i \to j\}} \pi(\bw_{j} \given \bw_{i}, \bw_{[j]\setminus \{i\}}, \btheta )
    \prod_{\substack{j=1,\dots,q,\\ \bl \in \calS_i \\ y_j(\bl) \text{ is observed}}} dF_j(y_j(\bl) \given w_j(\bl),  \bbeta_j, \gamma_j  ),
    \end{aligned}
\end{equation}
which takes the form $p(\bw_{i} \given \others) \propto \text{[$i$'s parents]} \times \text{[$i$'s children]} \times \text{[data at $i$]}$ and where the last term is a product of  one-dimensional densities due to conditional independence of the outcomes given the latent process. The update of $\bw_i$ proceeds by proposing a move $\bw_i \to \bw_i^*$ using density $q(\cdot \given \bw_i)$; then, $\bw_i^*$ is accepted with probability $\min\{1, \alpha \}$ where $\alpha = \frac{p(\bw_i^* \given \others) q(\bw_i \given \bw_i^*)}{p(\bw_i \given \others) q(\bw_i^* \given \bw_i)}$. We consider gradient-based update schemes that are accessible due to the sparsity of $\calG$ and the low dimensional terms in (\ref{eq:target_fullcond}).

\subsection{Langevin methods for meshed models}
Updating $\bw_{\calS}$ in spatial models via a Metropolis-adjusted Langevin algorithm proceeds in general by proposing a move to $\bw_i^*$ for each $i=1, \dots, M$ via
\begin{equation} \label{eq:precondmala}
    \begin{aligned}
    q(\bw_i^* \mid \bw_i) &= N\left(\bw_i + \varepsilon_i^2 \bM \nabla_{\bw_i} \log p(\bw_i \mid \others)/2, \varepsilon_i^2 \bM \right),\\
    \text{i.e.} \quad
    \bw_i^* &= \bw_i + \frac{\varepsilon_i^2}{2} \bM \nabla_{\bw_i}  \log p(\bw_i \mid \others) + \varepsilon_i \bM^{\frac{1}{2}} \bu,
    \end{aligned}
\end{equation}
where $\bu \sim N(\bzero, I_{n_i})$, $I_{n_i}$ is the identity matrix of dimension $n_i$, $\nabla_{\bw_i} p(\bw_i \mid \others)$ denotes the gradient of the full conditional log-density $\log p(\bw_i \mid \others)$ with respect to $\bw_i$, and $\varepsilon_i$ is a step size specific to node $i$ which can be chosen adaptively via dual averaging \citep[see, e.g., the discussion in][]{nuts}. %This implies that the move $\bw_{\calS} \to \bw_{\calS}^*$ is performed in $M$ substeps. After appropriately attributing each node of $\calG$ to one of $C$ colors such that for all $i$, $a_i$'s color is different from the colors of all elements of $\texttt{mb}(a_i)$, we can update $\bw_{\calS}$ in only $C$ sequential steps, each of which updates $\bw_i$ and $\bw_j$ in parallel if $a_i$ and $a_j$ are of the same color. With spatial meshing, DAGs can be chosen that lead to small $C$.
With (\ref{eq:target_fullcond}) as the target, let $\bbf_i$ be the $n_i q\times 1$ vector that stacks $n_i$ blocks of size $q\times 1$; each of the $n_i$ blocks has $\frac{\delta}{ \delta w_j(\bl)} \log dF(y_j(\bl) \given w_j(\bl), \bbeta_j, \gamma_j)$ as its $j$th element, for $\bl \in \calS_i$, and zeros if $y_j(\bl)$ is unobserved. Then, we obtain
\begin{equation} \label{eq:gradient}
    \begin{aligned}
    \nabla_{\bw_i} \log p(\bw_i \mid \others) &= \bbf_i + \frac{\delta}{ \delta \bw_i} \log p(\bw_{i} \given \bw_{[i]}, \btheta) + \sum_{j \to \{i \to j \}}\frac{\delta}{ \delta \bw_i} \log p(\bw_{j} \given \bw_{i}, \bw_{[j]\setminus \{i\}}, \btheta ).
    \end{aligned}
\end{equation}
The matrix $\bM$ in (\ref{eq:precondmala}) is a preconditioner also referred to as the mass matrix \citep{hmc_neal}. In the simplest setting, one sets $\bM = I_{n_i}$ to obtain a MALA update \citep{robertstweedie96}. If we assume that gradients can be computed with linear cost, MALA iterations run very cheaply in $O(qn_i)$ flops. 
However, we may conjecture that taking into account the geometry of the target beyond its gradient might be advantageous when seeking to formulate efficient updates. Complex update schemes that achieve this goal may operate on the Riemannian manifold \citep{girolamicalderhead11}, but lead to an increase in the computational burden relative to simpler schemes. A special case of manifold MALA corresponding to relatively small added complexity uses a position-dependent preconditioner $\bM_{\bw_i} = \bG_{\bw_i} = \left(-E \left[ \frac{\delta^2}{ \delta \bw_i^2} \log p(\bw_i \given \others ) \right] \right)^{-1}$. Let $\bF_i$ be the $n_i q\times n_i q$ diagonal matrix whose diagonal $\text{diag}(\bF_i)$ is a $n_i q\times 1$ vector that stacks $n_i$ blocks of size $q\times 1$; each of the $n_i$ blocks has $-E\left[ \frac{\delta^2}{ \delta^2 w_j(\bl)} \log dF(y_j(\bl) \given w_j(\bl), \bbeta_j, \gamma_j) \right]$ as its $j$th element, for $\bl \in \calS_i$, and zeros if $y_j(\bl)$ is unobserved. For a target taking the form of (\ref{eq:target_fullcond}) we find
\begin{equation} \label{eq:gmatrix}
    \begin{aligned}
    \bG_{\bw_i}^{-1} &= \bF_i - \frac{\delta^2}{ \delta \bw_i^2} \log p(\bw_{i} \given \bw_{[i]}, \btheta) - \sum_{j \to \{i \to j \}}\frac{\delta^2}{ \delta \bw_i^2} \log p(\bw_{j} \given \bw_{i}, \bw_{[j]\setminus \{i\}}, \btheta );
    \end{aligned}
\end{equation}
this choice leads to an interpretation of (\ref{eq:precondmala}) as a \textit{simplified} manifold MALA proposal (SM-MALA) in which the curvature of the target $p(\bw_i \given \others)$ is assumed constant. We make a connection between a modified SM-MALA update and the Gibbs sampler available when the latent process and all outcomes are Gaussian.
\begin{proposition}\label{prop:smmala_is_gibbs}
In the hierarchical model $\balpha \sim N_k(\balpha; \bm_{\alpha}, \bV_{\alpha})$, $\bx \given \balpha, S \sim N_n(\bx; A\balpha, S)$, consider the following proposal for updating $\balpha \given \bx, S$:
\[ \balpha^* = \balpha + \frac{\varepsilon_1^2}{2} \bG_{\balpha} \nabla_{\balpha}  \log p(\balpha \mid \others) + \varepsilon_2 \bG_{\balpha}^{\frac{1}{2}} \bu, \]
where $\bu \sim N_n(0, I_{n})$, and we set $\varepsilon_1 = \sqrt{2}$, $\varepsilon_2 = 1$. Then, $q(\balpha^* \given \balpha) = p(\balpha^* \given \bx, S)$, i.e. this modified SM-MALA proposal leads to always accepted Gibbs updates.
\end{proposition}
The proof is in the supplement, Section \ref{appx:melange_proof}. A corollary of this proposition in the context of spatially meshed models is that when $F_j(y_j(\bl); w_j(\bl), \bbeta_j, \gamma_j )  = N(y_j(\bl); w_j(\bl) + \bx_j(\bl)^\top \bbeta_j, \gamma_j^2)$ for all $j=1, \dots, q$, an algorithm based on the modified SM-MALA proposal with unequal step sizes for updating $\bw_i$ is a Gibbs sampler. In other words, SM-MALA updates are related to a generalization of Gibbs samplers that have been shown to scale to big spatial data analyses \citep{nngp, nngp_aoas, nngp_algos, meshedgp, spamtrees, grips}. With non-Gaussian outcomes, the probability of accepting the proposed $\bw_i^*$ depends on the ratio $q(\bw_i \given \bw_i^*)/q(\bw_i^* \given \bw_i)$. Computing this ratio requires $O(2q^3n_i^3)$ floating point operations since the dimension of $\bw_i$ and $\bw_i^*$ is $qn_i$ and one needs to compute both $\bG_{\bw_i}^{-\frac{1}{2}}$ and $\bG_{\bw_i^*}^{-\frac{1}{2}}$, e.g. via Cholesky or QR factorizations. For these reasons, SM-MALA proposals may lead to unsatisfactory performance with larger $q$ due to their steeper compute costs relative to simpler MALA updates. We propose a novel adaptive algorithm below to overcome these issues.

\subsection{Simplified Manifold Preconditioner Adaptation}
%Overall performance in the context of Algorithm \ref{algorithm:meshed_posterior} depends on the trade-off between quickly iterating across spatial nodes and more tailored proposals when sampling at each node. In other words, if overall efficiency is primarily determined by the Gibbs structure induced by the spatial DAG, then we may expect complex update schemes to be outperformed by a simpler MALA which quickly iterates through DAG nodes. Nevertheless, the fact that MALA might move too slowly towards high-probability regions of the target may reduce its practical appeal \hl{cite}. 
Using a dense, constant preconditioner $\bM$ in (\ref{eq:precondmala}) rather than the identity matrix leads to a computational cost of $O(q^2n_i^2)$ per MCMC iteration; this cost is larger than MALA updates, but ``good'' choices of $\bM$ might improve overall efficiency. Relative to position-dependent SM-MALA updates, a constant $\bM$ might be convenient if $q$ and/or $n_i$ are large, but it is unclear how $\bM$ can be fixed from the outset in the context of Algorithm \ref{algorithm:meshed_posterior}. In the context of model (\ref{eq:meshed_hierarchy}), we cannot take $\bM^{-1}$ as the expected Fisher information evaluated at the mode  due to the high dimensionality of the latent variables and their dependence on unknown hyperparameters. Adaptive methods may build a preconditioner (or its inverse) by starting from an initial guess $\bM_{(0)}$, then applying smaller and smaller changes to $\bM_{(m)}$ at iteration $m$ to get $\bM_{(m+1)}$. Past values of $\bw_i$ can be used to build a preconditioner: see, e.g., \cite{haario2001}, \cite{andrieuthoms2008}, \cite{marshallroberts12} for adaptive Metropolis, and \cite{atchade06} for MALA. These methods are not immediately advantageous because adaptation using past acceptances may be slow and lead to poor performance, especially in the within-Gibbs contexts in which we operate. Because $O(q^3 n_i^3)$ updates must be performed each time $\bM_{(m)}$ or its inverse are updated due to the need to compute a matrix square root (e.g., Cholesky), slow adaptation methods become increasingly unappealing compared to simpler methods, like MALA, or methods that systematically construct  a position-dependent preconditioner, like RM-MALA.

\begin{algorithm}
  \caption{The $m$th iteration of Simplified Manifold Preconditioner Adaptation.}\label{alg:simpa}
{ \small 
  \begin{algorithmic}[1]
  \Statex \textbf{Setup and inputs}: $d$-dimensional random vector $X \in \mathcal{X} \subseteq \Re^d$, $X \sim P$ whose density $p(\cdot)>0$ is continuous with respect to the Lebesgue measure, assume $K$ is a compact subset of $\mathcal{X}$, 
  fix the constants $D \gg 0$, $\kappa > 0$, $0<T^{\text{adapt}}<\infty$, step size $0 < \varepsilon < D$, denote $\bg_{\bx} = \nabla_{\bx}  \log p(\bx)$, $\tilde{\bg}_{\bx} = \bg_{\bx} \cdot \min \left\{ \frac{D}{\max_i \{ \bg_{\bx}[i] \} }, 1 \right\}$, $\bG_{\bx}^{-1} = -\frac{\delta^2}{ \delta \bx^2} \log p(\bx ) $, $\tilde{\bG}_{\bx}^{-1} = \bG_{\bx}^{-1} \cdot \min \left\{ \frac{D}{\max_i \{ \bG_{\bx}[i,i] \} }, 1 \right\} $, let the sequence $(\gamma_m, m \in \mathbb{N})$ be such that $\gamma_m >0$, $\gamma_m \downarrow 0$.\vspace{.2cm}
  \Statex \textbf{function} SiMPA:%($\bx_{(m)}, \bM_{(m-1)}^{\frac{1}{2}}, \gamma_m, \kappa, \varepsilon $):
\foritem Sample $z \sim U(0,1)$, $v \sim U(0,1)$, $\bu \sim N(\bzero, I_d)$.
\foritem Let $\bmu_{\text{(new)}} = \bx_{(m-1)} + \frac{\varepsilon^2}{2} \bM_{(m-1)}\tilde{\bg}_{\bx_{(m-1)}} $ and propose $\bx_{\text{(new)}} = \bmu_{\text{(new)}} + \varepsilon \bM_{(m-1)}^{\frac{1}{2}} \bu$.
  %\foritem \algindent compute $\overleftarrow{\bM}^{\frac{1}{2}}$ using $\overleftarrow{\bM} = \bM_{(m-1)} + \kappa ( \bG_{\bx_{\text{(new)}}} - \bM_{(m-1)} ) $ \label{alg:simpa:s4} \Comment{{\footnotesize $O(d^3)$ flops}}
  \foritem Let $\bmu_{\text{(back)}} = \bx_{\text{(new)}} + \frac{\varepsilon^2}{2} \bM_{(m-1)}\tilde{\bg}_{\bx_{\text{(new)}}} $.
  \foritem  Compute $$\alpha = \frac{ p(\bx_{\text{(new)}}) }{ p(\bx_{(m-1)})} \cdot \frac{N(\bx_{(m-1)}; \bmu_{\text{(back)}}, \varepsilon^2 \bM_{(m-1)}) }{ N(\bx_{\text{(new)}}; \bmu_{\text{(new)}}, \varepsilon^2 \bM_{(m-1)})}. $$
  \State \algindent \textbf{ if} $\alpha < v$ \textbf{and} $\|\bx_{(new)} - \bx_{(m)}\|<D$, \Comment{{\footnotesize proposal accepted }}
  \State \algindent \algindent Set $\bx_{(m)} = \bx_{\text{(new)}}.$ 
  \State \algindent \textbf{ else:} set $\bx_{(m)} = \bx_{(m-1)}$. \Comment{{\footnotesize proposal rejected }}%and $\bM_{(m)}^{\frac{1}{2}} = \overrightarrow{\bM}^{\frac{1}{2}}$ 
  \vspace{.1cm}
  \Statex \algindent \textbf{ if} $z < \gamma_m$ and $(\bx_{(m)} \in K$ or $m < T^{\text{adapt}})$: \Comment{{\footnotesize adapting}}\label{alg:simpa:s1}
  %\foritem \algindent compute $\overrightarrow{\bM}^{\frac{1}{2}}$ using $\overrightarrow{\bM} = \bM_{(m-1)} + \kappa ( \bG_{\bx_{(m-1)}} - \bM_{(m-1)} ) $ \label{alg:simpa:s1} \Comment{{\footnotesize $O(d^3)$ flops}}
  \foritem \algindent Set $\bM_{(m)}^{-1} = \bM_{(m-1)}^{-1} + \kappa ( \tilde{\bG}_{\bx_{(m)}}^{-1} - \bM_{(m-1)}^{-1} )$ and compute $\bM_{(m)}^{\frac{1}{2}}$. \label{alg:simpa:s8}
  \Statex \algindent \textbf{ else}: \Comment{{\footnotesize not adapting}} 
  \foritem \algindent Set $\bM_{(m)} = \bM_{(m-1)}$. \label{alg:simpa:s9}
\end{algorithmic}
}
\end{algorithm}

To resolve these issues, we outline our Simplified Manifold Preconditioner Adaptation (SiMPA) as Algorithm \ref{alg:simpa}. We present SiMPA in general terms as it operates independently of spatial meshing. The main feature of our algorithm is that it uses the negative Hessian matrix $\bG^{-1}_{\bx}$ to adaptively build a (position-independent) preconditioner. In spatially meshed models and corresponding within-Gibbs posterior sampling algorithms, $\bG_{\bx}$ can be computed easily using (\ref{eq:gmatrix}), also see Appendix \ref{sec:langevin_coreg}. 
Comparatively, an adaptive algorithm similar to \cite{atchade06}, which we label YA-MALA, replaces step \ref{alg:simpa:s8} in Algorithm \ref{alg:simpa} with $\overline{\bx}_{(m)} = \overline{\bx}_{(m-1)} + \kappa (\bx_{(m)} - \overline{\bx}_{(m-1)})$ and $\bM_{(m)} = \bM_{(m-1)} + \kappa (\bGamma_m - \bM_{(m-1)})$, where $\bGamma_m = (\bx_{(m)} - \overline{\bx}_{(m-1)}) (\bx_{(m)} - \overline{\bx}_{(m-1)})^\top + 10^{-6} I_d$ and leaves everything else the same. We show the benefits of adapting via SiMPA compared to YA-MALA in Section \ref{sec:spammix}.

In SiMPA, we reduce the number of iterations with $O(q^3 n_i^3)$ complexity by applying fixed changes to $\bM_{(m)}$ with probability $\gamma_m \to 0$ as $m\to \infty$. As a consequence, the (expected) cost at iteration $m$ is $O(q^2 n_i^2 + \gamma_m q^3 n_i^3)$ rather than $O(q^3 n_i^3)$. 
In the context of spatially meshed models, $n_i$ is small, and the quadratic cost on $q$ can be further reduced via coregionalization (we do so in Section \ref{sec:gaussianlmc}). In our applications, we use $\gamma_m = \mathbf{1}_{(m \le \overline{T})} + \mathbf{1}_{(m > \overline{T})} (m-\overline{T})^{-a}$, where $\mathbf{1}_{A}$ is the indicator for the occurrence of $A$, $\overline{T} < \infty$ is the number of initial iterations during which adaptation always occurs, and $a>0$ is the rate at which the probability of adaptation decays after $\overline{T}$. Small values of the parameter $\kappa$ lead to $\bM_{(m)}$ having long memory of the past. 

We conservatively choose $\overline{T}=500$, $a=1/3$, $\kappa=1/100$ as these values allowed ample burn-in time for all spatial nodes in all our applications. Preliminary analyses with $\overline{T}=1000$ led to an increase in compute time with no advantage in estimation, prediction, or efficiency. On the other hand, $\overline{T}=100$ or $a=1/2$ resulted in lower compute times at the cost of overall performance: letting $\gamma_m$ decay too quickly may lead to an inability to capture the appropriate geometry of the target density.

Because its update does not result in any increase in computational complexity, the step size $\varepsilon$ can be changed at each step, for example via dual averaging (DA) as in Algorithms 5 and 6 of \cite{nuts}. We use the same DA scheme when comparing SiMPA to other gradient-based sampling methods. DA involves updates to $\varepsilon$ at each iteration $m < T^{\text{adapt}}$ and none afterwards. Because $T^{\text{adapt}} < \infty$, DA has no impact on the eventual convergence of the chain.
Finally, the constant $D$ is used to limit the jump size of the proposals as well as bound the index set for adaptation. We need $D$ as well as additional conditions on the algorithmic behavior near the boundary of $K$ to satisfy the containment or bounded convergence condition \citep{robertsrosenthal07, robertsrosenthal09examples} that allows SiMPA to provably converge in total variation distance to the target distribution $P$ even when the state space is not compact. Intuitively, outside of the compact $K$ we stop adapting after iteration $T^{\text{adapt}}$, whereas we perform an infinite (diminishing) adaptation inside it, in order to satisfy the conditions of Theorem 21 of \cite{craiuetal15}. 
\begin{proposition} \label{prop:simpa_converges} 
Suppose $\pi$ is everywhere non-zero and twice differentiable so that $\bg_{\bx}$ and $\bG_{\bx}$ are well defined. Let $\varepsilon >0$, $K \subset \Re^d$, $D > 0$. Additionally assume that if $\bx_{(m)} \in  K$ with $\text{dist}(\bx_{(m)}, K^c) = u$ with $0\leq u \leq 1$ then the proposal is changed to $\bx_{(\text{new})} \sim N( \bx_{(m)} + \frac{\varepsilon^2}{2} \bM_{(T^\text{adapt})} \tilde{\bg}_{\bx_{(m)}}, \varepsilon^2 \bM_{(T^\text{adapt})})$. Then, Algorithm \ref{alg:simpa} converges in distribution to $P$.
\end{proposition}
\begin{comment}
\begin{proof}
We show that SiMPA satisfies the assumptions Theorem 21 of \cite{craiuetal15}. Algorithm \ref{alg:simpa} has by construction bounded jumps, no adaptation outside $K$ after iteration $T^{\text{adapt}}$ and the fixed kernel ouside $K$ is bounded above by $(2\pi)^{-d/2} |\bM_{T^{\text{adapt}}}|^{1/2}$. Because we bound $\bg_{\bx}$ and $\bG_{\bx}$ with $D$ by using $\tilde{\bg}_{\bx}$ and $\tilde{\bG}_{\bx}$, the adaptive proposal kernel inside $K$ is $Q_{\delta}(\bx', \bx)$ where $\delta \in \Delta$ and $\Delta$ is a compact index set. Because outside $K$ we use a fixed proposal kernel with continuous densities with respect to Lebesgue and we assumed that the target density $p(\cdot)$ is also continuous, the $\epsilon$-$\delta$ condition holds (eq. (6) in \citealt{craiuetal15}). Continuity of the target density and the proposal kernels hold by assumption. These assumptions are sufficient for the algorithm to satisfy the containment condition. The additional requirement to achieve convergence is diminishing adaptation, which holds by construction given the decreasing sequence $\{\gamma_m\}$. 
\end{proof}
\end{comment}
The proof is in the supplement, Section \ref{appx:melange_proof}. The containment condition would hold without introducing $K$ and without specifying the behavior of the algorithm near and outside $K$ by assuming that $\mathcal{X}$ itself is compact, which is in principle a restrictive assumption. In practice, $K$ can be fixed large enough so that the chain essentially never leaves it. The SiMPA preconditioner will not in general correspond to the negative Hessian computed at the mode of the target density; rather, by a law of large numbers argument it will converge to the expectation of the negative Hessian of the target density.

\section{Gaussian coregionalization of multi-type outcomes} \label{sec:gaussianlmc}
We have so far outlined general methods and sampling algorithms for big data Bayesian models on multivariate multi-type outcomes. In this section, we remain agnostic on the outcome distributions, but specify a Gaussian model of latent dependency based on coregionalization.
GPs are a convenient and common modeling option for characterizing latent cross-variability. We now assume the base process law $\Pi_{\btheta}$ is a $q$-variate GP, i.e. $\bw(\bl) \sim GP(\bzero, \Cov_{\btheta}(\cdot,\cdot))$. The matrix-valued \emph{cross-covariance} function $\Cov_{\btheta}(\cdot,\cdot)$ is parametrized by $\btheta$ and is such that $\Cov_{\btheta}(\cdot, \cdot) = [\text{cov}\{w_i(\bl),w_j(\bl')\}]_{i,j = 1}^q$, the $q\times q$ matrix with $(i,j)$th element given by the covariance between $w_i(\bl)$ and $w_j(\bl')$. $\Cov_{\btheta}(\cdot, \cdot)$ must be such that $\Cov_{\btheta}(\bl,\bl') = \Cov_{\btheta}(\bl',\bl)^{\top}$ and $\sum_{i=1}^n\sum_{j=1}^n \bz_i^{\top}\Cov_{\btheta}(\bl_i,\bl_j)\bz_j > 0$ for any integer $n$ and any finite collection of points $\{\bl_1,\bl_2,\ldots,\bl_n\}$ and for all $\bz_i \in \Re^{q}\setminus \{\bolds{0}\}$ \citep[see, e.g.,][]{genton_ccov}. %By definition of GP, any finite collection $\calL \subset \{ \bl_1, \dots, \bl_n \}$ leads to the multivariate normal distribution for $\bw_{\calL}$, i.e. $\bw_{\calL} = (\bw(\bl_1)^\top, \dots, \bw(\bl_n)^\top)^\top \sim N(\bzero, \Cov_{\btheta, \calL})$ where $\Cov_{\btheta, \calL}$ stacks pairwise cross-covariances at all locations in $\calL$. 

\subsection{Coregionalized cross-covariance functions}
The challenges in constructing valid cross-covariance functions can be overcome by considering a linear model of coregionalization \citep[LMC;][]{matheron82, wackernagel03, schmidtgelfand}. 
A stationary LMC builds $q$-variate processes via linear combinations of $k$ univariate processes, i.e. $\bw(\bl) = \sum_{h=1}^k \blambda_h v_h(\bl) = \bLambda \bv(\bl)$, where $\bLambda = [\blambda_1,\ldots,\blambda_k]$ is a $q \times k$ full (column) rank matrix with $(i,j)$th entry $\lambda_{ij}$, whose $i$th row is denoted $\blambda_{[i,:]}$, and each $v_j(\bl)$ is a univariate spatial process with correlation function $\rho_j(\bl, \bl') = \rho(\bl, \bl'; \bphi_j)$, and therefore $\btheta = ( \text{vec}(\bLambda)^\top, \bPhi^\top )^\top$ where $\bPhi = (\bphi_1^\top, \dots, \bphi_k^\top)^\top$. 
Independence across the $k\leq q$ components of $\bv(\bl)$ implies $\text{cov}\{v_j(\bl), v_h(\bl') \} = 0$ whenever $h\neq j$, and therefore $\bv(\bl)$ is a multivariate process with diagonal cross-correlation $\brho(\bl, \bl'; \bPhi)$. As a consequence, the $q$-variate $\bw(\cdot)$ process cross-covariance is defined as $\Cov_{\btheta}(\bl, \bl') = \bLambda \brho(\bl, \bl'; \bPhi) \bLambda^\top = \sum_{h=1}^k \blambda_h \blambda_h^{\top} \rho(\bl, \bl', \bphi_h)$. If $\| \bl - \bl' \| = 0$, then $\Cov_{\btheta}(\bzero) = \bLambda \brho(\bzero; \bPhi) \bLambda^\top = \bLambda \bLambda^\top$ since $\brho(\bzero; \bPhi) = \bI_k$. Therefore, when $k=q$, $\bLambda$ is identifiable e.g. as a lower-triangular matrix with positive diagonal entries corresponding to the Cholesky factorization of $\Cov_{\btheta}(\bzero)$ \citep[see e.g.,][and references therein for Bayesian LMC models]{finley2008,zhangbanerjee20}. When $k<q$, a coregionalization model is interpretable as a latent spatial factor model. For a set $\calL = \{\bl_1, \dots, \bl_n \}$ of locations, we let $\brho_{\bPhi, \calL}$ be the $kn \times kn$ block-matrix whose $(i,j)$ block is $\brho(\bl_i, \bl_j, \bphi)$--which has zero off-diagonal elements--and thus $\Cov_{\btheta, \calL} = (I_n \otimes \bLambda) \brho_{\bPhi, \calL} (I_n \otimes \bLambda^\top)$. Notice that the $qn \times 1$ vector $\bw_\calL$ can be represented by a $n\times q$ matrix $\bW$ whose $j$th column includes realizations of the $j$th margin of the $q$-variate process. Assuming a GP, we find $\bw_{\calL} = \text{vec}(\bW^\top) \sim N(\bzero, \Cov_{\btheta, \calL})$. We can also equivalently represent process realizations by outcome rather than by location: if we let $\tilde{\bw}_{\calL} = \text{vec}(\bW)$ then $\tilde{\bw}_{\calL} \sim N(\bzero, Q \Cov_{\btheta, \calL} Q^\top )$ where $Q$ is a permutation matrix that appropriately reorders rows of $\Cov_{\btheta, \calL}$ (and thus, $Q^\top$ reorders its columns). We can write $Q \Cov_{\btheta, \calL} Q^\top = \tilde{\Cov}_{\btheta, \calL} = (\bLambda^\top \otimes I_n) \tilde{\brho}_{\bPhi, \calL} (\bLambda \otimes I_n) = (\bLambda^\top \otimes I_n) J \brho_{\bPhi, \calL} J^\top (\bLambda \otimes I_n)$ where $J$ is a $nk \times nk$ permutation matrix that operates similarly to $Q$ but on the $k$ components of the LMC. Here, $ \tilde{\brho}_{\bPhi, \calL}$ is a block-diagonal matrix whose $j$th diagonal block is $\rho_{j, \calL}$, i.e. the $j$th LMC component correlation matrix at all locations. This latter representation clarifies that prior independence (i.e., a block diagonal $ \tilde{\brho}_{\bPhi, \calL}$) does not translate to independence along the $q$ outcome margins once the loadings $\bLambda$ are taken into account (in fact, $\Cov_{\btheta, \calL}$ is dense).

\subsection{Latent GP hierarchical model}
In practice, LMCs are advantageous in allowing one to represent dependence across $q$ outcomes via $k\ll q$ latent spatial factors.
We build a multi-type outcome spatially meshed model by specifying $\Pi$ in (\ref{eq:meshed_hierarchy}) as a latent Gaussian LMC model with MGP factors
\begin{equation} \label{eq:latent_gaussian_lmc}
    \begin{aligned}
    y_j(\bl) \given \eta_j(\bl), \gamma_j \sim F_j&(\eta_j(\bl), \gamma_j), \\
    \eta_j(\bl) = \bx_j(\bl)^\top \bbeta_j + \blambda_{[j,:]} \bv(\bl), & \qquad \qquad v_h(\cdot) \sim MGP_{\calG}(\bzero, \rho_h(\cdot, \cdot)), h = 1, \dots, k
    \end{aligned}
\end{equation}
whose posterior distribution is
\begin{equation} \label{eq:lmcposterior}
\begin{aligned}
    \pi( \{ \bbeta_j^{(t)}, \gamma_j^{(t)} \}_{j=1}^q, \bv_{\calT}, &\bPhi, \bLambda \given \by_{\calT}) \propto \pi(\bPhi) \prod_{h=1}^k \prod_{i=1}^M \pi(\bv_{h, i} \given \bv_{h, [i]} \bphi_h) \cdot\\ &\qquad\qquad\qquad  \prod_{j=1}^q \left( \pi(\bbeta_j, \gamma_j) 
    \cdot \prod_{\bl \in \calT_j}  dF_j(y_j(\bl) \given v_j(\bl), \blambda_{[j,:]}, \bbeta_j, \gamma_j) \right).
\end{aligned}
\end{equation}
The LMC assumption on $\bw(\cdot)$ using MGP margins leads to computational simplifications in evaluating the density of the latent factors. For each of the $M$ partitions, we now have a product of $k$ independent Gaussian densities of dimension $n_i$ rather than a single density of dimension $qn_i$. 

\subsection{Spatial meshing of Gaussian LMCs}\label{sec:meshed_lmc}
When seeking to achieve scalability of LMCs to large scale data via spatial meshing, it is unclear whether one should act directly on the $q$-variate spatial process $\bw(\cdot)$ obtained via coregionalization, or independently on each of the $k$ LMC component processes. We now show that the two routes are equivalent with MGPs if a single DAG and a single domain partitioning scheme are used. 
\begin{algorithm}
{ \small
  \caption{Posterior sampling and prediction of LMC model (\ref{eq:meshed_hierarchy}) with MGP priors.}\label{algorithm:lmc_meshed_posterior}
  \begin{algorithmic}[1]
  \Statex Initialize $\bbeta_j^{(0)}$, $\bLambda^{(0)}$ and $\gamma_j^{(0)}$ for $j=1, \dots, q$, $\bv_{\calS}^{(0)}$, and $\bPhi^{(0)}$
  \Statex \textbf{for} $t \in \{1, \dots, T^*, T^* + 1, \dots, T^* + T\}$ \textbf{do} \Comment{{\footnotesize sequential MCMC loop}}
    \Statex \algindent \textbf{for} $j=1, \dots, q$, \textbf{do \underline{in parallel}} 
    \foritem \algindent use SiMPA to update $\bbeta_j^{(t)}, \blambda_{[j,:]}^{(t)} \given \by_\calT, \bv_{\calS}^{(t-1)}, \gamma_j^{(t-1)}$ \label{alg:lmc_meshed_posterior:step1} \Comment{{\footnotesize $O(nq(p+k)^2$)}}
    %\Statex \algindent \textbf{end for}
    \Statex \algindent \textbf{for} $j=1, \dots, q$, \textbf{do \underline{in parallel}} 
    \foritem \algindent use Metropolis-Hastings to update $\gamma_j^{(t)} \given \by_\calT, \bv_{\calS}^{(t-1)}, \bbeta_j^{(t)}, \blambda_{[j,:]}^{(t)} $ \label{alg:lmc_meshed_posterior:step2} \Comment{{\footnotesize $O(nq)$}}
    %\Statex \algindent \textbf{end for}
    \foritem use Metropolis-Hastings to update $\bPhi^{(t)} \given \bv^{(t-1)}_\calS $ \label{alg:lmc_meshed_posterior:step3} \Comment{{\footnotesize $O(n k d^3  m^2)$ }}
    \Statex \algindent \textbf{for} $c \in \text{Colors}(\calG)$ \textbf{do} \Comment{{\footnotesize sequential}}
    \Statex \algindent \algindent \textbf{for} $i \in \{ i: \text{Color}(a_i) = c \}$ \textbf{do \underline{in parallel}} 
    \foritem \algindent \algindent use SiMPA to update $\bv_{i}^{(t)} \given \bv_{\mb(i)}^{(t)}, \by_i, \bLambda^{(t)}, \bPhi^{(t)}, \{ \bbeta_j^{(t)}, \gamma_j^{(t)} \}_{j=1}^q$ \label{alg:lmc_meshed_posterior:step4} \Comment{{\footnotesize $O(nmk^2)$}}
    %\Statex \algindent \algindent \textbf{end for} 
    %\Statex \algindent \textbf{end for} 
  %\Statex \textbf{end for}
  \Statex Assuming convergence has been attained after $T^*$ iterations:
\Statex discard $\{ \bbeta_j^{(t)}, \gamma_j^{(t)} \}_{j=1}^q, \bv_{\calS}^{(t)}, \bLambda^{(t)}, \bPhi^{(t)}$ for $t = 1, \dots, T^*$
\Statex \textbf{Output:} Correlated sample of size $T$ with density \[ \{ \bbeta_j^{(t)}, \gamma_j^{(t)} \}_{j=1}^q, \bv_{\calS}^{(t)}, \bLambda^{(t)}, \bPhi^{(t)} \sim \pi_\calG(\{ \bbeta_j, \gamma_j \}_{j=1}^q, \bv_{\calS}^{(t)}, \bLambda, \bPhi, \mid \by_{\calT}).\]
\Statex \textbf{Predict at $\bl^* \in \calU$}: for $t=1, \dots, T$ and $j=1, \dots, q$, sample from $\pi(\bv_{\bl^*}^{(t)} \given \bv_{[\bl^*]}^{(t)}, \bPhi^{(t)})$, then from $F_j(w_j(\bl^*)^{(t)}, \bbeta_j^{(t)}, \blambda_{[j,:]}^{(t)}, \gamma_j^{(t)})$ 
\end{algorithmic} }
\end{algorithm}

If the base process $\Pi$ is a $q$-variate coregionalized GP, then for $i=1, \dots, M$ the conditional distributions are $\pi(\bw_i \given \bw_{[i]}, \btheta) = N( \bw_i; \bH_i \bw_{[i]}, \bR_i )$ where $\bH_i = \Cov_{i, [i]} \Cov_{[i]}^{-1}$, $\bR_{i} = \Cov_i - \bH_i \Cov_{[i], i}$, and $\Cov(\bl, \bl') = \bLambda \brho(\bl, \bl') \bLambda^\top$ (we omit the $\btheta$ and $\bPhi$ subscripts for simplicity). When sampling, (\ref{eq:target_fullcond}) simplifies to
\begin{equation}\label{eq:mgp_fullcond}
    \begin{aligned}
     p(\bw_{i} \given \textrm{---}) \propto & N(\bw_{i}; \bH_i \bw_{[i]}, \bR_i) \prod_{j \in \{i \to j\}} N(\bw_{j}; \bH_{i\to j} \bw_i + \bH_{[j]\setminus \{i\}} \bw_{[j]\setminus \{i\}}, \bR_j) \cdot \\
    &\qquad\qquad\qquad \cdot \qquad
    \prod_{\substack{j=1,\dots,q,\\ \bl \in \calS_i \\ y_j(\bl) \text{ is observed}}} dF_j(y_j(\bl) \given w_j(\bl),  \bbeta_j, \gamma_j  ),
    \end{aligned}
\end{equation}
where the notation $i\to j$ and $[j]\setminus \{i\}$ refers to the partitioning of $\bH_j$ by column into $\bH_j = [\bH_{i\to j}\ \bH_{[j]\setminus \{i\}} ]$ and thus $\bw_{[j]\setminus \{i\}}$ corresponds to blocks of $\bw_{[j]}$ excluding $\bw_i$ (i.e. the co-parents of $i$ relative to node $j$). $\bH_i$ and $\bR_i$  have dimension $qn_i \times qn_{[i]}$ and $qn_i \times qn_i$, respectively. Their dimension depends on $q$, and the following proposition uncovers their structure.

\begin{proposition}\label{prop:mgp_lmc_equivalence} A $q$-variate MGP on a fixed DAG $\calG$, a domain partition $\mathbf{T}$, and a LMC cross-covariance function $\Cov_{\btheta}$ is equal in distribution to a LMC model built upon $k$ independent univariate MGPs, each of which is defined on the same $\calG$ and $\mathbf{T}$.
\end{proposition}
\noindent The proof proceeds by showing that if $\bw_i=(I_{n_i} \otimes \bLambda) \bv_i$ then $\pi(\bw_i \given \bw_{[i]}) = \pi(\bv_{i} \given \bv_{[i]})$ and that for all $i=1,\dots,M$ we can write $\pi(\bv_{i} \given \bv_{[i]}) = \prod_{h=1}^k \pi(v^{(h)}_i \given v^{(h)}_{[i]})$, concluding that $\pi_{\calG}(\bw_{\calS})=\prod_{i=1}^M \pi(\bw_i \given \bw_{[i]}) = \prod_{i=1}^M \prod_{h=1}^k  \pi(v^{(h)}_i \given v^{(h)}_{[i]}) = \prod_{h=1}^k \pi_{\calG}^{(h)}(v^{(h)}_{\calS})$ where $\pi_{\calG}^{(h)}$ is the density of the $h$th independent univariate MGP using $\calG$, $\mathbf{T}$, and correlation function $\rho_h(\cdot, \cdot)$. The complete derivation is available in the supplement.
A corollary of Proposition \ref{prop:mgp_lmc_equivalence} is that a \textit{different} spatially meshed GP can be constructed via unequal spatial meshing (i.e., different graphs and partitions) along the $k$ margins; this result intuitively says that an MGP behaves like a standard GP with respect to the construction of multivariate processes via LMCs and in other words, there is no loss in flexibility when using MGPs compared to the full GP.
The supplementary material provides details on $\nabla_{\bv_i} \log p(\bv \given \others)$ and $\bG_{\bv_i}$ for posterior sampling of the latent meshed Gaussian LMC models via Algorithm \ref{algorithm:meshed_posterior}. 

\section{Applications on bivariate non-Gaussian data} \label{sec:large_scale_spatial_counts}
We concentrate here on a scenario in which two possibly misaligned non-Gaussian outcomes are measured at a large number of spatial locations and we aim to jointly model them. We will consider a larger number of outcomes in Section \ref{sec:jsdm}, in the context of community ecology.
In addition to the analysis presented here, the supplement (Section \ref{appx:applications}) includes (1) a comparison of methods across 750 multivariate synthetic datasets, and (2) performance assessments of multiple sampling schemes in multivariate multi-type models using latent coregionalized QMGPs.

\subsection{Illustration: bivariate log-Gaussian Cox processes} \label{sec:lgcp}
When modeling spatial point patterns via log-Gaussian Cox processes with the goal of estimating the intensity surface, one typically proceeds by counting occurrences within cells in a regular grid of the spatial domain. We simulate this scenario by generating a bivariate Poisson outcome at each location of a $120\times 120$ regular grid, for a data dimension of $qn = 28800$. In model (\ref{eq:meshed_hierarchy}), we let $F_j$ be a Poisson distribution with intensity $\exp\{ \eta_j(\bl)\}$ at $\bl \in [0,1]^2$, where $\eta_j(\bl)= \bx(\bl)^\top \bbeta_j + w_j(\bl)$ is the log-intensity for count outcome $j$. We sample $3$ correlated covariates at each location independently as $\bx^\top(\bl) \sim N_3(\bzero, \bSigma_x)$ where $\bSigma_x$ is a correlation matrix with off-diagonal elements $\sigma_{12} = 0.9, \sigma_{13} = -0.3, \sigma_{23} = -0.6$, and we let $\bbeta_1 = (-0.5, -1, 0)^\top, \bbeta_2 = (-1, -0.5, 0.5)^\top$.
We fix the latent process $\Pi$ in one scenario as a coregionalized GP and in another as a coregionalized NNGP. In both cases, $w_j(\bl) = \blambda_{[j,:]} \bv(\bl)$ and $\bLambda\bLambda^{\top} = (\sigma_{ij})_{i,j=1,2}$ where $\sigma_{11} = 4$, $\sigma_{12} = \sigma_{21} = -1.3$, $\sigma_{22} = 1$, which yields a latent cross-correlation between the two outcomes of $\rho = -0.65$; the two spatial correlations used in the LMC model are $\rho_h(\|\bl - \bl' \|) = \exp\{ -\phi_h \|\bl - \bl' \| \}$ and we let $\phi_1 = 1.5, \phi_2 = 2.5$. We use R package \texttt{GpGp} to generate an NNGP using maxmin ordering of the spatial locations and $10$ neighbors. 
We depict the latent NNGP process along with the full data in Figure \ref{fig:mvpoisson_data}. We introduce missing values at 20\% of the spatial locations, independently for each outcome. As a result, our training data are misaligned.

\begin{figure}
    \centering
    \includegraphics[width=\textwidth]{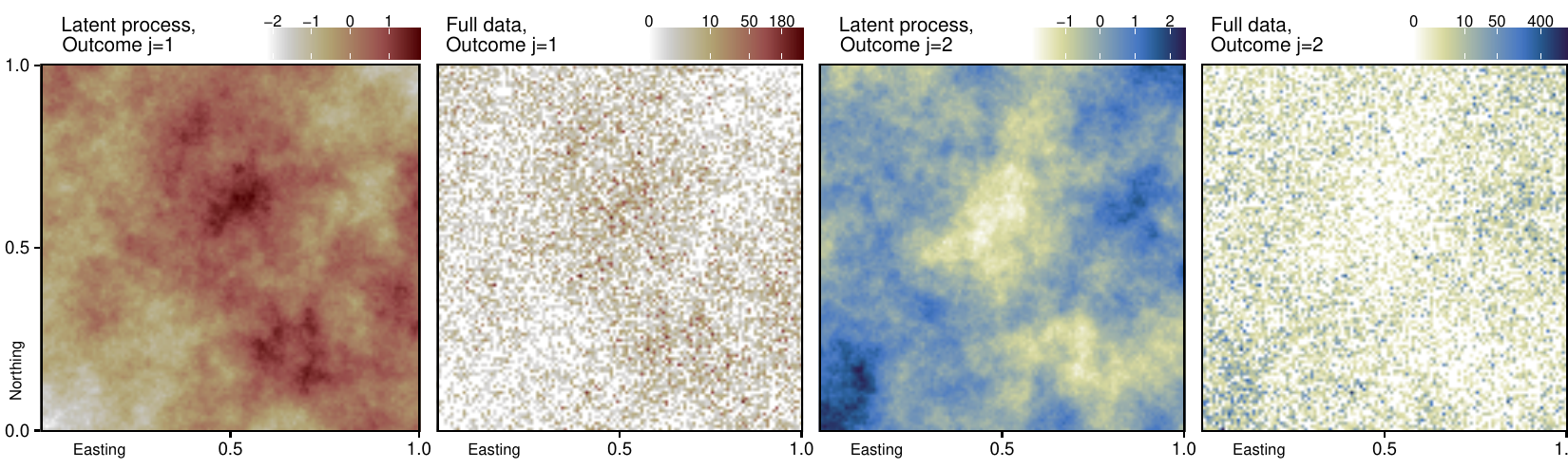}
    \caption{Latent NNGP process realization and corresponding synthetic count data at $14,400$ spatial locations for correlated spatial outcomes. We omit the plots corresponding to the unrestricted GP scenario as they are visually indistinguishable.}
    \label{fig:mvpoisson_data}
\end{figure}
\begin{figure}
    \centering
    \includegraphics[width=0.99\textwidth]{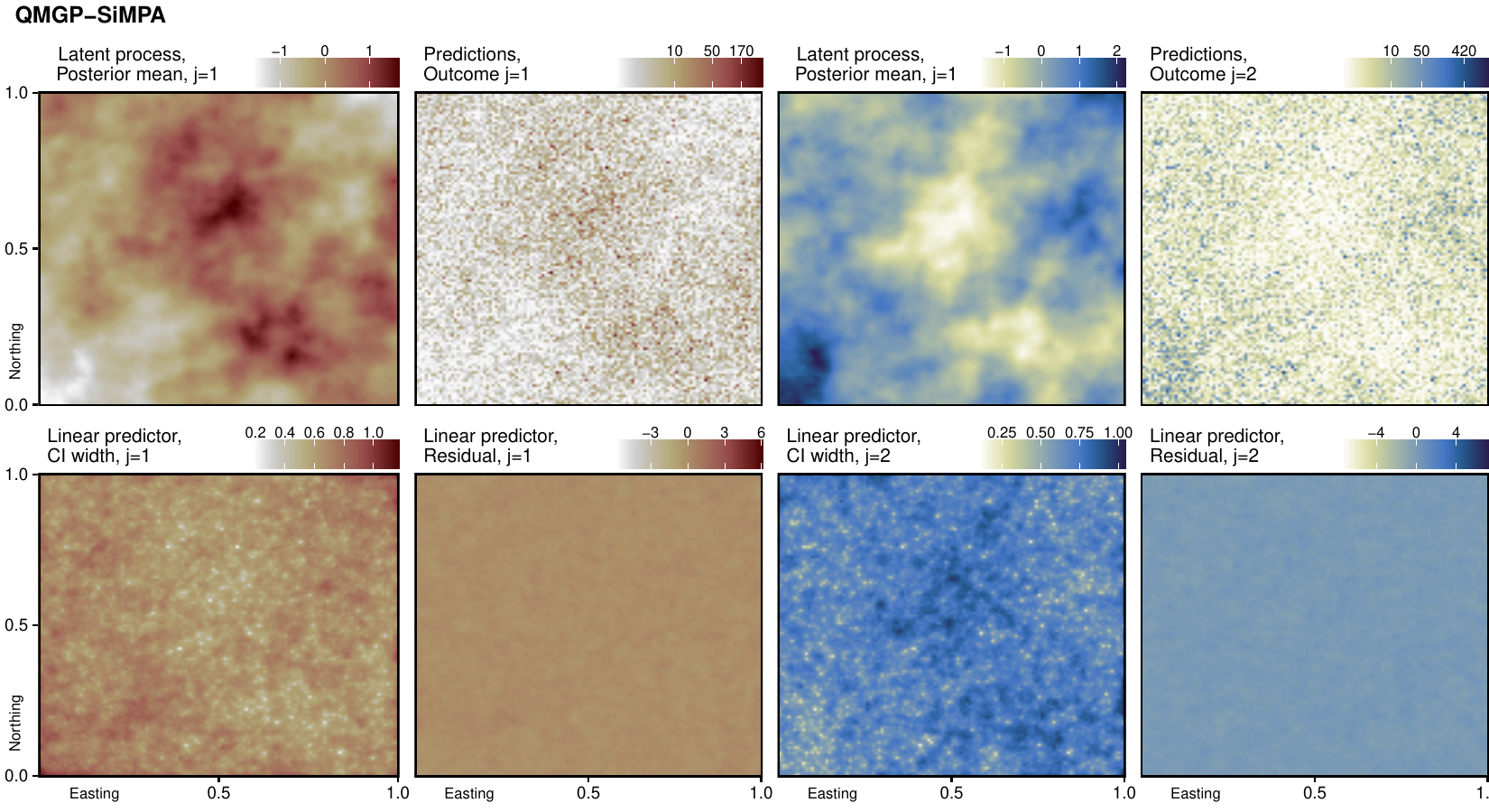}
    \caption{Output from fitting a coregionalized QMGP via SiMPA to simulated data in the latent NNGP scenario. Top row: estimated posterior mean of the spatial latent process and predictions for both outcomes. Bottom row: width of posterior credible intervals about log-intensity, and residual log-intensity.}
    \label{fig:mvpoisson_simpa}
\end{figure}
We investigate the comparative performance of several coregionalized QMGP variants computed via MALA, SM-MALA, SiMPA and NUTS. We also consider latent multivariate Student-t processes (\citealt{mvstudentreg, studenttprocess}) using an alternative cross-covariance specification based on \cite{apanasovich_genton2010}---in short ``AG10''---and previously used in \cite{meshedgp}, which we also implement in the meshed Gaussian case. We detail the specifics of spatial meshing and gradient-based sampling for Student-t processes in Section \ref{appx:studentt}. To the best of our knowledge, ours is the first implementation of a scalable Student-t process using DAGs. We also compare with a data transformation method based on NNGPs: for each outcome, we use $y^* = log(1+y)$, then fit NNGP models of the response on each outcome independently. All MCMC-based results are based on chains of length 30,000. 
% response to review
All gradient-based methods share the dual averaging setup for adapting the step size $\varepsilon$ and are thus allowed the same burn-in period. 
Finally, we implement an MCMC-free stochastic partial differential equations method \citep[SPDE;][]{spde} fit via INLA. 
The SiMPA-estimated posterior means for the latent process, predictions across the spatial domain, as well as the width of 95\% CIs about the linear predictors are reported in Figure \ref{fig:mvpoisson_simpa}, where we also highlight that the lack of visible spatial patterns in the linear predictor residuals is evidence of the ability of SiMPA to capture the spatial correlation in the data.

A summary of results from all implemented methods is available in Table \ref{tab:mv_lgcp}, which reports root mean square prediction error (RMSPE) and mean absolute error in prediction (MAEP) when predicting the log-intensity $\eta_{j, \text{test}}$ and the outcomes $y_{j, \text{test}}$, $j=1,2$ on the test set of $5740$ locations, and the empirical coverage of 95\% credible intervals (CI) about the log-intensity, in both scenarios. We observe that SiMPA offers excellent predictive performance and well calibrated credible intervals at a fraction of the compute cost, relative to state-of-the-art posterior sampling methods in this context. In the NNGP scenario, there is a disconnect between the fitted DAG (arising from a QMGP) and the data-generating DAG. This disconnect may explain why QMGP methods implementing the flexible AG10 cross-covariance function perform relatively better than in the GP scenario. Even in the NNGP setting, SiMPA retains excellent performance at a small compute cost.

Because the only difference between SiMPA and YA-MALA is in how the preconditioner is adapted, the relatively poor performance of YA-MALA can be attributed to it requiring a much longer burn-in period in practice. We attribute the poor performance of the implemented NNGP methods to the fact that they are unable to capture cross-variable dependence, as well as their being limited to Gaussian outcomes in R package \texttt{spNNGP}

Figure \ref{fig:mvpoisson_coverage} expands on the analysis of empirical coverage of CIs by reporting the performance of all models at additional quantiles, relative to the oracle coverage, i.e., the empirical coverage of the model in which all unknowns are set to their true value. A value of relative coverage near 1 implies that the empirical coverage of the $Q\%$ CI is close to the coverage of the true data generating model. From Figure \ref{fig:mvpoisson_coverage} we observe that SiMPA outpeforms other methods at this task.

\begin{table}[H]%[ht]
\centering
\resizebox{\textwidth}{!}{%
\begin{tabular}{|c|c|c|r|rr|rrr|r|rr|rrr|r|}
\hline
 &  &  & & \multicolumn{6}{c|}{Unrestricted GP} & \multicolumn{6}{c|}{Nearest-neighbor GP, NN $=10$}\\
 %\cmidrule{5-16}
\cline{5-16}
 &  &  &  & \multicolumn{2}{c|}{ $y_{j, \text{test}}(\bl)$} &  \multicolumn{3}{c|}{ $\eta_{j, \text{test}}(\bl)$} & \multirow{2}{*}{\footnotesize{Time(s)}} & \multicolumn{2}{c|}{ $y_{j, \text{test}}(\bl)$} &  \multicolumn{3}{c|}{ $\eta_{j, \text{test}}(\bl)$} & \multirow{2}{*}{\footnotesize{Time(s)}}\\
\multirow{-3}{*}{\shortstack[c]{Spatial\\model}} & \multirow{-3}{*}{\shortstack[c]{Covariance\\model}} & \multirow{-3}{*}{\shortstack[c]{Compute\\algorithm}} & \multirow{-3}{*}{$j$}& \footnotesize{RMSPE} & \footnotesize{MAEP} & \footnotesize{RMSPE} & \footnotesize{MAEP} & \footnotesize{Covg}$_{95\text{\%}}$  &  & \footnotesize{RMSPE} & \footnotesize{MAEP} & \footnotesize{RMSPE} & \footnotesize{MAEP} & \footnotesize{Covg}$_{95\text{\%}}$ & \\
\hline %\midrule
 &  &  & 1 & 3.01 & 1.39 & 0.32 & 0.25 & 69.48 &  & 2.87 & 1.36 & 0.32 & 0.26 & 69.27 & \\
\multirow{-2}{*}{\raggedright\arraybackslash SPDE} &  & \multirow{-2}{*}{\raggedright\arraybackslash INLA} & 2 & 6.14 & 1.90 & 0.33 & 0.26 & 63.44 & \multirow{-2}{*}{\raggedleft\arraybackslash 333} & 5.79 & 1.85 & 0.32 & 0.25 & 65.52 & \multirow{-2}{*}{\raggedleft\arraybackslash 334}\\
\cline{1-1} \cline{3-16}
 &  &  & 1 & 2.33 & 1.27 & 0.21 & 0.17 & 99.51 &  & 2.26 & 1.22 & 0.21 & 0.17 & 99.44 & \\
 &  & \multirow{-2}{*}{\raggedright\arraybackslash MALA} & 2 & 4.08 & 1.57 & 0.20 & 0.16 & 94.27 & \multirow{-2}{*}{\raggedleft\arraybackslash 90} & 3.77 & 1.53 & \textbf{0.19} & \textbf{0.15} & 93.58 & \multirow{-2}{*}{\raggedleft\arraybackslash 89}\\
 \cline{3-16}
 &  &  & 1 & 7.27 & 2.56 & 0.95 & 0.76 & 5.42 &  & 6.96 & 2.48 & 0.93 & 0.75 & 5.56 & \\
 &  & \multirow{-2}{*}{\raggedright\arraybackslash YA-MALA} & 2 & 18.98 & 4.92 & 1.28 & 1.01 & 3.92 & \multirow{-2}{*}{\raggedleft\arraybackslash 111} & 18.85 & 4.89 & 1.30 & 1.03 & 4.20 & \multirow{-2}{*}{\raggedleft\arraybackslash 108}\\
 \cline{3-16}
 &  &  & 1 & 2.18 & \textbf{1.22} & \textbf{0.17} & \textbf{0.14} & 95.83 &  & 2.16 & \textbf{1.19} & \textbf{0.17} & \textbf{0.14} & 95.21 & \\
 &  & \multirow{-2}{*}{\raggedright\arraybackslash SiMPA} & 2 & 4.03 & 1.57 & \textbf{0.19} & \textbf{0.15} & 94.55 & \multirow{-2}{*}{\raggedleft\arraybackslash 117} & 3.60 & 1.51 & \textbf{0.19} & \textbf{0.15} & \textbf{94.48} & \multirow{-2}{*}{\raggedleft\arraybackslash 116}\\
 \cline{3-16}
 &  &  & 1 & 2.19 & \textbf{1.22} & 0.18 & 0.15 & 96.28 &  & 2.17 & 1.20 & 0.18 & 0.15 & \textbf{94.93} & \\
 &  & \multirow{-2}{*}{\raggedright\arraybackslash RM-MALA} & 2 & 4.20 & \textbf{1.56} & 0.23 & 0.18 & 93.37 & \multirow{-2}{*}{\raggedleft\arraybackslash 187} & 3.97 & 1.56 & 0.23 & 0.18 & 92.57 & \multirow{-2}{*}{\raggedleft\arraybackslash 183}\\
 \cline{3-16}
 &  &  & 1 & 2.18 & 1.23 & \textbf{0.17} & \textbf{0.14} & \textbf{95.73} &  & 2.16 & 1.20 & \textbf{0.17} & \textbf{0.14} & 94.65 & \\
 &  & \multirow{-2}{*}{\raggedright\arraybackslash HMC} & 2 & \textbf{3.96} & \textbf{1.56} & \textbf{0.19} & \textbf{0.15} & \textbf{94.76} & \multirow{-2}{*}{\raggedleft\arraybackslash 246} & 3.54 & 1.51 & \textbf{0.19} & \textbf{0.15} & 93.82 & \multirow{-2}{*}{\raggedleft\arraybackslash 359}\\ 
 \cline{3-16}
 &  &  & 1 & \textbf{2.16} & \textbf{1.22} & 0.18 & \textbf{0.14} & 93.89 &  & 2.21 & 1.20 & 0.18 & \textbf{0.14} & 92.57 & \\
 & \multirow{-14}{*}{\raggedright\arraybackslash LMC} &  & 2 & 4.07 & 1.57 & 0.20 & 0.16 & 90.66 & \multirow{-2}{*}{\raggedleft\arraybackslash 620} & 3.56 & 1.51 & \textbf{0.19} & \textbf{0.15} & 90.97 & \multirow{-2}{*}{\raggedleft\arraybackslash 633}\\
  \cline{2-2} \cline{4-16}
 &  &  & 1 & 2.22 & 1.23 & 0.18 & \textbf{0.14} & 92.64 &  & \textbf{2.13} & \textbf{1.19} & \textbf{0.17} & \textbf{0.14} & 91.39 & \\
\multirow{-14}{*}{\raggedright\arraybackslash QMGP} &  &  & 2 & 4.01 & \textbf{1.56} & 0.20 & 0.16 & 91.60 & \multirow{-2}{*}{\raggedleft\arraybackslash 501} & 3.56 & \textbf{1.50} & \textbf{0.19} & \textbf{0.15} & 92.01 & \multirow{-2}{*}{\raggedleft\arraybackslash 480}\\
\cline{1-1} \cline{4-16}
 &  &  & 1 & 2.19 & 1.23 & 0.18 & \textbf{0.14} & 91.77 &  & 2.16 & 1.20 & 0.18 & \textbf{0.14} & 90.97 & \\
\multirow{-2}{*}{\raggedright\arraybackslash QMTP} & \multirow{-4}{*}{\raggedright\arraybackslash AG10} & \multirow{-6}{*}{\raggedright\arraybackslash NUTS} & 2 & 4.02 & 1.57 & 0.20 & 0.16 & 90.21 & \multirow{-2}{*}{\raggedleft\arraybackslash 857} & \textbf{3.50} & \textbf{1.50} & 0.19 & \textbf{0.15} & 91.28 & \multirow{-2}{*}{\raggedleft\arraybackslash 841}\\
\cline{1-16}
\multirow{2}{*}{NNGP} & \multirow{2}{*}{Exp} & \multirow{2}{*}{\shortstack[c]{Transform \&\\Response} } & 1 & 5.71 & 2.01 & 1.19 & 0.98 & 59.41 &  & 5.52 & 1.96 & 1.19 & 0.99 & 59.10 & \\
 & & & 2 & 15.55 & 3.63 & 1.36 & 1.12 & 58.09 & \multirow{-2}{*}{\raggedleft\arraybackslash 166} & 15.44 & 3.56 & 1.34 & 1.10 & 58.68 & \multirow{-2}{*}{\raggedleft\arraybackslash 171}\\
\hline %\bottomrule
\end{tabular}
}
\caption{Summary of out-of-sample results for all implemented models. Bolded values correspond to best performance.}
    \label{tab:mv_lgcp}
\end{table}

\begin{figure}
    \centering
    \includegraphics[width=.96\textwidth]{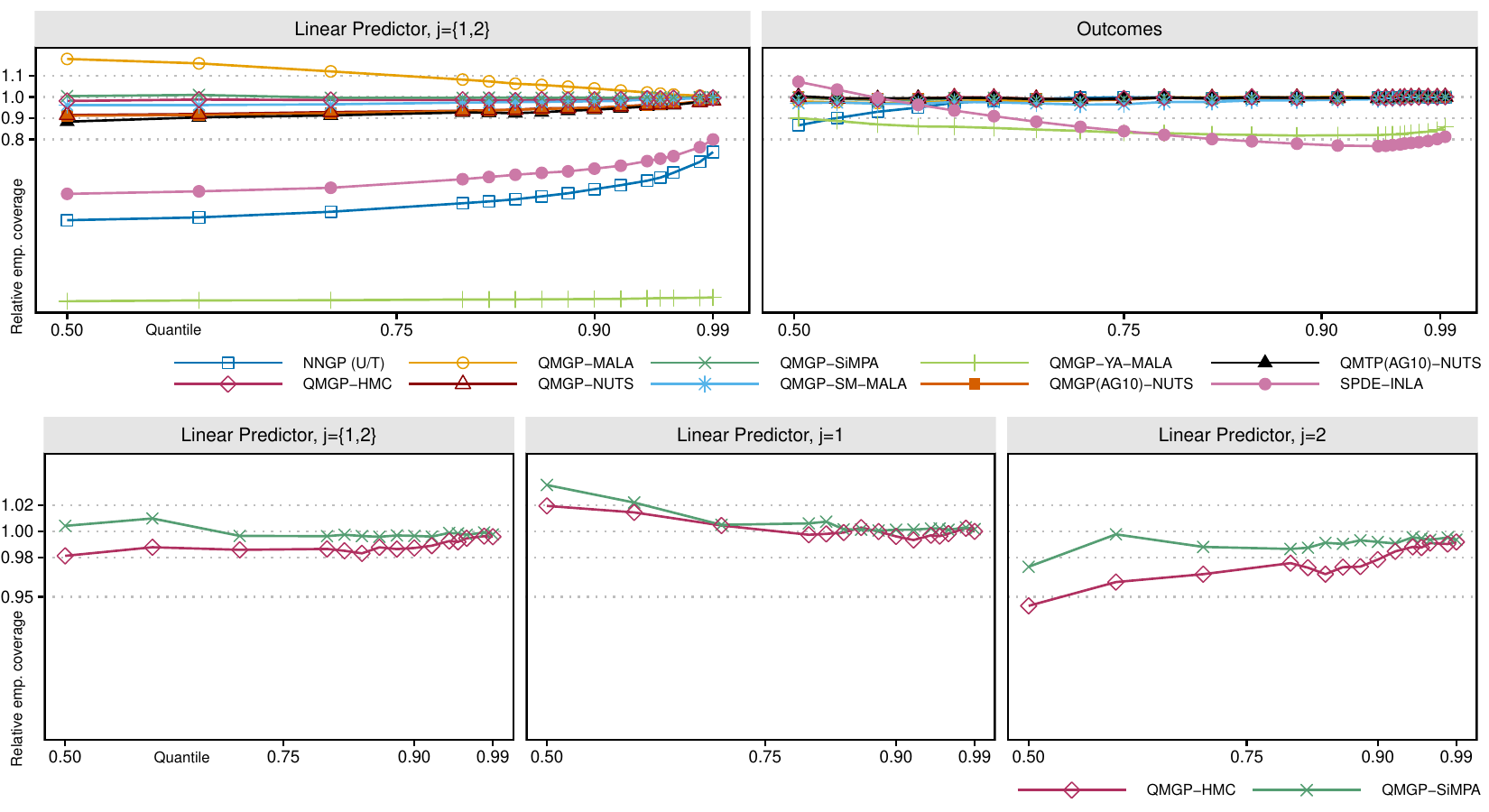}
    \caption{Top row: empirical coverage of uncertainty intervals at different quantiles, relative to the oracle model (values under 1 imply undercoverage of the interval), in the NNGP scenario. Bottom row: detailed comparison of relative coverage of SiMPA and HMC for the linear predictor of each outcome.}
    \label{fig:mvpoisson_coverage}
\end{figure}

\subsection{MODIS data: leaf area and snow cover} \label{sec:modis}
\begin{figure}
\centering
\begin{subfigure}{0.99\textwidth}
    \includegraphics[width=.99\textwidth]{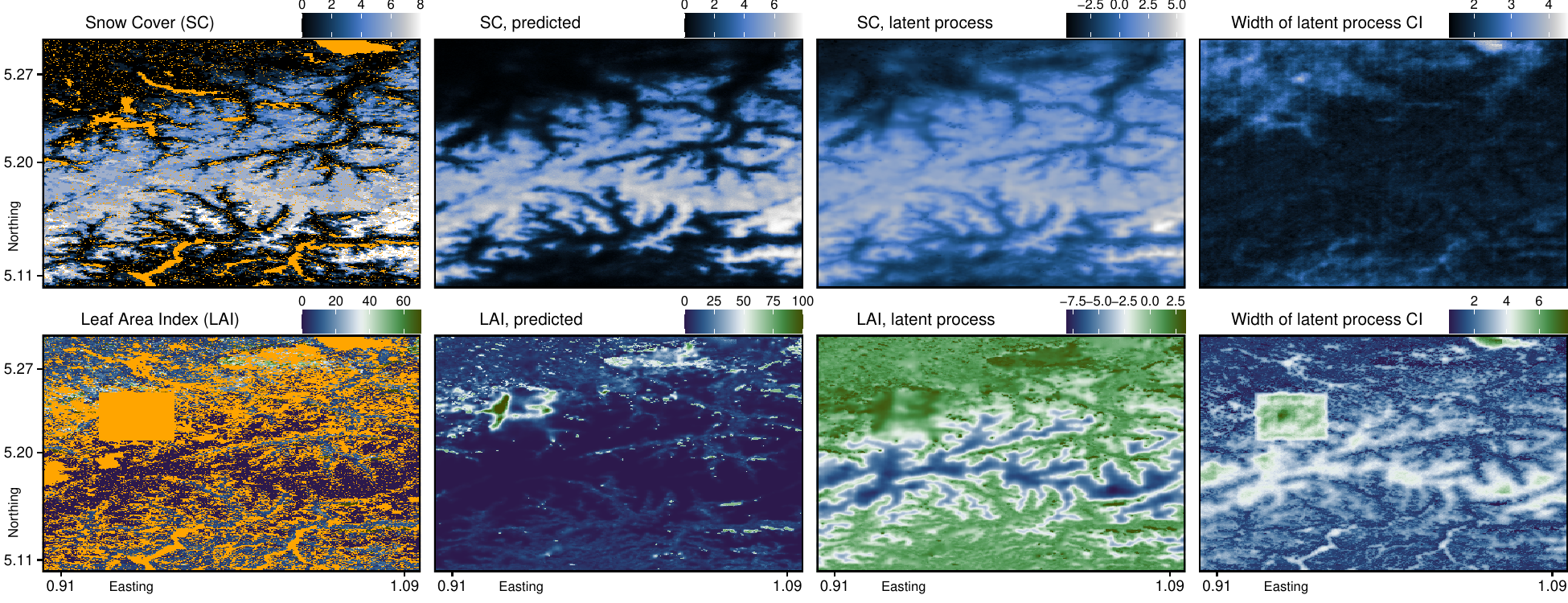}
    \caption{From the left: in-sample data, predictions of outcomes and latent processes, uncertainty quantification from a QMGP-SiMPA model using a Poisson distribution for LAI.}
    \label{fig:sat:simpaout}
\end{subfigure}
\hfill
\begin{subfigure}{0.99\textwidth}
    \includegraphics[width=.99\textwidth]{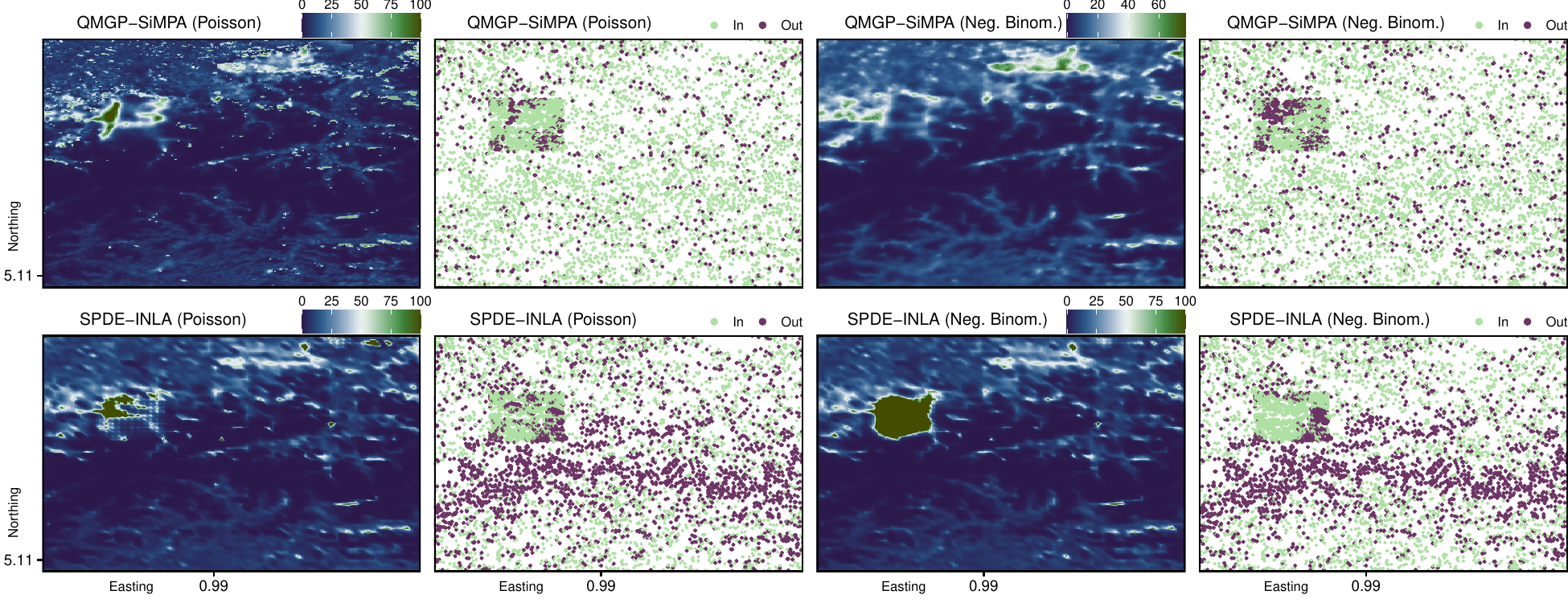}
    \caption{Each raster image reports whole-domain predictions of all tested models, whereas each dark dot in the point clouds represents a domain location at which the 75\% one-sided credible interval does not include the observed data point.}
    \label{fig:sat:output}
\end{subfigure}
\caption{Performance of QMGP-SiMPA and SPDE-INLA in the MODIS data application.}
\label{fig:figures}
\end{figure}
The dynamics of vegetation greenness are important drivers of ecosystem processes; in alpine regions, they are influenced by seasonal snow cover. Predictive models for vegetation greenup and senescence in these settings are crucial for understanding how local biological communities respond to global change \citep{longterm_snow,jonsson2010modis_snow,snow_effects,xieetal2020}. We consider remotely sensed leaf area and snow cover data from the MODerate resolution Imaging Spectroradiometer (MODIS) on the Terra satellite operated by NASA (v.6.1) at 122,500 total locations (a 350$\times$ 350 grid where each cell covers a 0.25km$^2$ area) over a region spanning northern Italy, Switzerland, and Austria, during the 8-day period starting on December 3rd, 2019 (Figure \ref{fig:sat:data}). Leaf area index (LAI; number of equivalent layers of leaves relative to a unit of ground area, available as level 4 product MOD15A2H) is our primary interest and is stored as a positive integer value but is missing or unavailable at 38.2\% of all spatial locations due to cloud cover or poor measurement quality. Snow cover (SC; number of days within an 8-day period during which a location is covered by snow, obtained from level 3 product MOD10A2) is measured with error and missing at 7.3\% of the domain locations. 

We create a test set by introducing missingness in LAI at 10,000 spatial locations, of which 5030 are chosen uniformly at random across the whole domain and 4970 belong to a contiguous rectangular region as displayed on the bottom left subplot of Figure \ref{fig:sat:simpaout}.
We attempt to explain LAI based on SC by fitting (\ref{eq:latent_gaussian_lmc}) on the bivariate outcome $\by(\bl) = (y_{\text{SC}}(\bl), y_{\text{LAI}}(\bl) )^\top$ where we assume a Binomial distribution with $8$ trials and logit link for SC, i.e. $E(y_{\text{SC}}(\bl) \mid \mu(\bl)) = 8\mu(\bl) = 8 (1+\exp\{ - \eta_{\text{SC}}(\bl)\} )^{-1}$, and a Poisson or Negative Binomial distribution for LAI. In both cases, $E(y_{\text{LAI}}(\bl) \mid \eta_{\text{LAI}}(\bl)) = \mu_{\text{LAI}}(\bl) = \exp\{ \eta_{\text{LAI}}(\bl) \}$; for the Poisson model, $Var(y_{\text{LAI}}(\bl) \mid \eta_{\text{LAI}}(\bl)) = \mu_{\text{LAI}}(\bl)$, whereas for the Negative Binomial model $
Var(y_{\text{LAI}}(\bl) \mid \eta_{\text{LAI}}(\bl)) = \mu_{\text{LAI}}(\bl) + \tau  \mu_{\text{LAI}}^2(\bl) $ where $\tau$ is an unknown scale parameter. We fit model (\ref{eq:latent_gaussian_lmc}) using latent coregionalized QMGPs with $k=2$ on a $50\times 50$ axis-parallel domain partition and run SiMPA for 30,000 iterations, of which 10,000 are discarded as burn-in and thinning the remaining ones with a 20:1 ratio, leading to a posterior sample of size 1,000. We compare our approaches in terms of prediction and uncertainty quantification about $y_{\text{LAI}}$ on the test set to a SPDE-INLA approach implemented on a $60\times 60$ mesh which led to similar compute times. As shown in Table \ref{tab:modis_results}, QMGP-SiMPA is competitive with or outperforms the SPDE-INLA method across all measured quantities. Figure \ref{fig:sat:output} reports predictive maps of the tested models (prediction values are censored at 100 for visualization clarity), along with a visualization of 75\% one-sided credible intervals which shows the SPDE-INLA method exhibiting undesirable spatial patterns, unlike QMGP-SiMPA.

\begin{table}[ht]
\centering
\resizebox{\textwidth}{!}{
\begin{tabular}{|r|c|cc|cc|ccc|c|}
  \hline
 \multirow{2}*{Method}  & \multirow{2}*{$F_{\text{LAI}}$} & \multirow{2}*{RMSPE} & \multirow{2}*{MedAE} & \multicolumn{2}{c|}{CRPS} & \multirow{2}*{$CI_{75}$} & \multirow{2}*{$CI_{95}$} & \multirow{2}*{$CI_{99}$} & Time \\
 &  &  &  & (mean) & (median) &  &  &  & (minutes) \\ 
  \hline
\multirow{2}*{QMGP-SiMPA} & Poisson & 16.543 & 1.322 & 3.916 & 1.199 & 0.867 & 0.974 & 0.989 & 25.4 \\
 & Neg. Binom. & 11.726 & 2.155 & 4.462 & 2.241 & 0.809 & 0.980 & 0.994 & 32 \\ 
 \hline
\multirow{2}*{SPDE-INLA} & Poisson & 27.839 & 2.154 & 4.695 & 1.214 & 0.835 & 0.938 & 0.961 & 25.8 \\ 
 & Neg. Binom. & 27019.470 & 2.444 & 54.986 & 1.720 & 0.875 & 0.975 & 0.987 & 86.5 \\ 
   \hline
\end{tabular}
}
\caption{Root mean square error (RMSPE), median absolute error (MedAE), continuous ranked probability score (CRPS), and empirical coverage of one-sided intervals ($CI_q$), over the out-of-sample test set of 6,998 locations.} \label{tab:modis_results} 
\end{table}

\section{Applications: spatial community ecology} \label{sec:jsdm}
Ecologists seek to jointly model the spatial occurrence of multiple species, while inferring the impact of phylogeny and environmental covariates (see, e.g., \citealt{DorazioRoyle05, Doseretall2022}).
In order to realistically model such a scenario, we consider cases in which a relatively large number of georeferenced outcomes is observed, with the goal of predicting their realization at unobserved locations and estimating their latent correlation structure after accounting for spatial and/or temporal variability. Presence/absence information for different species is encoded as a multivariate binary outcome. Our model for multivariate binary outcomes lets $F_j(y_j(\bl); \eta_j(\bl)) = Bern( \mu_j(\bl))$ where $\mu_j(\bl) = (1+\exp\{ -\eta_j(\bl) \})^{-1}$ and $v_h(\cdot) \sim QMGP(\bzero, \brho_h(\cdot, \cdot)), h=1, \dots, k$ in model (\ref{eq:latent_gaussian_lmc}), leading to coregionalized $k$-factor QMGPs which we fit via several Langevin methods, all of which use domain partitioning with blocks of size $\approx 36$ and independent standard Gaussian priors on the lower-triangular elements of the factor loadings $\bLambda$, unless otherwise noted. 

We compare QMGP methods fit via our proposed Langevin algorithms to joint species distribution models (JSDM) implemented in R package \texttt{Hmsc} \citep{hmsc_package}, a popular software package for community ecology. \texttt{Hmsc} uses a probit link for binary outcomes, i.e. $\mu_j(\bl) = \Phi(\eta_j(\bl))$ where $\Phi(\cdot)$ is the Gaussian distribution function; then, non-spatial JSDMs are implemented by letting $v_h(\bl) \sim N(0, 1)$ independently for all $\bl$ and $h=1,\dots, K$, whereas NNGP-based JSDMs assume  $v_h(\cdot) \sim NNGP(\bzero, \brho_h(\cdot, \cdot)), h=1, \dots, k$. We set the number of neighbors as $m=20$ in the NNGP specification. \texttt{Hmsc} assumes a cumulative shrinkage prior on the factor loadings \citep{bhattacharya_dunson}, which we set up with minimal shrinking ($a_1=2, a_2=2$) unless otherwise noted.

Section \ref{sec:spammix} in the Supplementary Material compares our methods with alternative posterior sampling algorithm in fitting a multi-species N-mixture model for multi-species abundance data in community ecology.

\subsection{Synthetic occupancy data}
We generate 30 datasets by sampling $q=10$ binary outcomes at $n=900$ locations scattered uniformly in the domain $[0,1]^2$: after sampling $k=3$ independent univariate GPs $v_j(\cdot) \sim GP(\bzero, C_{\varphi_j})$ where $C_{\varphi_j}(\bl, \bl') = \exp\{ - \varphi_j \| \bl - \bl' \| \}$ is the exponential covariance function with decay parameter $\varphi_j$, we construct a $q$-variate GP via coregionalization by letting $\bw(\bl) = \bLambda \bv(\bl)$ where $\bLambda$ is a $q \times k$ lower-triangular matrix. We then sample the binary outcomes using a probit link, i.e. $y_j(\bl) \sim Bern( \mu_j(\bl) )$ where $\mu_j(\bl) = \Phi(\bx(\bl)^\top \bbeta_j + w_j(\bl))$ for each $j=1, \dots, q$ and where $\bx(\bl)$ is a column vector of $p=2$ covariates. For each of the 30 datasets, we randomly set $\varphi_j \sim U(1/2, 10), j=1, \dots, k$, $\bLambda_{jj} \sim U(3/2, 2)$ for $j=1, \dots, k$, $\bLambda_{ij} \sim U(-2, 2)$ for $i<j$, and $\bbeta_j \sim N(\bzero, I_2/5)$. These choices lead to a wide range of latent pairwise correlations induced on the outcomes via $\bw(\cdot)$: letting $\bOmega = (\omega_{ij})_{i,j=1,\dots, q} = \bLambda \bLambda^\top$ represent the cross-covariance at zero spatial distance, we find the cross-correlations as $\bOmega_{\text{corr}} = \text{diag}(\omega_{jj}^{-1/2}) \bOmega  \text{diag}(\omega_{jj}^{-1/2})$. We realistically model long-range spatial dependence by choosing small values for $\varphi_j$, $j=1,\dots,k$. Lastly, we create a test set using 20\% of the locations for each outcome (missing data locations differ for each outcome).

\begin{figure}
    \centering
    \includegraphics[width=0.99\textwidth]{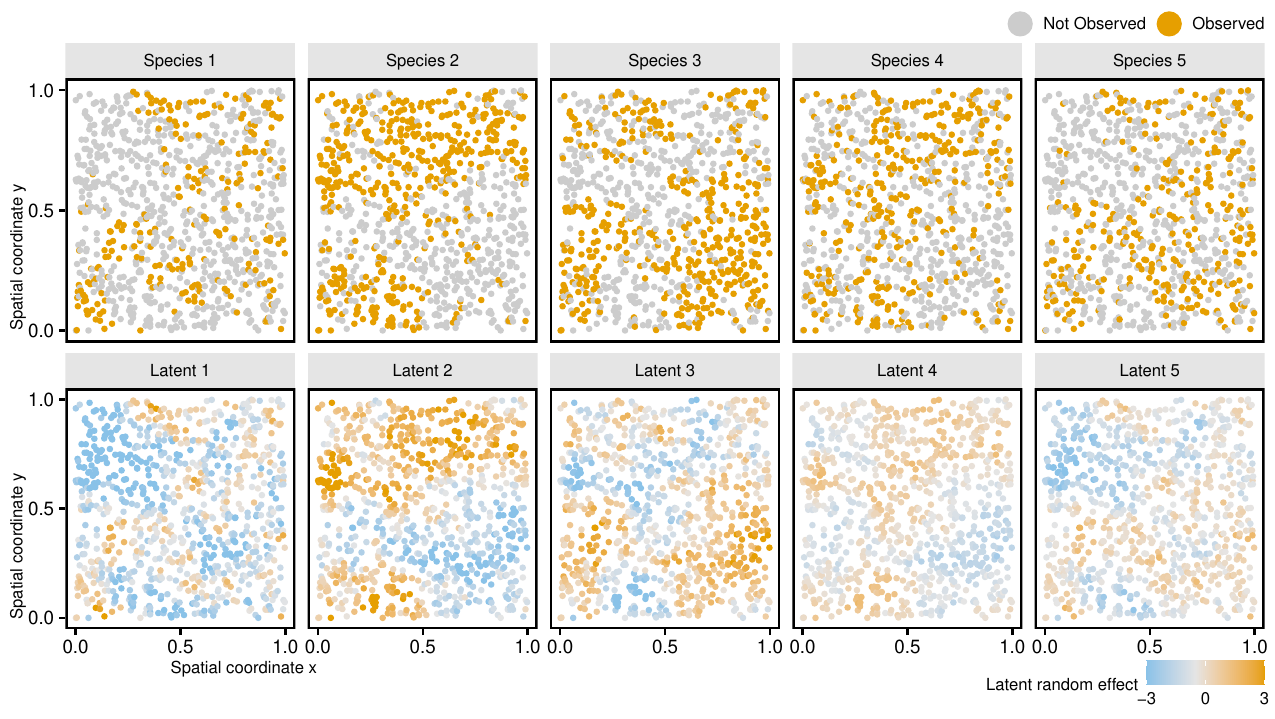}
    \caption{Synthetic dichotomous occurrence data (top row), and the spatial latent effects used to generate them (bottom row). Here, we show 5 (of 10) outcomes in 1 (of 30) simulated datasets.}
    \label{fig:species_synthetic_data}
\end{figure}

\begin{figure}
    \centering
    \includegraphics[width=0.98\textwidth]{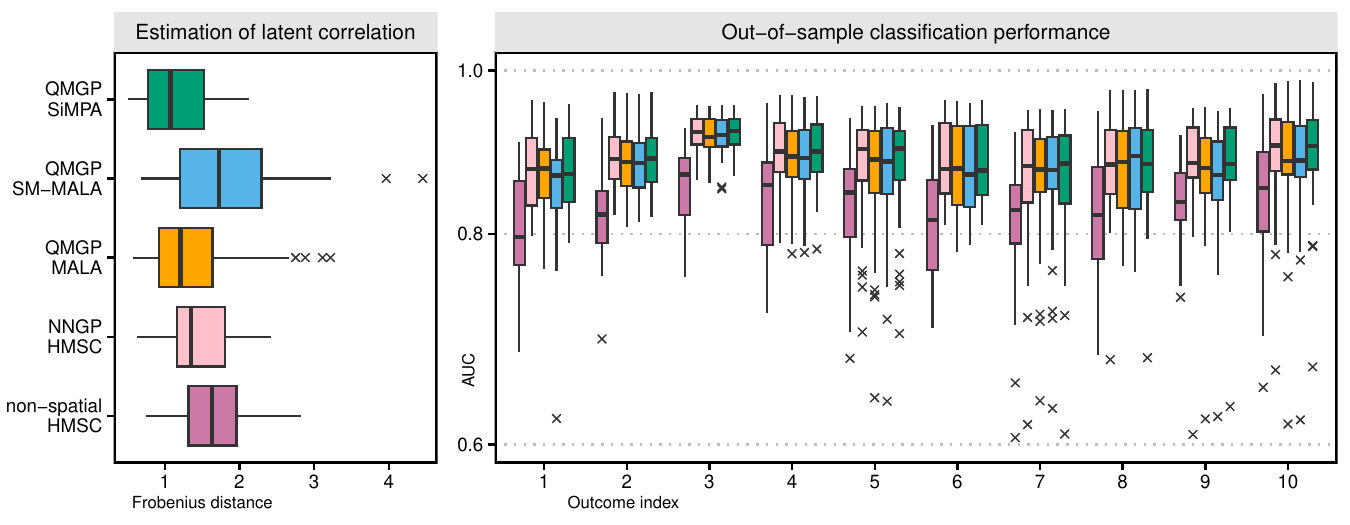}
    \caption{Box-plot summaries of estimation and classification performance over 30 datasets. Left: $\| \hat{\Omega}_{\text{corr}} - \Omega_{\text{corr}} \|_F$ for the competing methods. Right: AUC for each outcome.}
    \label{fig:species_synthetic}
\end{figure}

We use the setup of QMGPs and \texttt{Hmsc} outlined above, noting that the link function used to generate the data is correctly specified for \texttt{Hmsc} but not for our models based on QMGP due to our current software implementation in R package \texttt{meshed}. MCMC for all methods was run for 10,000 iterations, of which the first 5,000 is discarded as burn-in. 
We compare all models based on the out-of-sample classification performance on each of the 10 outcomes as measured via the area under the receiver operating characteristic curve (AUC). Since a primary interest in these settings is to estimate latent correlations across outcomes, we compare models based on $\| \hat{\Omega}_{\text{corr}} - \Omega_{\text{corr}} \|_F$, i.e. the Frobenius distance between the Monte Carlo estimate of cross-correlation and its true value. Therefore, smaller values of $\| \hat{\Omega}_{\text{corr}} - \Omega_{\text{corr}} \|_F$ are desirable. Figure \ref{fig:species_synthetic} shows box-plots summarising the results, whereas Table \ref{tab:hmsc_sim_time} reports averages along with compute times. In these settings, the non-spatial model unsurprisingly performed worst. Langevin methods for the spatial models proposed in this article -- and in particular SiMPA -- lead to improved classification performance, smaller errors in estimating latent correlations, and a 30-fold reduction in compute time, relative to the coregionalized NNGP method implemented via MCMC in \texttt{Hmsc}.

\begin{table}[ht]
\centering
\resizebox{5in}{!}{
\begin{tabular}{|r|ll|lll|}
  \hline
 %\multirow{2}*{Method} & \multicolumn{2}{c|}{\multirow{2}{*}{ \texttt{Hmsc} }}  & \multicolumn{3}{c|}{\melange\ } \\
 Method & \multicolumn{2}{c|}{ \texttt{Hmsc} } & MALA & SM-MALA & SiMPA \\ 
 \hline
 Prior on rand. eff. & {\small non-spatial} & NNGP & \multicolumn{3}{c|}{ QMGP } \\ 
 \hline
  Avg. AUC & 0.827 & \textbf{0.885} & 0.882 & 0.874 & \textbf{0.885} \\
  Min. AUC & 0.573 & 0.608 & 0.392 & 0.530 & \textbf{0.609} \\
  Max. AUC & 0.969 & 0.983 & 0.986 & \textbf{0.987} & 0.984 \\
  \hline
  $\| \hat{\Omega}_{\text{corr}} - \Omega_{\text{corr}} \|_F$ & 1.66 & 1.43 & 1.46 & 1.91 & \textbf{1.14} \\
  \hline
  Avg. time (minutes) &  5.15 & 17.4 & \textbf{0.44} & 0.73 & 0.53  \\
   \hline
\end{tabular}
}
\caption{Performance in classification, estimation, and compute time, over 30 synthetic datasets.} \label{tab:hmsc_sim_time} 
\end{table}

\subsection{North American breeding bird survey data} \label{nabbs}
The North American Breeding Bird Survey dataset contains avian point count data for more than 700 North American bird taxa (species, races, and unidentified species groupings). These data are collected annually during the breeding season, primarily June, along thousands of randomly established roadside survey routes in the United States and Canada.

We consider a dataset of $n=4118$ locations spanning the continental U.S., and $q=27$ bird species. The specific species we consider belong to the \textit{passeriforme} order and are observed at a number of locations which is between 40\% and 60\% of the total number of available locations -- Figure \ref{fig:nabbs:data} shows a subset of the data. We dichotomize the observed counts to obtain presence/absence data. The effective data size is $nq = $ 111,186. We implement Langevin methods using coregionalized QMGPs with $k=2, 4, 6, 8, 10$ spatial factors using exponential correlations with decay $\phi \sim U[0.1, 10]$ \textit{a priori}. We also test the sensitivity to the domain partitioning scheme by testing $8\times 4$ (coarse), $16\times 8$ (medium), and $32 \times 16$ (fine) axis-parallel domain partitioning schemes. Finer partitioning implies more restrictive spatial conditional independence assumptions. In implementing the shrinkage prior of \cite{bhattacharya_dunson}, \texttt{Hmsc} dynamically chooses the number of factors up to a maximum $k_{\text{max}}$: in the non-spatial  \texttt{Hmsc} model, letting $k_{\text{max}}=10$ ultimately leads to 6 or fewer factors being used during MCMC. In the spatial \texttt{Hmsc} models using NNGPs, we set $k_{\text{max}}=2$ or $k_{\text{max}}=5$ to restrict run times. 
Figure \ref{fig:nabbs:results} reports average classification performance and run times. QMGP-MALA scales only linearly with the number of factors, but its performance is strongly negatively impacted by partition size. QMGP-SM-MALA exhibits large improvements in classification performance, however these improvements come at a large run time cost. %Performance of SM-MALA also slightly worsens with a coarser partition due to the increased dimension of the sampled targets. 
%On the contrary, 
QMGP-SiMPA outperforms all other models while providing large time savings relative to SM-MALA and being less sensitive to the choice of partition. A QMGP-SiMPA model on the $32\times 16$ partition with $k=4$ outperforms a spatial NNGP-\texttt{Hmsc} model in classifying the 27 bird species with a reduction in run time of over three orders of magnitude (respectively 4.1 minutes and 70.7 hours). We provide a summary of the efficiency in sampling the elements of $\bOmega_{\text{corr}}$ in Table \ref{tab:nabbs:compare}, where we make comparisons of ESS/s relative to the non-spatial \texttt{Hmsc} model. While efficient estimation of $\bOmega_{\text{corr}}$ remains challenging, QMGP-SiMPA models show marked improvements relative to a state-of-the-art alternative.

\begin{figure}
    \centering
    \includegraphics[width=.99\textwidth]{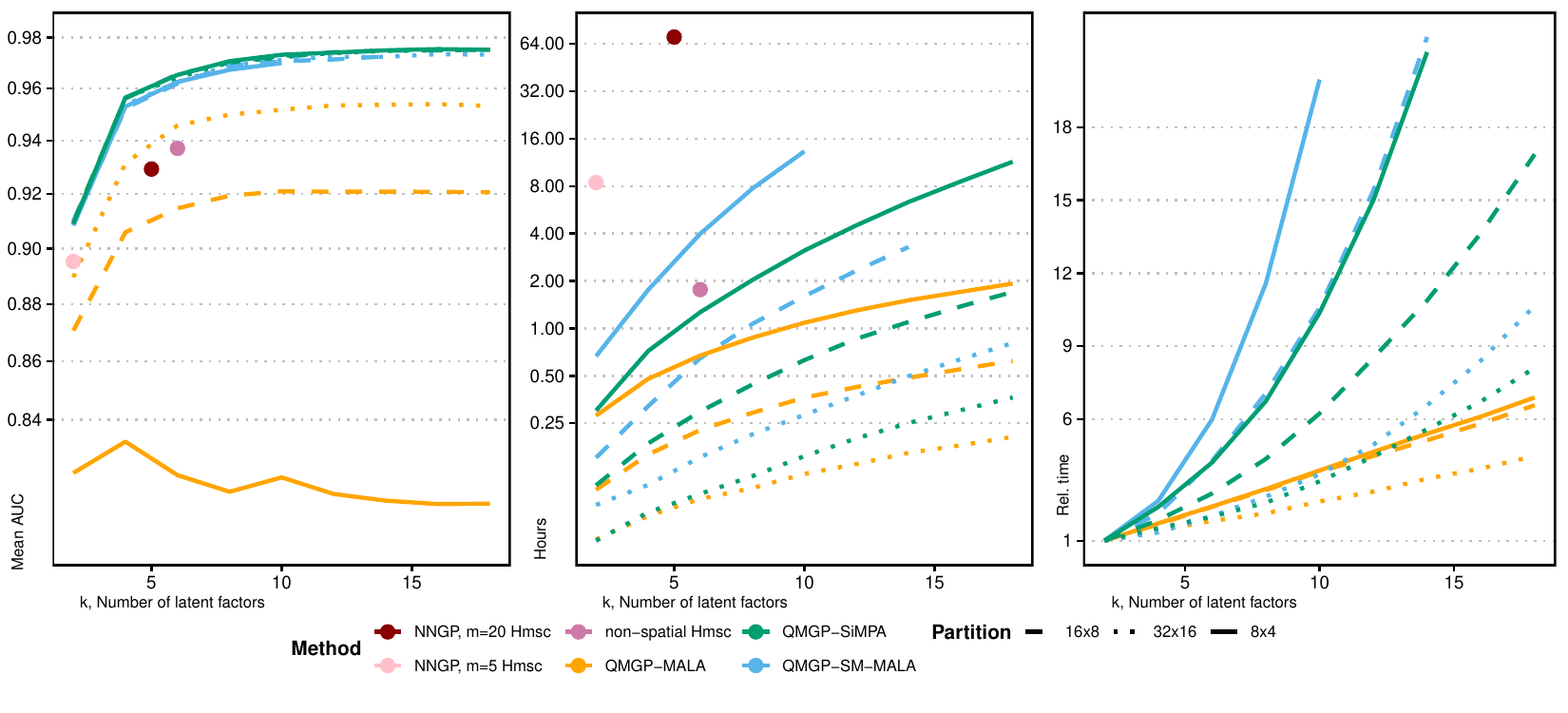}
    \caption{Left: mean AUC across the 27 bird species for different choices of $k$. Center: run times in hours. Right: run time of QMGP models as a proportion to the run time choosing $k=2$.}
    \label{fig:nabbs:results}
\end{figure}

\begin{figure}
    \centering
    \includegraphics[width=.66\textwidth]{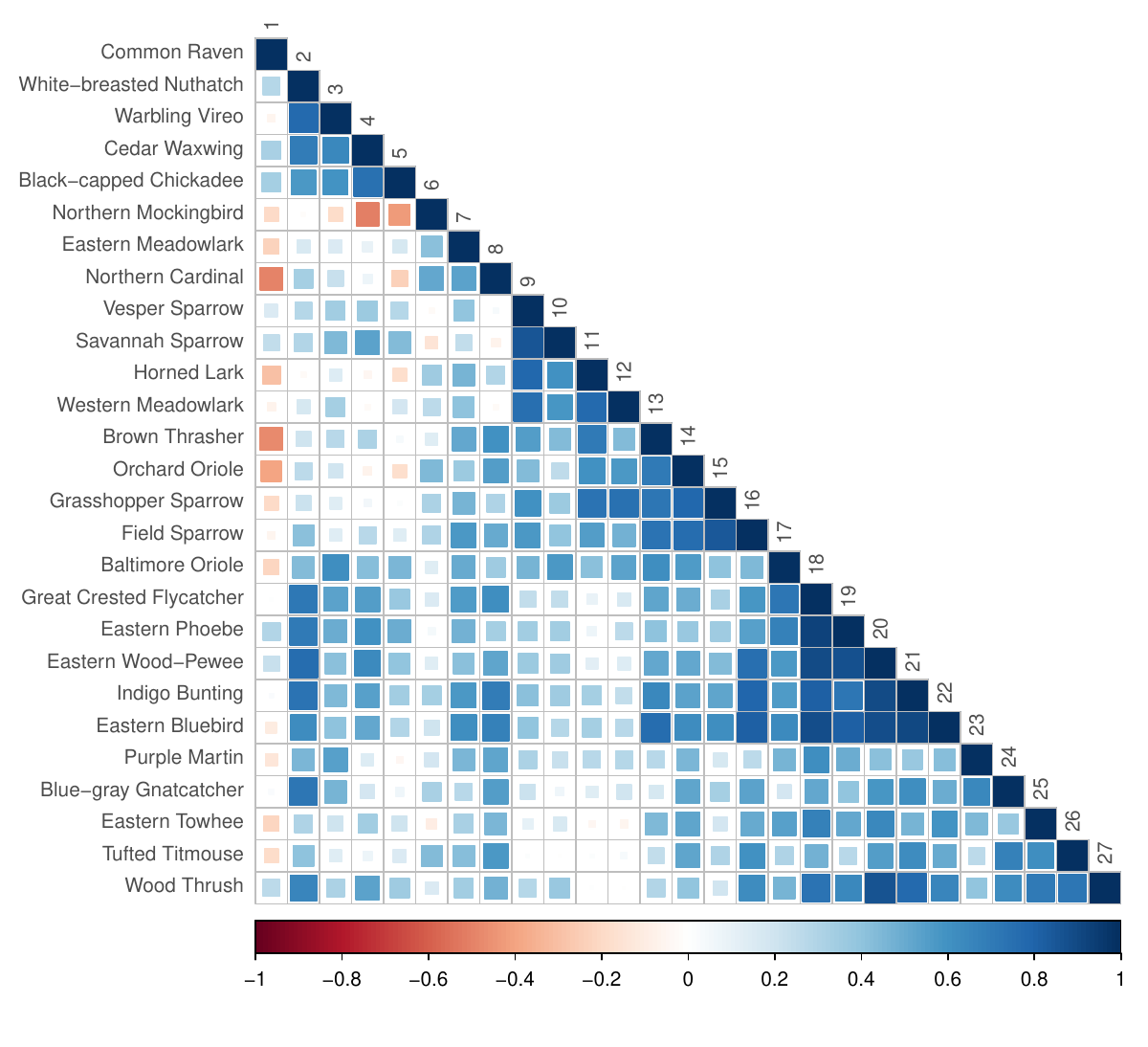}
    \caption{Lower triangular portion of $\bOmega_{\text{corr}}$, the estimated correlation among the 27 bird species.}
    \label{fig:nabbs:corrplot}
\end{figure}

\begin{table}[ht]
\centering
\resizebox{5.5in}{!}{
\begin{tabular}{|r|cc|cccc|}
  \hline
 %\multirow{2}*{Method} & \multicolumn{2}{c|}{\multirow{2}{*}{ \texttt{Hmsc} }}  & \multicolumn{3}{c|}{\melange\ } \\
 Method & \multicolumn{2}{c|}{ \texttt{Hmsc} } & \multicolumn{4}{c|}{ SiMPA } \\ 
 \hline
 Prior & {\small non-spatial} & NNGP & \multicolumn{4}{c|}{ QMGP } \\ 
 \hline
 $k$ & $\leq 10$ & 5 & 4 & 4 & 10 & 10 \\ 
  Setting &  & $m=20$ & $32\times 16$ & $8\times 4$ & $32\times 16$ & $8\times 4$ \\ 
 \hline
  Avg. AUC & 0.9349 & 0.9293 & 0.9565 & 0.9565 & 0.9728 & \textbf{0.9732} \\
  \hline
  Time (minutes) & 87.45 & 4245.02 & \textbf{4.08} & 43.10 & 9.27 & 187.24  \\
   \hline
  \multicolumn{6}{|l}{ ESS/s \footnotesize{for elements of $\bOmega_{\text{corr}}$ (relative to \texttt{Hmsc} non-spatial)} } &   \\
  \hline
  min & \textbf{1} & $10^{-4}$ & 0.57 & 0.02  & 0.05 & 0.01 \\
  median & 1 & 0.012 & \textbf{2.12} & 0.33  & 0.86 & 0.05 \\
  mean & 1 & 0.015 & \textbf{3.69} & 0.45  & 1.20 & 0.07  \\
  max & 1 & 0.102 & \textbf{42.47} & 3.95 & 11.15 & 0.56 \\
  \hline
\end{tabular}
}
\caption{Out-of-sample performance in classification of the 27 bird species, compute time, and efficiency in estimation of $\bOmega_{\text{corr}}$, relative to a non-spatial JDSM model.} \label{tab:nabbs:compare} 
\end{table}

\section{Discussion} \label{sec:discussion}
We have introduced Bayesian hierarchical models based on DAG constructions of latent spatial processes for large scale non-Gaussian multivariate multi-type data which may be misaligned, along with computational tools for adaptive posterior sampling. We illustrated our methods using applications with data sizes in the tens to hundreds of thousands, with compute times ranging from a few seconds to under 30 minutes in a single workstation. The compute time for a single SiMPA iteration for a univariate Poisson outcome observed on gridded coordinates with $n=10^6$ is under 0.2 seconds after burn-in; in other words, our methods enable running MCMC for hundreds of thousands of iterations on massive spatial data under a total time budget of 12 hours.

We have applied our methodologies using practical cross-covariance choices such as models of coregionalization built on independent stationary covariances. However, nonstationary models are desirable in many applied settings. Recent work \citep{bags} highlights that DAG choice must be made carefully when considering explicit models of nonstationary, as spatial process models based on sparse DAGs induce nonstationarities even when using stationary covariances. 
Our work in this article will enable new research into nonstationary models of large scale non-Gaussian data.
Furthermore, our methods can be applied for posterior sampling of Bayesian hierarchies based on more complex conditional independence models of multivariate dependence \citep{deyetal20}. 

In our work, we have assumed a common DAG and partitioning for all spatial variables. In some settings, these assumptions may lead to inflexibility in modeling variables with fundamentally different dependence structures and/or spatial domain constraints (see, e.g., \citealt{boragp}). In the multivariate setting, one potentially useful direction towards building a highly flexible class of models is to infer different DAGs for different factors within a spatial factor model by extending the methods of \cite{bags}. Understanding how to generally build flexible and scalable spatial factor models using different DAGs and accounting for unequal domain constraints is an interesting direction for future research. 

Our methodologies rely on the ability to embed the assumed spatial DAG within the larger Bayesian hierarchy and lead to drastic reductions in wall clock time compared to models based on unrestricted GPs. Nevertheless, high posterior correlations of high dimensional model parameters may still negatively impact overall sampling efficiency in certain cases. Motivated by recent progress in improving sampling efficiency of multivariate Gaussian models \citep{grips}, future research will consider generalized strategies for improving MCMC performance in spatial factor models of highly multivariate non-Gaussian data. Finally, optimizing DAG choice for MCMC performance is another interesting path, and recent work on the theory of Bayesian computation for hierarchical models \citep{zanellaroberts21} might motivate further development for spatial process models based on DAGs.

\if0\blind
{
	\subsection*{Acknowledgements}
The authors have received funding from the European Research Council (ERC) under the European Union’s Horizon 2020 research and innovation programme (grant agreement No 856506), and grants R01ES028804 and R01ES035625 of the United States National Institutes of Health (NIH).
} \fi

\newpage
\vspace{2cm}
\begin{center}
    {\Huge \textsc{Supplementary Material}}
\end{center}

\appendix

\section{Spatial meshing with projections} \label{appx:meshing_not_predictive}
The customary setup of a DAG-based model based on spatial meshing is to let $\calS \cap \calT \approx \calT$ as the resulting large overlap between knots and observed locations avoids sampling at non-reference locations. However, it is often desirable to allow flexible choices of $\calS$; for example, there are some computational advantages when $\calS$ is a grid and $\calT$ are irregularly spaced, or when the data are observed with particular patterns \citep{grips}. 
In order to let $\calS$ be more flexibly determined while also avoiding sampling $\bw(\bl)$ at non-reference locations, we introduce a linear projection operator $\bH(\bl)$ of dimension $q\times qn_{[\bl]}$ and where $n_{[\bl]}$ is the number of locations in $[\bl] \subset \calS$; after denoting $\tilde{\bw}(\bl) = \bH(\bl) \bw_{[\bl]}$, we assume that if $\bl\in \calS$ then $\bH(\bl)$ is such that $\tilde{\bw}(\bl) = \bw(\bl)$. Then, we build the outcome model as
\begin{equation} \label{eq:gridmeshed_hierarchy}
    \begin{aligned}
    y_j(\bl) \given \eta_j(\bl), \gamma_j &\sim F_j(\eta_j(\bl), \gamma_j), \qquad \eta_j(\bl) = \bx_j(\bl)^\top \bbeta_j + \tilde{w}_j(\bl), \\
    \bw(\cdot) &\sim \Pi_{\calG}
    \end{aligned}
\end{equation}
where we have replaced $\bw(\bl)$ with $\tilde{\bw}(\bl)$. Setting $\bH(\bl)$ such that $\tilde{\bw}(\bl) = E[ \bw(\bl) \given \bw_{[\bl]}]$ leads to an interpretation of (\ref{eq:gridmeshed_hierarchy}) as a ``local'' predictive process \citep{gp_predictive_process}. 
The posterior distribution for this model is:
\begin{equation} \label{eq:gridmeshedposterior}
\begin{aligned}
    \pi( \{ \bbeta_j, \gamma_j \}_{j=1}^q, \bw_{\calS}, \btheta \given \by_{\calT})  &\propto 
    \pi(\btheta) \pi_{\calG}(\bw_\calS \given \btheta) \prod_{j=1}^q \pi(\bbeta_j, \gamma_j) \prod_{\bl \in \calT_j}  dF_j(y_j(\bl) \given \tilde{w}_j(\bl), \bbeta_j(\bl), \gamma_j).
\end{aligned}
\end{equation}

In this scenario, omitting the residual term $\be(\bl) = \bw(\bl) - \tilde{\bw}(\bl)$ from (\ref{eq:gridmeshed_hierarchy}) leads to advantages in sampling, but possible oversmoothing of the latent spatial surface due to the fact that $\text{var}[\tilde{\bw}(\bl)] < \text{var}[\bw(\bl)]$. In the conditionally conjugate Gaussian setting, such biases can be partly corrected \citep{gp_pp_biasadj, grips}. Certain ad-hoc solutions may be available by allowing spatial variation of $\gamma_j$, i.e. replacing it with $\gamma_j(\bl)$. However, we may choose to ignore the residual term because (1) it is common to assume smoother surfaces with non-Gaussian data, (2) we can choose $\calS$ to be very large, reducing possible biases, (3) we can revert to model (\ref{eq:meshed_hierarchy}) by setting $\calS = \calT$. Posterior sampling for (\ref{eq:gridmeshed_hierarchy}) proceeds via Algorithm~\ref{algorithm:gridmeshed_posterior}.

\begin{algorithm}[H]
{ \small
  \caption{Posterior sampling of model (\ref{eq:gridmeshed_hierarchy}).}\label{algorithm:gridmeshed_posterior}
  \begin{algorithmic}[1]
  \Statex Initialize $\bbeta_j^{(0)}$ and $\gamma_j^{(0)}$ for $j=1, \dots, q$, $\bw_{\calS}^{(0)}$, and $\btheta^{(0)}$
  \Statex \textbf{for} $t \in \{1, \dots, T^*, T^* + 1, \dots, T^* + T\}$ \textbf{do} \Comment{{\footnotesize MCMC loop}}
    \foritem for $j=1, \dots, q$, sample $\bbeta_j^{(t)} \given \by_\calT, \tilde{\bw}_{\calT}^{(t-1)}, \gamma_j^{(t-1)}$
    \foritem for $j=1, \dots, q$, sample $\gamma_j^{(t)} \given \by_\calT, \tilde{\bw}_{\calT}^{(t-1)}, \bbeta_j^{(t)}$
    \foritem sample $\btheta^{(t)} \given \by_{\calT}, \bw^{(t-1)}_{\calS}, \{ \bbeta_j^{(t)}, \gamma_j^{(t)} \}_{j=1}^q$
    \foritem for $i=1, \dots, M$, sample $\bw_{i}^{(t)} \given \bw_{\mb(i)}^{(t)}, \by_i, \btheta^{(t)}, \{ \bbeta_j^{(t)}, \gamma_j^{(t)} \}_{j=1}^q$ \Comment{{\footnotesize reference sampling}}
  \Statex \textbf{end for}
  \Statex Assuming convergence has been attained after $T^*$ iterations:
\Statex discard $\{ \bbeta_j^{(t)}, \gamma_j^{(t)} \}_{j=1}^q, \bw_{\calS}^{(t)}, \btheta^{(t)}$ for $t = 1, \dots, T^*$
\Statex \textbf{Output:} Correlated sample of size $T$ with density \[ \{ \bbeta_j^{(t)}, \gamma_j^{(t)} \}_{j=1}^q, \bw_{\calS}^{(t)},  \btheta^{(t)} \sim \pi(\{ \bbeta_j, \gamma_j \}_{j=1}^q, \bw_{\calS}^{(t)},  \btheta \mid \by_{\calT}).\]
\end{algorithmic} }
\end{algorithm}

\begin{figure}  
    \centering
\resizebox{.95\columnwidth}{!}{%
\begin{tikzpicture}
		\node [latent] (6) at (-9, 1) {\Large $\bw_{11}$};
		\node [latent] (9) at (-7.5, 2.75) {$\tilde{\bw}(\bl)$};
		\node [obs] (13) at (-6, 3.75) {$y_1(\bl)$};
		\node  (91) at (-6, 3.1) {$\vdots$};
		\node [obs] (14) at (-6, 2.25) {$y_q(\bl)$};
		
		\node [latent] (0) at (-9, 5) {\Large $\bw_{12}$};
		\node [latent] (10) at (-7.5, 6.75) {$\tilde{\bw}(\bl)$};
		\node [obs] (15) at (-6, 7.75) {$y_1(\bl)$};
		\node  (93) at (-6, 7.1) {$\vdots$};
		\node [obs] (16) at (-6, 6.25) {$y_q(\bl)$};
		
		\node [latent] (2) at (-1, 5) {{\Large $\bw_{22}$}};
		\node [latent] (11) at (0.5, 6.75) {$\tilde{\bw}(\bl)$};
		\node [obs] (17) at (2, 7.75) {$y_1(\bl)$};
		\node  (94) at (2, 7.1) {$\vdots$};
		\node [obs] (18) at (2, 6.25) {$y_q(\bl)$};
		
		\node [latent] (8) at (-1, 1) {\Large $\bw_{21}$};
		\node [latent] (12) at (0.5, 2.75) {$\tilde{\bw}(\bl)$};
		\node [obs] (19) at (2, 3.75) {$y_1(\bl)$};
		\node  (92) at (2, 3.1) {$\vdots$};
		\node [obs] (20) at (2, 2.25) {$y_q(\bl)$};
		
		\node [latent] (29) at (-9, 9) {\Large $\bw_{13}$};
		\node  (32) at (-7, 10.25) {$\qquad \cdots \qquad $};
		
		\node [latent] (30) at (-1, 9) {\Large $\bw_{23}$};
		\node  (33) at (1, 10.25) {$\qquad \cdots \qquad $};
		
		\node [latent] (31) at (7, 9) {\Large $\bw_{33}$};
		\node  (34) at (9, 10.25) {$\qquad \cdots \qquad $};
		
		\node [latent] (22) at (7, 5) {\Large $\bw_{32}$};
		\node  (35) at (9, 6.25) {$\qquad \cdots \qquad $};
		
		\node [latent] (21) at (7, 1) {\Large $\bw_{31}$};
		\node  (36) at (9, 2.25) {$\qquad \cdots \qquad $};
		
		\node  (37) at (-1, 11.5) {{\Large \rotatebox{90}{$\cdots$}}};
		\node  (38) at (12, 5) {{\Large $\cdots$ }};
		
		\edge {6} {9};
		\edge {9} {13};
		\edge {9} {14};
		\edge {6} {8};
		\edge {8} {2};
		\edge {6} {0};
		\edge {0} {2};
		\edge {0} {10};
		\edge {10} {16};
		\edge {10} {15};
		\edge {2} {11};
		\edge {11} {18};
		\edge {11} {17};
		\edge {8} {12};
		\edge {12} {19};
		\edge {12} {20};
		\edge {0} {29};
		\edge {29} {30};
		\edge {2} {30};
		\edge {8} {21};
		\edge {2} {22};
		\edge {21} {22};
		\edge {22} {31};
		\edge {30} {31};
		\edge {29} {32};
		\edge {30} {33};
		\edge {31} {34};
		\edge {22} {35};
		\edge {21} {36};
		
		\plate {data1} {(9)(13)(14)} {$\bl \in \calT\cap \calU_{11}$} ;
		\plate {data2} {(10)(15)(16)} {$\bl \in \calT\cap \calU_{12}$} ;
		\plate {data4} {(11)(17)(18)} {$\bl \in \calT\cap \calU_{22}$} ;
		\plate {data5} {(12)(19)(20)} {$\bl \in \calT\cap \calU_{21}$} ;
		
		\plate {dataplus1} {(32)} {$\bl \in \calT\cap \calU_{13}$} ;
		\plate {dataplus2} {(33)} {$\bl \in \calT\cap \calU_{23}$} ;
		\plate {dataplus3} {(34)} {$\bl \in \calT\cap \calU_{33}$} ;
		\plate {dataplus4} {(35)} {$\bl \in \calT\cap \calU_{32}$} ;
		\plate {dataplus5} {(36)} {$\bl \in \calT\cap \calU_{31}$} ;
\end{tikzpicture}
}

    \caption{Directed acyclic graph representing a special case of model (\ref{eq:gridmeshed_hierarchy}), for locations at which at least one outcome is observed. For simplicity, we consider $\calS \cap \calT = \emptyset$ and omit the directed edges from $(\bbeta_j, \gamma_j)$ to each $y_j(\bl)$. If $y_j(\bl)$ is unobserved and therefore $\bl \notin \calT_j$, the corresponding node is missing.}
    \label{fig:gridqmgpdag}
\end{figure}

\section{Choice of DAG and partition}
\begin{figure}[H]
    \centering
\resizebox{.4\columnwidth}{!}{%
\begin{tikzpicture}
\path (2, 4.3) node (desc) {{\tiny Partitioning scheme for $\Pi_{\btheta, \calG}^{(1)}$ }};
\draw[line width=0.05cm] (0,0) -- (4,0) -- (4,4) -- (0,4) -- (0,0);
\draw[ dash pattern=on .2cm off .1cm] (0,2) -- (4,2);
\draw[ dash pattern=on .2cm off .1cm] (2,0) -- (2,4);
\path (1,1) node (w1) {$\bw_1$};
\path (3,1) node (w2) {$\bw_2$};
\path (1,3) node (w3) {$\bw_3$};
\path (3,3) node (w4) {$\bw_4$};
\path (-.3, 0) node (D) {$\calD$};
\end{tikzpicture}
}
\resizebox{.4\columnwidth}{!}{%
\begin{tikzpicture}
\path (2, 4.3) node (desc) {{\tiny Partitioning scheme for $\Pi_{\btheta, \calG}^{(2)}$ }};
\draw[line width=0.05cm] (0,0) -- (4,0) -- (4,4) -- (0,4) -- (0,0);
\draw[ dash pattern=on .2cm off .1cm] (0,2) -- (4,2);
\draw[ dash pattern=on .2cm off .1cm] (2,0) -- (2,4);
\draw[ dash pattern=on .05cm off .05cm] (3,0) -- (3,4);
\path (1,1) node (w1) {$\bw_1$};
\path (1,3) node (w3) {$\bw_3$};
\path (2.5,1) node (w41) {$\bw_{2,1}$};
\path (2.5,3) node (w21) {$\bw_{4,1}$};
\path (3.5,1) node (w22) {$\bw_{2,2}$};
\path (3.5,3) node (w42) {$\bw_{4,2}$};
\path (-.3, 0) node (D) {$\calD$};
\end{tikzpicture}
}
\caption{Illustration of the two partitioning schemes. On the right, we juxtapose the second partitioning scheme to clarify the changes relative to the scenario on the left.}
\label{fig:partitionchoice}
\end{figure}
Spatially meshed models on the same partition of $\calS$ can be compared in terms of the sparsity of $\calG$. If edges are added to a sparse DAG $\calG_1$ to obtain $\calG_2$, the child process $\Pi_{\calG_2}$ is closer to the parent process $\Pi_{\btheta}$ (in a Kullback-Leibler (KL) sense) relative to $\Pi_{\calG_1}$ \citep{meshedgp}. For treed DAGs and recursive domain partitioning, the KL divergence of $\Pi_{\calG}$ from $\Pi$ can be reduced by increasing the block size at the root nodes \citep{spamtrees}.
Here, we analyse the modeling implications different non-nested domain partitions have, while using the same DAG structure to govern dependence between partition regions. This scenario occurs e.g. when constructing a cubic MGP model (QMGP). %In a QMGP, the nodes of the cubic DAG are mapped to domain regions constructed via axis-parallel partitioning. The basic DAG structure is shown in Figure \ref{fig:qmgpdag}. For domains of dimension $d$, each reference node has at most $d$ parents (i.e. the nodes that precede it according to the natural order along each axis), and the patterns in the DAG ensure a coarse coloring exists, leading to massive parallelization when sampling the latent process a posteriori \citep[see][Section 4.2]{meshedgp}.  

We consider two partitions of the x-coordinate axis within a $2\times 2$ axis-parallel partitioning scheme (Figure \ref{fig:partitionchoice}) and construct $\Pi_{\calG}^{(i)}$, $i=1,2$ based on each partitioning scheme. 
According to the first partitioning scheme, $\bw_{\calS}$ (in short, $\bw$) is partitioned as $\bw = \{\bw_1, \bw_2, \bw_3, \bw_4 \}$ whereas with the alternative we have $\bw = \{\bw_1^*, \bw_{2,2}, \bw_3^*, \bw_{4,2} \}$ where $\bw_1^* = \{\bw_1, \bw_{2,1} \} $ and $\bw_3^* = \{ \bw_3, \bw_{4,1} \}$. When analysing the relative KL divergence of these two models from $\Pi$, we see
\begin{align*}
    KL(\pi \| \pi_{\calG}^{(2)}) - KL(\pi \| \pi_{\calG}^{(1)}) &= \int \log \frac{\pi(\bw)}{\pi_{\calG}^{(2)}(\bw)} \pi(\bw) d\bw - \int \log \frac{\pi(\bw)}{\pi_{\calG}^{(1)}(\bw)} \pi(\bw) d\bw \\
    %&= \int \left( \log \pi(\bw) - \log \pi_{\calG}^{(2)}(\bw) \right) \pi(\bw) d\bw - \int \left( \log \pi(\bw)  - \log \pi_{\calG}^{(1)}(\bw) \right) \pi(\bw) d\bw \\
    &= \int \log \pi_{\calG}^{(1)}(\bw) \pi(\bw) d\bw - \int \log \pi_{\calG}^{(2)}(\bw) \pi(\bw) d\bw \\
    &= \int \left( \log \pi_{\calG}^{(1)}(\bw)  - \log \pi_{\calG}^{(2)}(\bw) \right) \pi(\bw) d\bw
\end{align*}
Since we fix the same $\calG$ across partitions, we have 
\begin{align*}
    \pi_{\calG}^{(1)}(\bw_\calS) &= \pi(\bw_1) \pi(\bw_2 \given \bw_1) \pi(\bw_3 \given \bw_1) \pi(\bw_4 \given \bw_2, \bw_3)\\
    &=  \pi(\bw_1) \pi(\bw_{2,1} \given \bw_1) p(\bw_{2,2} \given \bw_{1}, \bw_{2,1}) \pi(\bw_3 \given \bw_1) \pi(\bw_{4,1} \given \bw_{2,1}, \bw_{2,2}, \bw_{3} )\ \cdot \\
    &\qquad \cdot\ \pi(\bw_{4,2} \given \bw_{2,1}, \bw_{2,2}, \bw_{3}, \bw_{4,1})\\
    \pi_{\calG}^{(2)}(\bw_\calS) &= \pi(\bw_1^*) \pi(\bw_{2,2} \given \bw_1^*) \pi(\bw_3^* \given \bw_1^*) \pi(\bw_{4,2} \given \bw_{2,2}, \bw_3^*) \\
    &= \pi(\bw_1) \pi(\bw_{2,1} \given \bw_1) p(\bw_{2,2} \given \bw_{1}, \bw_{2,1}) \pi(\bw_3 \given \bw_1, \bw_{2,1}) \pi(\bw_{4,1} \given \bw_1, \bw_{2,1}, \bw_{3} )\ \cdot \\
    &\qquad \cdot \ \pi(\bw_{4,2} \given \bw_{2,2}, \bw_{3}, \bw_{4,1}),
\end{align*} 
and therefore the sign of $KL(\pi \| \pi_{\calG}^{(2)}) - KL(\pi \| \pi_{\calG}^{(1)})$ depends on 
\begin{align*}
    \log\frac{\pi_{\calG}^{(1)}(\bw)}{\pi_{\calG}^{(2)}(\bw)} &= \log\left(  \frac{\pi(\bw_3 \given \bw_1)}{\pi(\bw_3 \given  \textcolor{teal}{\bw_{2,1}}, \bw_1 )} \frac{\pi(\bw_{4,1} \given \textcolor{teal}{\bw_{2,2}}, \bw_{2,1}, \bw_{3} )}{\pi(\bw_{4,1} \given \textcolor{teal}{\bw_1}, \bw_{2,1}, \bw_{3} )} \frac{\pi(\bw_{4,2} \given \textcolor{teal}{\bw_{2,1}}, \bw_{2,2}, \bw_{3}, \bw_{4,1})}{\pi(\bw_{4,2} \given \bw_{2,2}, \bw_{3}, \bw_{4,1})} \right),
\end{align*}
where we see that the performance of $\Pi_{\calG}^{(1)}$ relative to $\Pi_{\calG}^{(2)}$ in approximating $\Pi$ is undetermined because there is no ordering between the number of edges in $\Pi_{\calG}^{(1)}$ and $\Pi_{\calG}^{(2)}$. Nevertheless, the above discussion remains useful in practice when the reference set $\calS$ is chosen at observed locations. For example, if data are unavailable at $(2,1)$, then $\bw_{2,1}$ has length zero, and one would then choose $\Pi_{\calG}^{(1)}$ over $\Pi_{\calG}^{(2)}$ if uncertainty about $\bw_{4,1}$ is reduced by knowledge of $\bw_{2,2}$ more than it is reduced by knowledge of $\bw_1$.

\section{Gradient-based sampling} \label{appx:melange_proof}
We outline proofs for propositions of Section \ref{sec:melange}. \\
\textbf{Proposition \ref{prop:smmala_is_gibbs}}. In the hierarchical model $\balpha \sim N_k(\balpha; \bm_{\alpha}, \bV_{\alpha})$, $\bx \given \balpha, S \sim N_n(\bx; A\balpha, S)$, consider the following proposal for updating $\balpha \given \bx, S$:
\[ \balpha^* = \balpha + \frac{\varepsilon_1^2}{2} \bG_{\balpha} \nabla_{\balpha}  \log p(\balpha \mid \others) + \varepsilon_2 \bG_{\balpha}^{\frac{1}{2}} \bu, \]
where $\bu \sim N_n(0, I_{n})$, and we set $\varepsilon_1 = \sqrt{2}$, $\varepsilon_2 = 1$. Then, $q(\balpha^* \given \balpha) = p(\balpha^* \given \bx, S)$, i.e. this modified SM-MALA proposal leads to always accepted Gibbs updates.
\begin{proof} We compute
\begin{align*}
    \nabla_{\balpha} \log p(\balpha \given \others) &= \nabla_{\balpha} \log p(\bx \given \balpha, S) \pi(\balpha) = \nabla_{\balpha}\log\left\{ N_n(\bx; A\balpha, S) N_k(\balpha; \bm_{\alpha}, \bV_{\alpha}) \right\} \\
    &= - \frac{1}{2} \nabla_{\balpha} \{ (\balpha - \bm_{\alpha})^\top \bV_{\alpha}^{-1} (\balpha - \bm_{\alpha}) + (\bx - A \balpha)^\top S^{-1} (\bx - A \balpha) \} \\
    %&= A^\top S^{-1} (\bx - A \balpha) - \bV_{\alpha}^{-1}(\balpha - \bm_{\alpha}) \\
    %&= A^\top S^{-1} \bx - \bV_{\alpha}^{-1}(\balpha - \bm_{\alpha}) - A^\top S^{-1} A \balpha \\
    &= A^\top S^{-1} \bx + \bV_{\alpha}^{-1}\bm_{\alpha} - \left(A^\top S^{-1} A + \bV_{\alpha}^{-1} \right) \balpha 
\end{align*}
from which we immediately find $\bG_{\balpha} = \left( A^\top S^{-1} A  + \bV_{\alpha}^{-1} \right)^{-1}$. Then, the update is
\begin{align*}
    \balpha^* &= \balpha + \frac{\varepsilon_1^2}{2} \left( A^\top S^{-1} A  + \bV_{\alpha}^{-1} \right)^{-1} \left( A^\top S^{-1} \bx + \bV_{\alpha}^{-1}\bm_{\alpha} - \left(A^\top S^{-1} A + \bV_{\alpha}^{-1} \right)\balpha \right) + \tilde{\bu} \\
    &= \frac{\varepsilon_1^2}{2} \left( A^\top S^{-1} A  + \bV_{\alpha}^{-1} \right)^{-1} \left( A^\top S^{-1} \bx + \bV_{\alpha}^{-1}\bm_{\alpha}\right) - \left(1- \frac{\varepsilon_1^2}{2} \right)\balpha + \tilde{\bu},
\end{align*}
where $\tilde{\bu} \sim N(\bzero, \varepsilon_2^2 \left( A^\top S^{-1} A  + \bV_{\alpha}^{-1} \right)^{-1} )$. Setting $\varepsilon_1 = \sqrt{2}$ and $\varepsilon_2=1$ leads to the Gibbs update one obtains from a Gaussian likelihood and a Gaussian conjugate prior. In fact, since $q(\balpha^* \given \balpha) = p(\balpha^* \given \bx, S)$ then the acceptance probability for $\balpha^*$ is $\frac{p(\balpha^* \given \bx, S)q(\balpha \given \balpha^* )}{p(\balpha \given \bx, S)q(\balpha^* \given \balpha)} = 1$.
\end{proof}

\textbf{Proposition \ref{prop:simpa_converges}}. 
Suppose $\pi$ is everywhere non-zero and twice differentiable so that $\bg_{\bx}$ and $\bG_{\bx}$ are well defined. Let $\varepsilon >0$, $K \subset \Re^d$, $D > 0$. Additionally assume that if $\bx_{(m)} \in  K$ with $\text{dist}(\bx_{(m)}, K^c) = u$ with $0\leq u \leq 1$ then the proposal is changed to $\bx_{(\text{new})} \sim N( \bx_{(m)} + \frac{\varepsilon^2}{2} \bM_{(T^\text{adapt})} \tilde{\bg}_{\bx_{(m)}}, \varepsilon^2 \bM_{(T^\text{adapt})})$. Then, Algorithm \ref{alg:simpa} converges in distribution to $P$.

\begin{proof}
We show that SiMPA satisfies the assumptions of Theorem 21 of \cite{craiuetal15}. Algorithm \ref{alg:simpa} has by construction bounded jumps, no adaptation outside $K$ after iteration $T^{\text{adapt}}$ and the fixed kernel outside $K$ is bounded above by $(2\pi)^{-d/2} |\bM_{T^{\text{adapt}}}|^{1/2}$. Because we bound $\bg_{\bx}$ and $\bG_{\bx}$ with $D$ by using $\tilde{\bg}_{\bx}$ and $\tilde{\bG}_{\bx}$, the adaptive proposal kernel inside $K$ is $Q_{\delta}(\bx', \bx)$ where $\delta \in \Delta$ and $\Delta$ is a compact index set. Because outside $K$ we use a fixed proposal kernel with continuous densities with respect to Lebesgue measure and we assumed that the target density $p(\cdot)$ is also continuous, the $\epsilon$-$\delta$ condition holds (eq (6) in \citealt{craiuetal15}). Continuity of the target density and the proposal kernels hold by assumption. These assumptions are sufficient for the algorithm to satisfy the containment condition. The additional requirement to achieve convergence is diminishing adaptation, which holds by construction given the decreasing sequence $\{\gamma_m\}$. 
\end{proof}

\section{Coregionalization of MGPs}
\subsection{Equivalency result}
\begin{proposition} A $q$-variate MGP on a fixed DAG $\calG$, a domain partition $\mathbf{T}$, and a LMC cross-covariance function $\Cov_{\btheta}$ is equal in distribution to a LMC model built upon $k$ independent univariate MGPs, each of which is defined on the same DAG $\calG$ and the same domain partition $\mathbf{T}$.
\end{proposition}
\begin{proof}
For $i=1, \dots, M$, we want to show that the conditional densities $\pi(\bw_i \mid \bw_{[i]}) = N(\bw_i; \bH_i \bw_{[i]}, \bR_i)$ a $q$ variate MGP based on LMC cross-covariance $\Cov(\bl, \bl') = \bLambda \brho(\bl, \bl) \bLambda^\top$ (we drop $\btheta$ and $\bPhi$ subscripts on $\Cov$ and $\brho$, respectively, for simplicity) can be obtained equivalently via a LMC in which the $k$ margins are univariate MGPs
\begin{equation} %\label{eq:lmc_proof_1}
\begin{aligned}
    \Cov_{i, [i]} &= (I_{n_i} \otimes \bLambda) \brho_{i, [i]} (I_{n_{[i]}} \otimes \bLambda^\top) \qquad \Cov_{[i]}^{-1} = (I_{n_{[i]}} \otimes (\bLambda^\top)^{+}) \brho_{[i]}^{-1} (I_{n_{[i]}} \otimes \bLambda^{+}) \\
    \bH_i \bw_{[i]} &= \Cov_{i, [i]} \Cov_{[i]}^{-1} \bw_{[i]} \\
    &= (I_{n_i} \otimes \bLambda) \brho_{i, [i]} (I_{n_{[i]}} \otimes \bLambda^\top) (I_{n_{[i]}} \otimes (\bLambda^\top)^{+}) \brho_{[i]}^{-1} (I_{n_{[i]}} \otimes \bLambda^{+}) (I_{n_{[i]}} \otimes \bLambda) \bv_{[i]} \\
    &= (I_{n_i} \otimes \bLambda) \brho_{i, [i]} (I_{n_{[i]}} \otimes \bLambda^\top (\bLambda^\top)^{+}) \brho_{[i]}^{-1} (I_{n_{[i]}} \otimes \bLambda^{+} \bLambda) \bv_{[i]}\\
    &= (I_{n_i} \otimes \bLambda) \brho_{i, [i]} \brho_{[i]}^{-1} \bv_{[i]} = (I_{n_i} \otimes \bLambda) \Ddot{\bH}_i \bv_{[i]},
\end{aligned}
\end{equation}
where we denoted $\Ddot{\bH}_i = \brho_{i, [i]} \brho_{[i]}^{-1}$ and $\bLambda^{+}$ denotes the Moore-Penrose pseudoinverse $\bLambda^+ = (\bLambda^\top \bLambda)^{-1}$ (which exists because $\bLambda$ is assumed of full column rank), and therefore $\bLambda^+ \bLambda = I_k = \bLambda^\top (\bLambda^\top)^{+}$. Similarly,
\begin{equation} %\label{eq:lmc_proof_2}
\begin{aligned}
    \bR_i &= \Cov_i - \bH_i \Cov_{[i], i} = (I_{n_i} \otimes \bLambda) \brho_i (I_{n_i} \otimes \bLambda^\top) - 
    (I_{n_i} \otimes \bLambda) \brho_{i, [i]} \brho_{[i]}^{-1} \brho_{[i], i} (I_{n_{i}} \otimes \bLambda^\top) \\
    &= (I_{n_i} \otimes \bLambda) \left( \brho_i  - 
     \brho_{i, [i]} \brho_{[i]}^{-1} \brho_{[i], i} \right) (I_{n_{i}} \otimes \bLambda^\top) =(I_{n_i} \otimes \bLambda)  \Ddot{\bR}_i (I_{n_{i}} \otimes \bLambda^\top).
\end{aligned}
\end{equation}
Then
\begin{equation}
\begin{aligned}
    \pi(\bw_i \given \bw_{[i]}) &\propto |\bR_i|^{-\frac{1}{2}} \exp\left\{ -\frac{1}{2}(\bw_i - \bH_i \bw_{[i]})^\top \bR_i (\bw_i - \bH_i \bw_{[i]}) \right\} \\
    &= |(I_{n_i} \otimes \bLambda)  \Ddot{\bR}_i (I_{n_{i}} \otimes \bLambda^\top)|^{-\frac{1}{2}}\ \cdot \\
    &\qquad \cdot\ \exp\left\{ -\frac{1}{2}((I_{n_i} \otimes \bLambda)\bv_i - (I_{n_i} \otimes \bLambda)\Ddot{\bH}_i \bv_{[i]})^\top \ \cdot \right. \\
    &\left.\qquad\qquad\qquad \ \cdot \left( (I_{n_i} \otimes \bLambda)  \Ddot{\bR}_i (I_{n_{i}} \otimes \bLambda^\top) \right)^{-1} ((I_{n_i} \otimes \bLambda)\bv_i - (I_{n_i} \otimes \bLambda)\Ddot{\bH}_i \bv_{[i]}) \right\} \\
    &= |\Ddot{\bR}_i|^{-\frac{1}{2}} \exp \left\{ -\frac{1}{2}(\bv_i - \Ddot{\bH}_i \bv_{[i]})^\top \Ddot{\bR}_i^{-1} (\bv_i - \Ddot{\bH}_i \bv_{[i]}) \right\} = \pi(\bv_i \given \bv_{[i]}).
\end{aligned}
\end{equation}
We then proceed by reordering $\bv_i$, $\Ddot{\bH}_i$ and $\Ddot{\bR}_i$ by factor index (from $h=1, \dots, k$) rather than by location (see discussion above). After letting $K_i$ denote the appropriate permutation matrix that applies such reordering and letting $v_i^{(h)}$ be the $n_i \times 1$ vector whose elements are realizations of the $h$th latent factor at the reference subset $\calS_i$, we can write
\begin{align*}
    K_i \bv_i &= \begin{bmatrix} v_i^{(1)} \\ \vdots \\ v_i^{(k)} \end{bmatrix} \qquad\qquad  K_i \Ddot{\bH}_i \bv_i = \begin{bmatrix} \tilde{H}_i^{(1)} v_{[i]}^{(1)} \\ \vdots \\ \tilde{H}_i^{(1)}v_{[i]}^{(k)} \end{bmatrix} \\
    &K_i \Ddot{\bR}_i^{-1} K_i^\top = \text{blockdiag}\left\{ \tilde{R}^{(1)}_i , \dots, \tilde{R}^{(h)}_i \right\},
\end{align*}
where $\tilde{H}^{(h)}_i v^{(h)}_{[i]} = \rho^{(h)}_{i, [i]} \rho_{[i]}^{(h)^{-1}} v^{(h)}_{[i]}$ and $\tilde{R}^{(h)}_i = \rho^{(h)}_i  - \rho^{(h)}_{i, [i]} \rho_{[i]}^{(h)^{-1}} \rho^{(h)}_{[i], i} $, with $\rho^{(h)}_{i, [i]}$ denoting the correlation function of the $h$th LMC component evaluated between pairs of $\calS_i $ and $\calS_{[i]}$ and the other terms are defined analogously. Since reordering does not affect the joint density $\pi(\bv_i \given \bv_{[i]})$, we obtain
\begin{align*}
    \pi(K_i \bv_i \given \bv_{[i]}) = \pi(\bv_i \given \bv_{[i]}) &= \prod_{h=1}^k N(v^{(h)}_i; \tilde{H}^{(h)}_i, \tilde{R}^{(h)}_i).
\end{align*}
We have shown that the density of $(\bw_i \given \bw_{[i]})$ is the same as that of $(\bv_i \given \bv_{[i]})$ and that it can be written as a product of independent conditional densities. Then, for $i=1, \dots, M$:
\begin{align*}
    \pi_{\calG}(\bw_{\calS}) &= \prod_{i=1}^M \pi(\bw_i \given \bw_{[i]}) = \prod_{i=1}^M \pi(\bv_i \given \bv_{[i]}) = \prod_{i=1}^M \prod_{h=1}^k N(v^{(h)}_i; \tilde{H}^{(h)}_i, \tilde{R}^{(h)}_i) \\
    &= \prod_{h=1}^k \prod_{i=1}^M  N(v^{(h)}_i; \tilde{H}^{(h)}_i, \tilde{R}^{(h)}_i) = \prod_{h=1}^k \pi^{(h)}_{\calG}(v^{(h)}_{\calS}).
\end{align*}
We have shown that the meshed density $\pi_{\calG}$ at $\calS$ is equal to the product of $k$ independent meshed densities which are defined on the same DAG $\calG$ and the same partitioning of the spatial domain (i.e., $k$ independent MGPs).
\end{proof}

\subsection{Langevin methods for coregionalized MGPs} \label{sec:langevin_coreg}
We now show how Algorithm \ref{algorithm:meshed_posterior} is specified for the latent MGP model with LMC cross-covariance using \melange\ when targeting (\ref{eq:mgp_fullcond}). Let $K_i$ be the permutation matrix that reorders $\bv_i$ by factor, i.e. the $h$th block of $\tilde{\bv}_i = K_i \bv_i$ is the $n_i \times 1$ vector $v_i^{(h)}$, for $h=1,\dots,k$. Then, after letting $\bH_i = (I_{n_i} \otimes \bLambda) \Ddot{\bH}_i$ and $\bR_i = (I_{n_i} \otimes \bLambda) \Ddot{\bR}_i (I_{n_i} \otimes \bLambda^\top)$ and $r_j^{(h)} = v_j^{(h)} - \tilde{H}^{(h)}_{[j]\setminus \{ i\}} v^{(h)}_{[j]\setminus \{ i\}}$, the gradient $\nabla_{\bv_i} p(\bv_i \mid \others)$ can be found as we get
\begin{equation} \label{eq:gradient_gaussian}
    \begin{aligned}
    \nabla_{\bv_i} p(\bv_i \mid \others) &= - \Ddot{\bR}_i \left( \bv_i - \Ddot{\bH}_i \bv_{[i]} \right) + \bbf_i\\
    &= - K_i^\top \begin{bmatrix} 
    \tilde{R}_i^{(1)} \left(v_i^{(1)} - \tilde{H}_i^{(1)} v^{(1)}_{[i]} \right) +  \tilde{H}^{(1)^{\top}}_{i\to j} \tilde{R}^{(1)^{-1}}_{j} \left(r_j^{(1)} - \tilde{H}^{(1)}_{i\to j} v_i^{(1)}\right)
    \\ 
    \vdots \\
    \tilde{R}_i^{(k)} \left(v_i^{(k)} - \tilde{H}_i^{(k)} v^{(k)}_{[i]} \right) +  \tilde{H}^{(k)^{\top}}_{i\to j} \tilde{R}^{(k)^{-1}}_{j} \left(r_j^{(k)}- \tilde{H}^{(k)}_{i\to j} v_i^{(k)}\right)
    \end{bmatrix}+ \bbf_i,
    \end{aligned}
\end{equation}
where, letting $\calS_i = \{\bl_1, \dots, \bl_{n_i} \}$, we compute $\bbf_i = (\bbf_{i, \bl_{1}}^\top, \dots, \bbf_{i, \bl_{n_i}}^\top)^\top$ as the $n_i k \times 1$ vector whose $\bl$ block is \[ \bbf_{i, \bl} = \bLambda^\top \begin{bmatrix} \nabla_{\bv(\bl)} dF_1(y_1(\bl) \given \bv(\bl), \blambda_{[1,:]},  \bbeta_q, \gamma_q  )\\ \vdots \\ 
\nabla_{\bv(\bl)} dF_q(y_q(\bl) \given \bv(\bl), \blambda_{[q,:]},  \bbeta_q, \gamma_q  ) \end{bmatrix}. \]
For SM-MALA and SiMPA (Algorithm \ref{alg:simpa}) we compute 
\begin{equation} \label{eq:gmatrix_gaussian}
    \begin{aligned}
    \bG_{\bv_i}^{-1} &= K_i^\top\left( \oplus\left\{ \tilde{R}_{i}^{(h)} + \tilde{H}^{(h)^{\top}}_{i\to j} \tilde{R}_j^{(h)^{-1}} \tilde{H}^{(h)}_{i\to j} \right\}_{h=1}^k  + K_i \bF_i K_i^\top \right)K_i, \\
    \end{aligned}
\end{equation}
where $\oplus$ is the direct sum operator, $\bF_i = \oplus \{ \bA_i(\bl) \}_{\bl \in \calS_i}$, and after letting $x_j(\bl) =\blambda_{[j,:]}\bv(\bl)$, we compute $\bA_i(\bl) = -\sum_{j=1}^q \blambda_{[j,:]}^\top \blambda_{[j,:]} E\left[ \frac{\delta^2}{\delta^2 x_j(\bl)} \log dF_j(y_j(\bl) \given \bv(\bl), \blambda_{[j,:]}, \beta_j, \gamma_j) \right] $.

\subsection{Complexity in fitting coregionalized cubic MGPs} \label{appx:complexity_lmcqmgp}
We now consider model (\ref{eq:latent_gaussian_lmc}) and replace the GP prior with an MGP based on LMCs (as in Section \ref{sec:meshed_lmc}) using a cubic mesh (Figure \ref{fig:qmgpdag}), whose main feature is that the number of parents of each reference node is at most $d$ when the dimension of the input space is $d$ (in spatial settings, $d=2$). The resulting coregionalized QMGP is implemented on $k$ factors to model dependence across $q\ge k$ outcomes, when at $n$ locations we observe at least one of them. We assume $\barcalT = \emptyset$, SiMPA updates at each block and let $H$ refer to the number of available processors for parallel computations.

In the resulting Algorithm \ref{algorithm:lmc_meshed_posterior}, step \ref{alg:meshed_posterior:step1} requires the update of $q$ sets of $p$ covariates plus $k$ factor loadings. SiMPA can be used here for an expected cost at iteration $m$ of $O(\gamma_mnq(p+k)^3 +  nq(p+k)^2)$ which is approximately $O(qn(p+k)^2)$ for large $m$ because $\gamma_m \downarrow 0$. The compute time is $O(\gamma_mnq(p+k)^3/H + nq(p+k)^2/H)$, respectively, because $(\bbeta_j, \blambda_{[j,:]}) \perp (\bbeta_h, \blambda_{[h,:]}) \given \by_\calT, \bv_\calS$. Step \ref{alg:meshed_posterior:step2} costs $O(qn)$ flops assuming a Metropolis update, and the compute time is $O(qn/H)$. Step \ref{alg:meshed_posterior:step3} involves the evaluation of $k$ independent sets of MGP densities, each of which is a product of $M$ Gaussian conditional densities. We make the simplifying assumption that $n_i \approx m \approx n/M$ and $n_{[i]} \le dm \approx dn/M$ for all $i=1, \dots, M$---we are taking $M$ partitions of size $m$ and a cubic mesh which attributes at most $d$ parents to each node in the DAG. The cost for this update is due to computing $\Ddot{\bR}_{i}$ for all $i$, which is $O(kM(dm)^3) = O(nkd^3 m^2)$ flops in $O(nkd^3 m^2/H)$ time. Finally, reference sampling of $\bv_i$, $i=1, \dots, M$, whose sizes are $mk$, is performed via SiMPA in $O(\gamma_m nm^2k^3 + nmk^2)$ flops and in $O(\gamma_m nm^2k^3/H + nmk^2/H)$ time, respectively, assuming that each color of $\calG$ includes at least $H$ nodes. 
In summary, the cost of a $k$-factor coregionalized QMGP fit via SiMPA is linear in $n$ and $q$, which may be large, quadratic on $k$ and $p$, which we assume relatively small, and cubic on the domain dimension $d$, which is 2 or 3 for the spatial and spatiotemporal settings on which we focus.

\section{Applications Supplement} \label{appx:applications}
In all our applications, all methods are configured to use up to 16 CPU threads in a workstation with 128GB memory and an AMD Ryzen 9 5950X CPU on the Ubuntu 22.04.2 LTS operating system and using Intel MKL version 2019.5.28 for BLAS/LAPACK. R package \texttt{meshed} (v.0.2) allows one to set the number of OpenMP \citep{dagum1998openmp} threads, whereas \texttt{Hmsc} takes advantage of parallelization via BLAS when performing expensive operations (e.g., \texttt{chol($\cdot$)}). The \texttt{R-INLA} package used to implement SPDE-INLA methods can similarly take advantage of multithreaded operations.

\subsection{Bivariate counts analysis on 750 synthetic datasets} \label{appx:750}
The comparison above is based on a single dataset; we replicate the same analysis on 750 smaller datasets. We generate Poisson data on a $50 \times 50$ regular grid, for a total of 2500 observations for $y_j(\bl) \sim Pois(\exp\{ \eta_j(\bl) \})$ where $\bolds{\eta}(\bl) = \bLambda \bv(\bl)$ and $\bv(\cdot)$ is a bivariate GP with independent Mat\'ern correlations with $\nu_j = 1/2$ for $j=1, 2$ and $\phi_2 = 2.5$. We choose $\phi_1 \in \{ 2.5, 12.5, 25\}$. We introduce missing values at $1/5$ of spatial locations independently for each outcome.
We fix the $2\times 2$ loading matrix via
\[ \bLambda = \begin{bmatrix} \lambda_{11} & \lambda_{12} \\ \lambda_{21} & \lambda_{22} \end{bmatrix} = \texttt{chol}\left( \begin{bmatrix} \lambda_1 & 0 \\ 0 & 1  \end{bmatrix} \cdot \begin{bmatrix}  
1 & \rho \\ \rho & 1 
\end{bmatrix} \cdot \begin{bmatrix}  \lambda_1 & 0 \\ 0 & 1  \end{bmatrix} \right),
\]
\begin{figure}
    \centering
    \includegraphics[width=.5\textwidth]{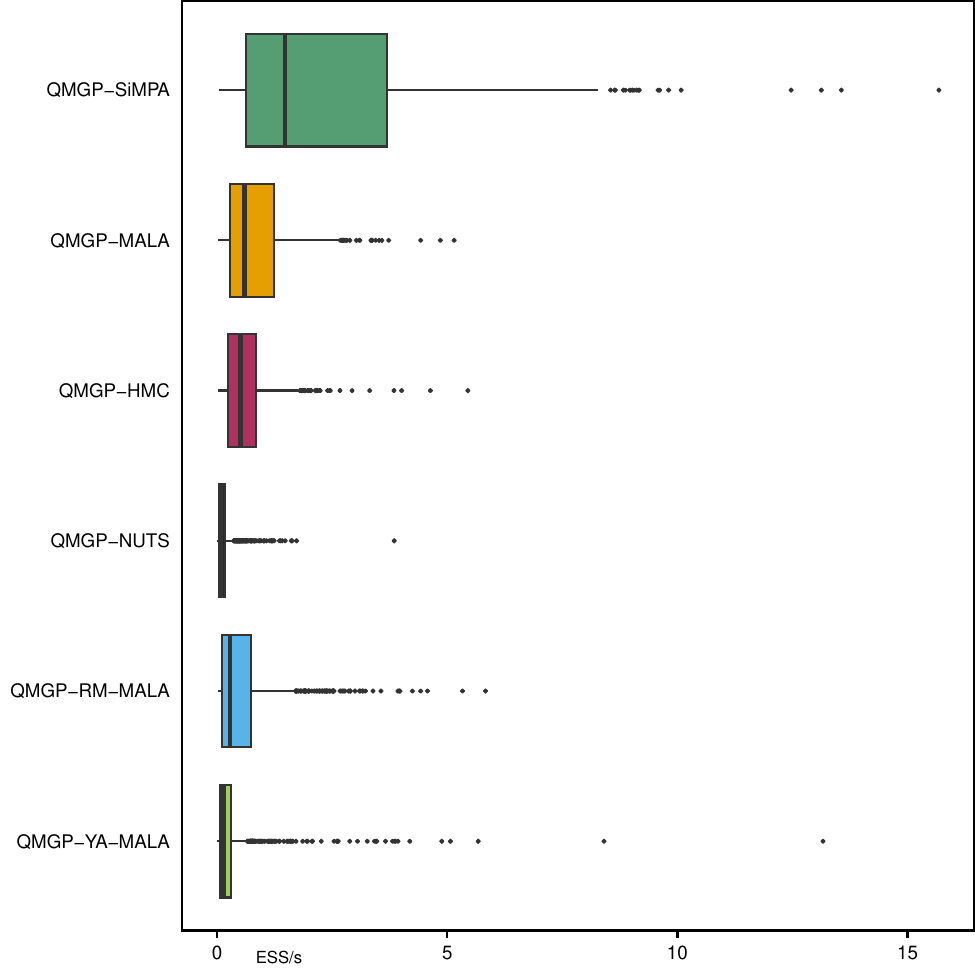}
    \caption{Efficiency in terms of ESS/s in the estimation of $\rho$ over 750 simulated datasets.}
    \label{fig:mvpoisson_estim_ess}
\end{figure}
which implies $\lambda_{11} = \lambda_1$, $\lambda_{12} = 0$, and $\lambda_{21}$ and $\lambda_{22}$ are such the latent correlation between the first and second margin is exactly $\rho$. We choose $\lambda_1 \in \{ \frac{\sqrt{2}}{2}, 2 \}$ and $\rho \in \{ -.9, -0.25, 0, 0.65, .9\}$. We generate 25 datasets for every combination of values of $\phi_1$, $\lambda_1$ and $\rho$. 
We target estimation of the latent correlation $\rho = \text{Corr}(w_1(\bl), w_2(\bl))$ in terms of absolute error and efficiency (ESS/s), along with the empirical coverage of 95\% intervals for the log-intensity for both outcomes. We compare SiMPA with several other methods -- all the coregionalized QMGP methods use parameter expansion as in \cite{grips}. Figure \ref{fig:mvpoisson_estim_ess} and Table \ref{tab:appx_lgcp750_summary_table} summarize our findings across the 750 datasets: SiMPA has low estimation error, high sampling efficiency, and excellent uncertainty quantification relative to all other tested methods. %Additional details are available in the supplement.

\begin{table}[]
    \centering
    \resizebox{0.5\textwidth}{!}{%
    \begin{tabular}{|l|ccc|}
\hline
 & $\rho$  & \multicolumn{2}{c|}{$\eta_{\text{test}}(\bl)$} \\
\multirow{-2}{*}{Method} & RMSE & RMSPE & Covg$_{\text{95\%}}$\\
\hline
QMGP-SiMPA & \textbf{0.08} & \textbf{0.45} & \textbf{0.95}\\
QMGP-MALA & 0.09 & 0.46 & 0.94\\
QMGP-HMC & 0.09 & 0.55 & 0.93\\
QMGP-NUTS & 0.16 & 0.46 & 0.93\\
QMGP-RM-MALA & 0.37 & 0.86 & 0.87\\
QMGP-YA-MALA & 0.42 & 1.83 & 0.07\\
QMGP(AG10)-NUTS & 0.56 & 0.91 & 0.94\\
QMTP(AG10)-NUTS & 0.47 & 0.46 & 0.92\\
SPDE-INLA & 0.21 & 0.66 & 0.67\\
\hline
\end{tabular} }
    \caption{Performance summary across 750 datasets in the estimation of the latent correlation and the linear predictor on the test sets.}
    \label{tab:appx_lgcp750_summary_table}
\end{table}

\begin{comment}
\subsection{Model outputs in the bivariate count data application}
We report model outputs from the SPDE-INLA and QMGP-MALA models for the bivariate count data application of Section 5.1.
\begin{figure}[H]
    \centering
    \includegraphics[width=0.99\textwidth]{figures/mvpoisson_inla.pdf}
    \caption{Output from fitting a SPDE model via INLA. Top row: recovered log-intensity and predictions for both outcomes. Bottom row: width of posterior credible intervals about log-intensity, and residual log-intensity.}
    \label{fig:mvpoisson_inla}
\end{figure}

\begin{figure}[H]
    \centering
    \includegraphics[width=0.99\textwidth]{figures/mvpoisson_meshg1.pdf}
    \caption{Output from fitting a coregionalized QMGP via MALA. Top row: recovered log-intensity and predictions for both outcomes. Bottom row: width of posterior credible intervals about log-intensity, and residual log-intensity.}
    \label{fig:mvpoisson_simpa}
\end{figure}
\end{comment}

\subsection{Latent process sampler efficiency in multi-type data}
The analysis in the previous section models both outcomes as Poisson counts. In this section, we use the same setup and $\lambda_1 \in \{2.5, 12.5 \}$, but consider the following pairs of outcome types: $\{ (\text{Gaussian}, \text{Poisson}), (\text{Neg. Binomial}, \text{Binomial}), (\text{Neg. Binomial}, \text{Poisson}) \}$, for a total of 1500 datasets, of which 500 include a Binomial or Gaussian outcome and 1000 include a Poisson or Neg. Binomial outcome. Because we target a comparison of posterior sampling efficiency in integrating out the latent spatial effects via MCMC, we fix all unknowns ($\bLambda$, $\phi_1$, $\phi_2$) to their true values except for the latent process. For each dataset, we calculate ESS/s for samples of $\bw(\bl_i)$, $\bl_i = 1, \dots n $. After computing the median ESS/s as a summary efficiency measure for each dataset, we compute the mean of this measure over all datasets. Efficiency summary results are reported in Table \ref{tab:samplers_efficiency}. We also report each method's RMSPE and coverage about $\eta_j(\bl), j=1,2$ in Table \ref{tab:samplers_predict}. SiMPA is again more efficient than other methods in integrating out the spatial effects, while matching or outperforming them in  out-of-sample inference about model parameters.

\begin{table}[ht]
\centering
\begin{tabular}{|l|rrrr|}
  \hline
Method & Binomial & Gaussian & Negative Binomial & Poisson \\ 
  \hline
MALA & 1.15 & 8.35 & 1.31 & 2.81 \\ 
NUTS & 0.25 & 2.15 & 0.32 & 0.68 \\ 
SiMPA & \textbf{4.26} & \textbf{17.65} & \textbf{9.03} & \textbf{8.90} \\ 
SM-MALA & 1.81 & 7.85 & 3.45 & 3.59 \\ 
   \hline
\end{tabular}
\caption{Efficiency in posterior sampling of $\bw(\cdot)$, in terms of ESS/s, for different types of outcomes in the bivariate synthetic data application with multi-type outcomes.}
    \label{tab:samplers_efficiency}
\end{table}

\begin{table}[ht]
\centering
\resizebox{\textwidth}{!}{
\begin{tabular}{|l|rrrr|rrrr|}
  \hline
\multirow{2}*{Method}  & \multicolumn{4}{c|}{RMSPE} & \multicolumn{4}{c|}{Covg. 95\%} \\ 
 & Binomial & Gaussian & Neg. Bin. & Poisson & Binomial & Gaussian & Neg. Bin. & Poisson \\ 
  \hline
MALA & 0.450 & 0.332 & 14.206 & 1.198 & 0.897 & 0.942 & 0.899 & 0.930 \\ 
  NUTS & 0.449 & 0.328 & 13.932 & 1.195 & 0.932 & 0.944 & 0.934 & 0.941 \\ 
  SiMPA & \textbf{0.449} & \textbf{0.327} & \textbf{13.923} & \textbf{1.194} & \textbf{0.944} & \textbf{0.948} & \textbf{0.947} & \textbf{0.947} \\ 
  SM-MALA & \textbf{0.449} & \textbf{0.327} & 13.971 & 1.212 & 0.943 & \textbf{0.948} & 0.939 & 0.940 \\ 
   \hline
\end{tabular}
}
\caption{RMSPE in predicting $\eta_j(\bl)$ at locations in the test set and empirical coverage of 95\% credible intervals about $\eta_j(\bl)$, $j=1,2$, for different types of outcomes in the bivariate synthetic data application with multi-type outcomes.}
    \label{tab:samplers_predict}
\end{table}

\subsection{Multi-species N-mixture abundance modeling} \label{sec:spammix}
The total number of individuals of a certain animal species in a region is known as the local abundance.
Community ecologists seek to estimate abundance using spatially replicated count data of multiple species. At each spatial location, the observed counts correspond to a portion of the latent abundance of each of $q$ species. The unobserved abundance of species $j$ can be estimated via a model for count data that accounts for imperfect detection. See \cite{Royle2004}, \cite{Mimnaghetal22} and reference therein. In the context of joint species distribution models of count data, one lets the local abundance depend on covariates and latent variables accounting for cross-species dependence through a Poisson log-linear model, with the observed abundance then having a conditional binomial likelihood.
\begin{figure}
    \centering
    \includegraphics[width=0.9\textwidth]{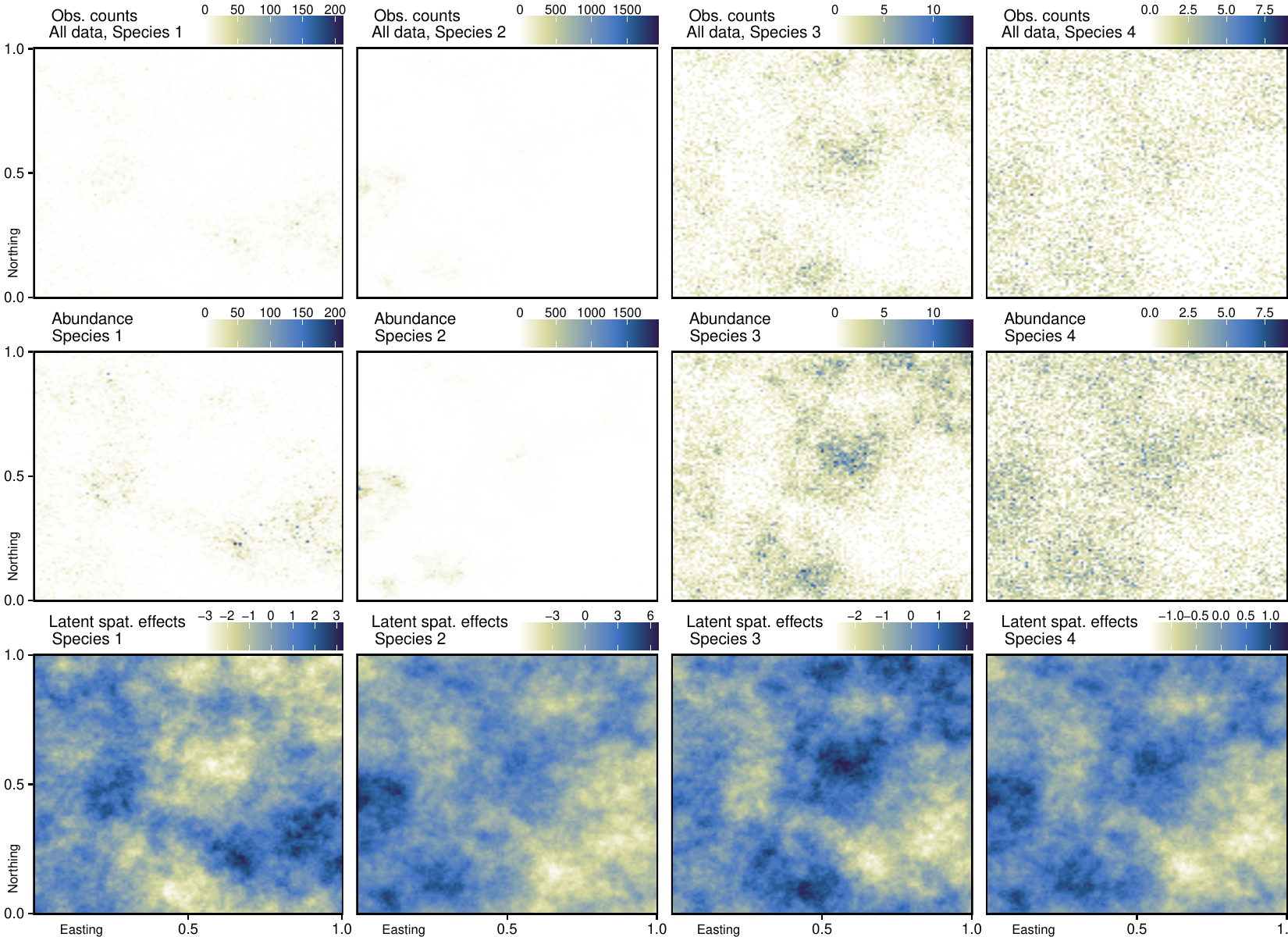}
    \caption{Simulated data on multi-species abundance. Top row: all counts observed with imperfect detection, including the 20\% missing from the training data; mid row: unobserved species abundance; bottom row: realization of $w_j(\bl) = \blambda_{[j,:]} \bv(\bl)$, $\bl \in \calD$.}
    \label{fig:abundance_plot_data}
\end{figure}
Here, we consider an extension of \cite{Mimnaghetal22} to include MGP random effects:
\begin{equation}\label{eq:nmixture}
\begin{alignedat}{3}
    v_h(\cdot) &\sim MGP_{\calG}(\bzero, \rho_h(\cdot,  \cdot;\btheta_h)),\qquad &&h=1,\dots, k \\
    w_j(\bl) &= \blambda_{[j,:]} \bv(\bl) &&j=1,\dots, q\\
    N_j(\bl) \mid \bbeta_j, \blambda_{[j,:]}, \bv(\bl) &\sim \text{Poisson}( \mu_j(\bl) ) &&\mu_j(\bl) = \exp\{ \bx_j(\bl)^\top\bbeta_j + w_j(\bl) \}\\
    y_j(\bl) \mid N_j(\bl), \bolds{\xi}_j &\sim \text{Binomial}\left(N_j(\bl), p_j(\bl) \right) \qquad &&p_j(\bl) = \left[ 1 + \exp\{-\bz_j(\bl)^\top \bolds{\xi}_j \} \right]^{-1},   
\end{alignedat}
\end{equation}
where we let a set of species-specific covariates $\bx_j(\bl)$ explain the latent species abundance and another set $\bz_j(\bl)$ impact the detection probability and hence the observed counts $y_j(\bl)$. A latent factor model with $k\le q$ spatial random effects is used to characterize dependence across species in abundances.
Applied goals include the estimation of $\bbeta_j$ and $\bxi_j$ for $j=1, \dots, q$, the cross-covariance function $\Cov_{\btheta} = \bLambda \brho(\bl, \bl',  \bPhi) \bLambda^\top$ via the estimation of $\btheta = (\text{vec}(\bLambda)^\top, \bPhi^\top)^\top$, and the local abundance of species $j$ at $\bl$, $N_j(\bl)$. Posterior computations for (\ref{eq:nmixture}) simplify by marginalizing $N_j(\bl)$ from the model likelihood; the marginal model is $
    p(y_j(\bl) \mid \text{---}) = \text{Poisson}\left( p_j(\bl) \mu_j(\bl)  \right)$.
After collecting posterior samples of $\bbeta_j, \bolds{\xi}_j, \bLambda, \bv$, we estimate $N_j(\bl)$ using the fact that $N_j(\bl) \mid N_j(\bl) > y_j(\bl) = y_j(\bl) + \tilde{N}_j(\bl)$, where $\tilde{N}_j(\bl) \sim \text{Poisson}([1-p_j(\bl)] \mu_j(\bl) )$. If $y_j(\bl)$ is missing, we proceed by first sampling from $\pi(\bv(\bl) \mid \bv_\calS, \bPhi)$, then $N_j(\bl) \sim \text{Poisson}(\mu_j(\bl))$. 

We simulate abundance data of $q=4$ species at $n=$14,400 spatial locations on a regular grid using model (\ref{eq:nmixture}). We sample $k=2$ latent factors from independent unrestricted GPs with exponential correlation and spatial decays $\phi_1 = \phi_2 = 4$. The factor loadings are set to $(\lambda_{11}, \lambda_{21}, \lambda_{31}, \lambda_{41}, \lambda_{22}, \lambda_{32}, \lambda_{42}) = (1.3, -0.65, -0.9, -0.3, 2, 0.35, 0.4)$; these values lead to latent spatial cross-species correlations ranging from $\text{corr}(w_3(\bl), w_1(\bl) ) %= \lambda_{[3,:]}\lambda_{[1,:]}^\top/ \sqrt{\lambda_{[3,:]}\lambda_{[3,:]}^\top \lambda_{[1,:]}\lambda_{[1,:]}^\top} 
\approx -0.93$ to $\text{corr}(w_4(\bl), w_2(\bl) ) \approx 0.95$. These correlations decrease for increasing spatial distances as modeled by the underlying exponential covariances. In order to generate the latent abundance and the observed counts at each location, we sample the covariate vector $(x(\bl) , \bz(\bl)^\top )^\top$ independently from a Gaussian distribution with correlation matrix $\Sigma_x$ whose off-diagonal elements are $\sigma_{x, z_1} = 0.8$, $\sigma_{x, z_2} = -0.3$ $\sigma_{z_1, z_2} = -0.7$. We let $(\beta_1, \beta_2, \beta_3, \beta_4) = (-1, 0.5, 0, 0)$ and $\bolds{\xi}_1 = (1, -1)^\top$, $\bolds{\xi}_2 = (-1, 1)^\top$, $\bolds{\xi}_3 = (0.5, -0.5)^\top$, $\bolds{\xi}_4 = (-1, -1)^\top$. Finally, for each of the $4$ species, we introduce missingness by independently dropping 20\% of the observed count data from the training set uniformly at random; the counts of at least one species are missing at 8,480 locations. Because not all species are observed at all spatial locations, the resulting data are misaligned. This scenario mirrors a situation in which a subset of the total number of individuals of species $j$ are counted at a subset of all locations. Figure \ref{fig:abundance_plot_data} reports the full dataset (including the missing data) along with the latent variables.

\begin{table}[H]
\centering
\resizebox{0.7\columnwidth}{!}{%
\begin{tabular}{l|c|cc|cc|cc|c}
  \hline
\multirow{2}*{Method} & \multirow{2}*{$j$} & RMSE  & MAE & CI$_{95}$ & ESS/s & RMSE & ESS/s & \multirow{2}*{Time(s)} \\ 
& & \multicolumn{2}{c|}{$N_j(\bl)$} & \multicolumn{2}{c|}{ $\bw_j(\bl)$} & \multicolumn{2}{c|}{ $(\beta, \bxi, \blambda_{[j,:]})$} &  \\
  \hline
    \multirow{4}*{SiMPA} 
   & 1 & 2.156 & 0.443 & 0.949 & \textbf{3.16} & \textbf{0.0197} & \textbf{9.09}& \multirow{4}*{139} \\ 
   & 2 & 5.948 & 0.443 & 0.951 & \textbf{3.91} & \textbf{0.0289} & \textbf{2.78}& \\ 
   & 3 & 0.951 & 0.479 & \textbf{0.950} & \textbf{3.42} & 0.0569 & \textbf{15.8}& \\ 
   & 4 & \textbf{0.828} & 0.494 & 0.959 & \textbf{4.11} & 0.0156 & \textbf{25.6}& \\ 
   \hline
\multirow{4}*{MALA} 
   & 1 & 2.164 & 0.447 & \textbf{0.950} & 0.80 & 0.0203 & 1.82&  \multirow{4}*{\textbf{114}}  \\ 
   & 2 & 5.954 & 0.449 & 0.947 & 1.64 & 0.2640 & 1.02& \\ 
   & 3 & \textbf{0.950} & 0.482 & 0.940 & 0.90 & 0.0749 & 1.86& \\ 
   & 4 & \textbf{0.828} & 0.495 & 0.948 & 1.36 & 0.0487 & 1.71& \\ 
   \hline
\multirow{4}*{SM-MALA} 
   & 1 & 2.182 & \textbf{0.423} & 0.929 & 1.87 & 0.0473 & 7.24& \multirow{4}*{194} \\ 
   & 2 & 9.786 & \textbf{0.410} & 0.932 & 2.09 & 0.4330& 1.73& \\ 
   & 3 & 0.951 & 0.488 & 0.920 & 2.14 & 0.0498 & 9.77& \\ 
   & 4 & 0.830 & \textbf{0.488} & 0.920 & 2.27 & 0.0596 & 19.5& \\  
   \hline
\multirow{4}*{HMC} 
   & 1 & \textbf{2.148} & 0.444 & 0.945 & 0.92 & 0.0421 & 1.97& \multirow{4}*{214}  \\ 
   & 2 & \textbf{5.874} & 0.444 & \textbf{0.950} & 2.04 & 0.0370 & 1.70& \\  
   & 3 & 0.951 & 0.479 & 0.944 & 1.13 & 0.0606 & 2.32 & \\  
   & 4 & \textbf{0.828} & 0.494 & 0.952 & 1.92 & 0.0140 & 2.51& \\  
   \hline
  \multirow{4}*{NUTS} 
   & 1 & 2.158 & 0.442 & 0.946 & 0.17 & 0.0242 & 0.93& \multirow{4}*{907} \\ 
   & 2 & 5.876 & 0.446 & 0.941 & 0.24 & 0.0534 & 0.61& \\  
   & 3 & \textbf{0.950} & \textbf{0.477} & 0.944 & 0.20 & \textbf{0.0563} & 1.34& \\  
   & 4 & \textbf{0.828} & 0.495 & \textbf{0.949} & 0.27 & \textbf{0.0107} & 1.31& \\  
   \hline
  \multirow{4}*{YA-MALA} 
   & 1 & 2.178 & 0.424 & 0.329 & 0.10 & 0.1073 & 0.09& \multirow{4}*{132} \\ 
   & 2 & 7.437 & 0.414 & 0.346 & 0.08 & 0.2970 & 0.20& \\  
   & 3 & 0.954 & 0.478 & 0.346 & 0.10 & 0.1151 & 0.08& \\  
   & 4 & 0.829 & 0.493 & 0.356 & 0.09 & 0.0644 & 0.10& \\  
   \hline 
  \multirow{4}*{Ellipt-SS} 
   & 1 & 3.204 & 0.486 & 0.937 & 0.79 & 0.2953 & 0.15& \multirow{4}*{264} \\ 
   & 2 & 15.570 & 0.567 & 0.821 & 0.84 & 0.3609 & 0.08& \\  
   & 3 & 1.014 & 0.512 & 0.852 & 0.79 & 0.0906 & 0.33& \\  
   & 4 & 0.834 & 0.509 & 0.874 & 0.80 & 0.0582 & 0.28& \\  
   \hline
\end{tabular}}
\caption{A comparison of posterior sampling methods for fitting the same model for abundance data with imperfect detection based on latent QMGPs. For the four species, we compare the root mean squared error (RMSE) as well as the mean absolute error (MAE) in estimating the latent abundance $N_j(\bl)$. For $\bw_j(\bl)$, we report the empirical coverage of 95\% credible intervals (CI$_{95}$) and the median effective sample size (ESS) per unit time across spatial locations. We also report the RMSE and median ESS/s in estimating the vector $(\beta, \bxi, \blambda_{[j,:]})$ for each species.}
    \label{tab:abundance_results}
\end{table}

We fit model (\ref{eq:nmixture}) with a QMGP prior on the latent effects. To build the QMGP prior, we use axis-parallel partitioning to tessellate the spatial domain into $M=400$ blocks each including $36$ spatial locations. We choose this partitioning setup to ensure all sampling methods proceed swiftly and without making the overly restrictive spatial conditional independence assumptions that would result from a finer partitioning scheme. We detail the common posterior sampling algorithm used for fitting model (\ref{eq:nmixture}) in Appendix \ref{appx:nmix} as a minor modification of Algorithm \ref{algorithm:lmc_meshed_posterior}.

We compare our proposed SiMPA with MALA, simplified Riemannian manifold MALA (RM-MALA; \citealt{girolamicalderhead11}), HMC and NUTS with dual averaging (Algorithms 5 and 6 in \citealt{nuts}, respectively), the elliptical slice sampler \citep{elliptss}, and YA-MALA. All methods perform 20,000 MCMC iterations, of which we drop the first half as burn-in. All gradient-based methods use dual averaging to adapt $\varepsilon$ for $T^{\text{adapt}} = 10,000$ iterations. 
Figure \ref{fig:abundance_simpa} reports the SiMPA-estimated latent effects along with uncertainty quantification. Table \ref{tab:abundance_results} summarises our findings: because all sampling methods target the same posterior distribution, we do not expect major discrepancies in point estimates. Our SiMPA method is on par with other state-of-the-art gradient-based methods when estimating unknown model parameters, but outperforms them in terms of sampling efficiency measured as ESS per unit time. Because the SiMPA 95\% intervals on the latent effects are subjectively better calibrated than those from other methods, it more robustly quantifies uncertainty about the latent spatial effects. 
Finally, because SiMPA and YA-MALA adapt the preconditioner at the same iterations, we conclude that SiMPA is a much more efficient adaptation scheme for MALA preconditioning in this context.

\begin{figure}
    \centering
    \includegraphics[width=\textwidth]{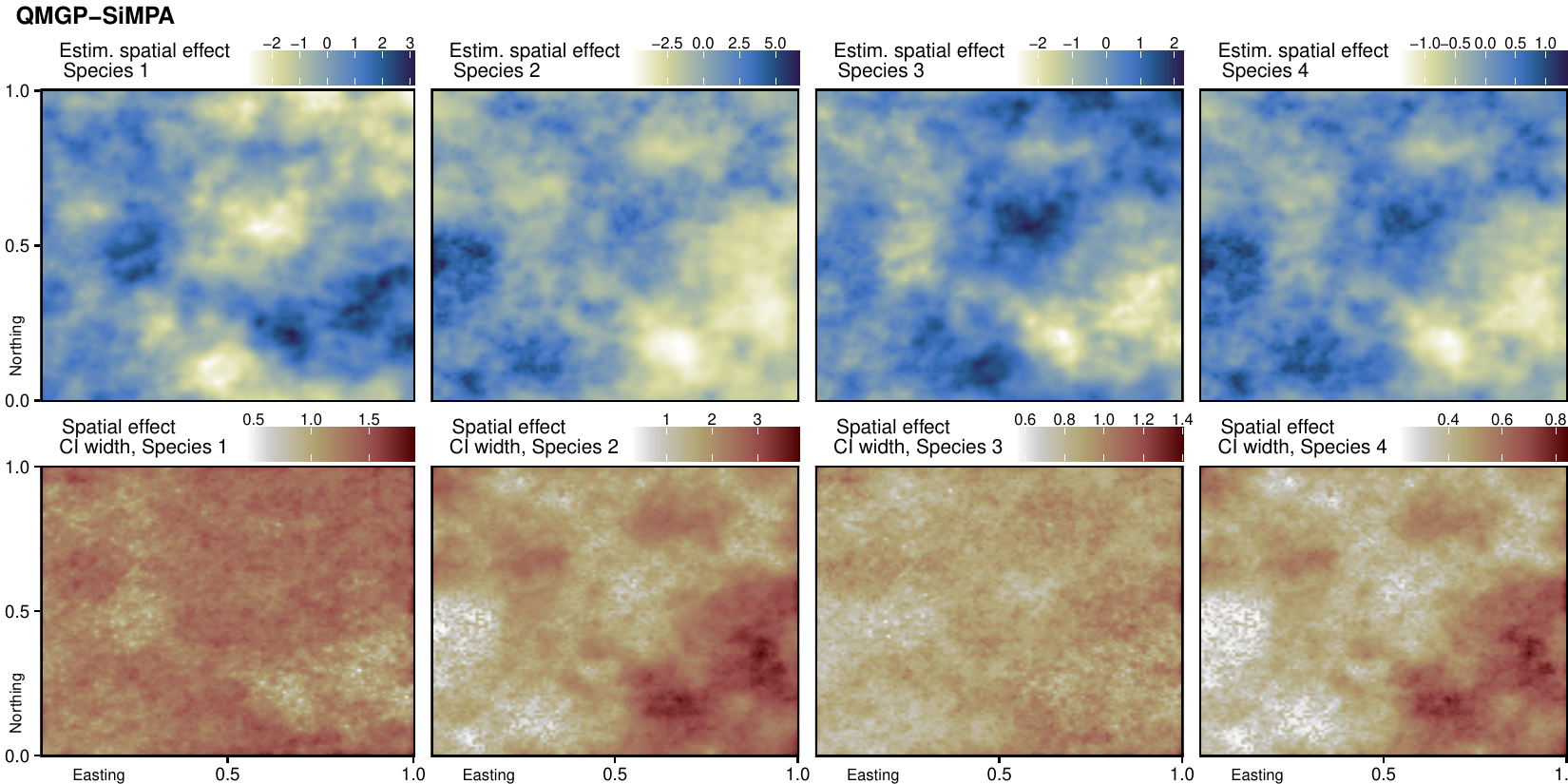}
    \caption{Estimation of the latent spatial effects in the multi-species abundance model. Top row: posterior mean of the species-specific spatial random effects $\bw_j(\bl)$; bottom row: width of the 95\% pointwise credible interval on $\bw_j(\bl)$ as computed via SiMPA.}
    \label{fig:abundance_simpa}
\end{figure}
%In such scenarios, researchers are interested in understanding -- Species occurrence \citep{DorazioRoyle05, Doseretall2022} 
%GLM cannot be used to model imperfect detection.

\subsection{Posterior sampling of the multi-species N-mixture model} \label{appx:nmix}
All the tested sampling methods are used in steps \ref{alg:spammix:step1} and \ref{alg:spammix:step3} of Algorithm \ref{algorithm:spammix}.
\begin{algorithm}
{ \small
  \caption{Posterior sampling and prediction of LMC model (\ref{eq:nmixture}) with MGP priors.}\label{algorithm:spammix}
  \begin{algorithmic}[1]
  \Statex Initialize $\bbeta_j^{(0)}, \bxi_j^{(0)}$, $\bLambda^{(0)}$ for $j=1, \dots, q$, $\bv_{\calS}^{(0)}$, and $\bPhi^{(0)}$
  \Statex \textbf{for} $t \in \{1, \dots, T^*, T^* + 1, \dots, T^* + T\}$ \textbf{do} \Comment{{\footnotesize sequential MCMC loop}}
    \Statex \algindent \textbf{for} $j=1, \dots, q$, \textbf{do \underline{in parallel}} 
    \foritem \algindent Block-update $\bbeta_j^{(t)}, \bxi_j^{(t)}, \blambda_{[j,:]}^{(t)} \given \by_\calT, \bv_{\calS}^{(t-1)}$ \label{alg:spammix:step1} %\Comment{{\footnotesize $O(nq(p+k)^2$)}}
    %\Statex \algindent \textbf{end for}
    \foritem use Metropolis-Hastings to update $\bPhi^{(t)} \given \bv^{(t-1)}_\calS $ \label{alg:spammix:step2} %\Comment{{\footnotesize $O(n k d^3  m^2)$ }}
    \Statex \algindent \textbf{for} $c \in \text{Colors}(\calG)$ \textbf{do} \Comment{{\footnotesize sequential}}
    \Statex \algindent \algindent \textbf{for} $i \in \{ i: \text{Color}(a_i) = c \}$ \textbf{do \underline{in parallel}} 
    \foritem \algindent \algindent Update $\bv_{i}^{(t)} \given \bv_{\mb(i)}^{(t)}, \by_i, \bLambda^{(t)}, \bPhi^{(t)}, \{ \bbeta_j^{(t)}, \gamma_j^{(t)} \}_{j=1}^q$ \label{alg:spammix:step3} %\Comment{{\footnotesize $O(nmk^2)$}}
    %\Statex \algindent \algindent \textbf{end for} 
    %\Statex \algindent \textbf{end for} 
  %\Statex \textbf{end for}
%  \Statex Assuming convergence has been attained after $T^*$ iterations:
%\Statex discard $\{ \bbeta_j^{(t)}, \bxi_j^{(t)} \}_{j=1}^q, \bv_{\calS}^{(t)}, \bLambda^{(t)}, \bPhi^{(t)}$ for $t = 1, \dots, T^*$
%\Statex \textbf{Output:} Correlated sample of size $T$ with density \[ \{ \bbeta_j^{(t)}, \bxi_j^{(t)} \}_{j=1}^q, \bv_{\calS}^{(t)}, \bLambda^{(t)}, \bPhi^{(t)} \sim \pi_\calG(\{ \bbeta_j, \bxi_j \}_{j=1}^q, \bv_{\calS}^{(t)}, \bLambda, \bPhi, \mid \by_{\calT}).\]
%\Statex \textbf{Predict at $\bl^* \in \calU$}: for $t=1, \dots, T$ and $j=1, \dots, q$, sample from $\pi(\bv_{\bl^*}^{(t)} \given \bv_{[\bl^*]}^{(t)}, \bPhi^{(t)})$, then from $F_j(w_j(\bl^*)^{(t)}, \bbeta_j^{(t)}, \blambda_{[j,:]}^{(t)}, \bxi_j^{(t)})$ 
\end{algorithmic} }
\end{algorithm}

Block updating $\bbeta_j, \bxi_j, \blambda_{[j,:]}$ using SiMPA requires second order information about the target. Because $\bx_j(\bl)^\top \bbeta_j + \blambda_{[j,:]} \bv(\bl) = (\bx_j(\bl)^\top, \bv(\bl)^\top) (\bbeta_j^\top, \blambda_{[j,:]}) = \tilde{\bx}_j(\bl) \tilde{\bbeta}_j$, we can proceed without loss of generality by outlining the block-update for $(\bbeta_j, \bxi_j)$ in the model without spatial effects. After letting $\pi_j(\bl) = (1 + \exp\{ - \bz_j(\bl)^\top \bxi_j \})^{-1}$ and $\alpha_j(\bl) = \exp\{ \bx_j(\bl)^\top \bbeta_j \}$, we find
\begin{align*}
   -\frac{\delta^2 P(y_j(\bl) = y)}{\delta^2 (\bbeta_j, \bxi_j)} &= \begin{bmatrix}
g_{11} & g_{12} \\ g_{12} & g_{22}
\end{bmatrix},
\end{align*}
where 
$g_{11} = \alpha_j(\bl) \pi_j(\bl) \bx_j(\bl)\bx_j(\bl)^\top$, $g_{12} = \alpha_j(\bl) \pi_j(\bl) (1-\pi_j(\bl)) \bx_j(\bl) \bz_j(\bl)^\top$, $g_{22} =\pi_j(\bl)(1-\pi_j(\bl))(\alpha_j(\bl)(2\pi_j(\bl) - 1) - y) \bz_j(\bl) \bz_j(\bl)^\top$.

\section{Spatial meshing of Student-t processes} \label{appx:studentt}
GPs are desirable thanks to their convenient properties; however, a similar construction based on cross-covariances can be used to model $\bw(\cdot)$ as a $q$-variate Student-t process (TP), in which case we write $\bw(\bl) \sim TP_{\nu_0}(\bzero, \Cov(\cdot,\cdot))$ where $\nu_0 > 2 \in \Re$ is a degrees of freedom parameter which controls tail heaviness; similarities with GPs include closedness under marginalization and analytic forms of conditional densities. Then, for any $\calL$, the random effects have a multivariate Student-t distribution, i.e. $\bw_{\calL} \sim MVT_{\nu_0}(0, \Cov_{\calL})$. In the limiting case $\nu_0 \to \infty$ one obtains a GP with cross-covariance $\Cov(\cdot,\cdot)$. \cite{studenttprocess} and \cite{mvstudentreg} introduce and consider TPs as alternatives to GPs in regression, citing improved flexibility owing to the ability of a TP to capture more extreme behavior. There are difficulties associated to using TPs in regression, notably the lack of closedness under linear combinations. This implies that spatial meshing of multivariate TPs built upon a LMC does not equate the LMC of spatially meshed univariate TPs. 

The TP is closed under marginalization and conditioning, which implies that it is relatively easy to build a spatially meshed TP. Letting $\bw_\calL = \bw$ and $\Cov_\calL = \Cov$ for simplicity, the density of a zero-mean TP evaluated at $\bw$, denoted as $MVT(\nu, \bzero, \Cov)$, is defined as \citep{studenttprocess}
\begin{equation*} \label{eq:student_mvdensity}
    \begin{aligned}
    p(\bw \given \nu_0) &= \frac{\Gamma(\frac{\nu + n}{2})}{((\nu-2) \pi)^{n/2}} |\Cov|^{-\frac{1}{2}} \left( 1 + \frac{1}{\nu - 2} \bw^\top \Cov^{-1} \bw \right)^{-\frac{\nu + n}{2}}.
    \end{aligned}
\end{equation*}
The above density formulation leads to $\text{cov}(\bw) = \Cov$. Closedness of the TP under marginalization and conditioning leads to the TP conditional densities also being multivariate t's; we find
\begin{equation*} \label{eq:student_conditional}
    \begin{aligned}
    \pi(\bw_i \given \bw_{[i]}) &\sim MVT\left(\nu + n_{[i]}, \bH_i \bw_{[i]}, \frac{\bb  + \nu- 2}{n_{[i]} + \nu - 2} \bR_i \right),
    \end{aligned}
\end{equation*}
where $\bH_i$ and $\bR_i$ are defined like in the GP, and the new term $\bb = \bw_{[i]}^\top \Cov^{-1}_{[i]} \bw_{[i]}$ determines how the conditional variance of $\bw_i \given \bw_{[i]}$ also depends on the values of $\bw_{[i]}$. In fact, $\text{cov}(\bw_i \given \bw_{[i]}) = \frac{\bb  + \nu- 2}{n_{[i]} + \nu - 2} \bR_i$, where the (covariance-weighted) sum of squares $\bb$ is used to inform the conditional density about the observed variance in the conditioning set. In fact, if $\bb/n_{[i]}$ is large (i.e., the conditioning set has larger spread), then the conditional variance is also larger. This intuitive behavior is missing from a GP, which we obtain in this context by letting $\nu \to \infty$ (or $n_{[i]} \to \infty$, which is uninteresting when doing spatial meshing).

\noindent 
\textbf{Gradient based sampling for MTPs.}\\
When building gradient-based MCMC methods for posterior sampling MTPs, we require $\nabla_{\bw_i} \log p(\bw_i \given \others) = \bbf_i + \frac{\delta}{ \delta \bw_i} \log p(\bw_{i} \given \bw_{[i]}, \btheta) + \sum_{j \to \{i \to j \}}\frac{\delta}{ \delta \bw_i} \log p(\bw_{j} \given \bw_{i}, \bw_{[j]\setminus \{i\}}, \btheta )$. In particular, letting $\br_i = \bw_i - \bH_i \bw_{[i]}$ we find 
\begin{equation*} 
    \begin{aligned}
    \frac{\delta}{ \delta \bw_i} \log \pi(\bw_{i} \given \bw_{[i]}, \btheta) &= -  \frac{\nu + n_i + n_{[i]}}{\nu - 2 + \bw_{[i]} \Cov_{[i]}^{-1} \bw_{[i]} + \br_i^\top \bR_{i}^{-1} \br_i} \bR_{i}^{-1} \br_i,
    \end{aligned}
\end{equation*}
and we proceed similarly for $\nabla_{\bw_i} \log \pi(\bw_j \given \bw_i, \bw_{[j] \setminus \{i \}}, \btheta)$, where $\pi(\bw_j \given \bw_i, \others)$ is a MVT density of $\bw_j$ but not of $\bw_i$ because MVT are not closed under linear combinations. We partition $\bH_j$ and $\Cov_{[j]}^{-1}$ as
\begin{align*}
    \bH_j &= \begin{bmatrix}
    A & B
    \end{bmatrix} \qquad \Cov_{[j]}^{-1} = \begin{bmatrix}
    C & D\\
    D^\top & E
    \end{bmatrix},
\end{align*}
with $A$ and $C$ corresponding to blocks which refer to node $a_i \in [j]$, whereas $B$ and $E$ refer to nodes $[j]\setminus a_i$. Let $\tilde{\bw}_j = \bw_j - B\bw_{[j]\setminus \{ i\}}$, $\alpha = \frac{\nu + n_j + n_{[j]}}{2}$, $\beta = \nu - 2 + \bw_{[j] \setminus \{ i\}}^\top E _{[j] \setminus \{ i\}}\bw_{[j] \setminus \{ i\}}$, $c_1 = \bw_i^\top C \bw_i + 2\bw_i^\top D \bw_{[j]\setminus\{i \}}$, $c_2=(\tilde{\bw}_{j} - A\bw_i)^\top \bR_j^{-1}(\tilde{\bw}_{j} - A\bw_i)$. Then
\begin{equation*} 
    \begin{aligned}
    \nabla_{\bw_i} \log \pi(\bw_j \given \bw_i, \bw_{[j] \setminus \{i \}}, \btheta) &= \frac{\delta}{\delta \bw_i} \left\{ -\alpha \log \left( 1+\frac{(\tilde{\bw}_j - A\bw_i )^\top \bR_j^{-1}(\tilde{\bw}_j - A\bw_i) }{  \bw_i^\top C \bw_i + 2\bw_i ^\top D \bw_{[j] \setminus \{ i\}} + \beta }\right) \right\} \\
    &= \frac{2\alpha}{\beta + c_1 + c_2} \left(A^\top \bR_j^{-1}(\bw_j - A\bw_i) + \frac{c_2 (C\bw_i + D\bw_{[j]\setminus\{i\}})}{\beta + c_1} \right).
    \end{aligned}
\end{equation*}

%However, our novel methodology of Section \ref{sec:melange} can be used to not only sample models with non-Gaussian first stages, but also models with non-Gaussian latent processes; the TP is an example of such general applicability of our method.

%\bibliographystyle{hapalike}
\bibliography{biblio}

\begin{thebibliography}{84}
\providecommand{\natexlab}[1]{#1}
\providecommand{\url}[1]{\texttt{#1}}
\expandafter\ifx\csname urlstyle\endcsname\relax
  \providecommand{\doi}[1]{doi: #1}\else
  \providecommand{\doi}{doi: \begingroup \urlstyle{rm}\Url}\fi

\bibitem[Andrieu and Thoms(2008)]{andrieuthoms2008}
Christophe Andrieu and Johannes Thoms.
\newblock A tutorial on adaptive {MCMC}.
\newblock \emph{Statistics and Computing}, 18:\penalty0 343--373, 2008.
\newblock \doib{10.1007/s11222-008-9110-y}.

\bibitem[Apanasovich and Genton(2010)]{apanasovich_genton2010}
Tatiyana~V. Apanasovich and Marc~G. Genton.
\newblock Cross-covariance functions for multivariate random fields based on
  latent dimensions.
\newblock \emph{Biometrika}, 97:\penalty0 15--30, 2010.
\newblock \doib{10.1093/biomet/asp078}.

\bibitem[Atchadé(2006)]{atchade06}
Yves~F. Atchadé.
\newblock An adaptive version for the {Metropolis} adjusted {Langevin}
  algorithm with a truncated drift.
\newblock \emph{Methodology and Computing in Applied Probability}, 8:\penalty0
  235--254, 2006.
\newblock \doib{10.1007/s11009-006-8550-0}.

\bibitem[Banerjee(2017)]{sudipto_ba17}
Sudipto Banerjee.
\newblock High-dimensional {Bayesian} geostatistics.
\newblock \emph{Bayesian Analysis}, 12\penalty0 (2):\penalty0 583--614, 2017.
\newblock \doib{10.1214/17-BA1056R}.

\bibitem[Banerjee(2020)]{sudipto_ss20}
Sudipto Banerjee.
\newblock Modeling massive spatial datasets using a conjugate {Bayesian} linear
  modeling framework.
\newblock \emph{Spatial Statistics}, 37:\penalty0 100417, 2020.
\newblock \doib{10.1016/j.spasta.2020.100417}.

\bibitem[Banerjee et~al.(2008)Banerjee, Gelfand, Finley, and
  Sang]{gp_predictive_process}
Sudipto Banerjee, Alan~E. Gelfand, Andrew~O. Finley, and Huiyan Sang.
\newblock Gaussian predictive process models for large spatial data sets.
\newblock \emph{Journal of the Royal Statistical Society, Series B},
  70:\penalty0 825--848, 2008.
\newblock \doib{10.1111/j.1467-9868.2008.00663.x}.

\bibitem[Banerjee et~al.(2010)Banerjee, Finley, Waldmann, and
  Ericsson]{gp_pp_biasadj}
Sudipto Banerjee, Andrew~O. Finley, Patrik Waldmann, and Tore Ericsson.
\newblock Hierarchical spatial process models for multiple traits in large
  genetic trials.
\newblock \emph{Journal of American Statistical Association}, 105\penalty0
  (490):\penalty0 506--521, 2010.
\newblock \doib{10.1198/jasa.2009.ap09068}.

\bibitem[Betancourt(2018)]{hmc_conceptual}
Michael Betancourt.
\newblock A conceptual introduction to {Hamiltonian Monte Carlo}, 2018.
\newblock \arXiv{1701.02434}.

\bibitem[Bhattacharya and Dunson(2011)]{bhattacharya_dunson}
A.~Bhattacharya and D.~B. Dunson.
\newblock Sparse {Bayesian} infinite factor models.
\newblock \emph{Biometrika}, 98\penalty0 (2):\penalty0 291--306, 2011.
\newblock \doib{10.1093/biomet/asr013}.

\bibitem[Blomstedt et~al.(2019)Blomstedt, {Parente Paiva Mesquita}, Lintusaari,
  Sivula, Corander, and Kaski]{blomstedtetal2019}
Paul Blomstedt, Diego {Parente Paiva Mesquita}, Jarno Lintusaari, Tuomas
  Sivula, Jukka Corander, and Samuel Kaski.
\newblock Meta-analysis of {Bayesian} analyses, 2019.
\newblock \arXiv{1904.04484}.

\bibitem[Bradley et~al.(2018)Bradley, Holan, and Wikle]{bradley_ba}
Jonathan~R. Bradley, Scott~H. Holan, and Christopher~K. Wikle.
\newblock Computationally efficient multivariate spatio-temporal models for
  high-dimensional count-valued data (with discussion).
\newblock \emph{Bayesian Analysis}, 13\penalty0 (1):\penalty0 253--310, 2018.
\newblock \doib{10.1214/17-BA1069}.

\bibitem[Bradley et~al.(2019)Bradley, Holan, and Wikle]{bradley_jasa}
Jonathan~R. Bradley, Scott~H. Holan, and Christopher~K. Wikle.
\newblock Bayesian hierarchical models with conjugate full-conditional
  distributions for dependent data from the natural exponential family.
\newblock \emph{Journal of the American Statistical Association}, 2019.
\newblock \doib{10.1080/01621459.2019.1677471}.

\bibitem[Carpenter et~al.(2017)Carpenter, Gelman, Hoffman, Lee, Goodrich,
  Betancourt, Brubaker, Li, and Riddell]{stan}
Bob Carpenter, Andrew Gelman, Matthew~D. Hoffman, Daniel Lee, Ben Goodrich,
  Michael Betancourt, Jiqiang Brubaker, Marcus~Guo, Peter Li, and Allen
  Riddell.
\newblock Stan: A probabilistic programming language.
\newblock \emph{Journal of Statistical Software}, 76\penalty0 (1), 2017.
\newblock \doib{10.18637/jss.v076.i01}.

\bibitem[Chen et~al.(2020)Chen, Wang, and Gorban]{mvstudentreg}
Zexun Chen, Bo~Wang, and Alexander~N. Gorban.
\newblock Multivariate {Gaussian} and {Student-t} process regression for
  multi-output prediction.
\newblock \emph{Neural Computing and Applications}, 32:\penalty0 3005–3028,
  2020.
\newblock \doib{10.1007/s00521-019-04687-8}.

\bibitem[Craiu et~al.(2015)Craiu, Gray, Łatuszyński, Madras, Roberts, and
  Rosenthal]{craiuetal15}
Radu~V. Craiu, Lawrence Gray, Krzysztof Łatuszyński, Neal Madras, Gareth~O.
  Roberts, and Jeffrey~S. Rosenthal.
\newblock Stability of adversarial {Markov} chains, with an application to
  adaptive {MCMC} algorithms.
\newblock \emph{The Annals of Applied Probability}, 25\penalty0 (6):\penalty0
  3592 -- 3623, 2015.
\newblock \doib{10.1214/14-AAP1083}.

\bibitem[Cressie and Johannesson(2008)]{frk}
Noel Cressie and Gardar Johannesson.
\newblock Fixed rank kriging for very large spatial data sets.
\newblock \emph{Journal of the Royal Statistical Society, Series B},
  70:\penalty0 209--226, 2008.
\newblock \doib{10.1111/j.1467-9868.2007.00633.x}.

\bibitem[Dagum and Menon(1998)]{dagum1998openmp}
Leonardo Dagum and Ramesh Menon.
\newblock {OpenMP}: an industry standard api for shared-memory programming.
\newblock \emph{Computational Science \& Engineering, IEEE}, 5\penalty0
  (1):\penalty0 46--55, 1998.

\bibitem[Datta et~al.(2016{\natexlab{a}})Datta, Banerjee, Finley, and
  Gelfand]{nngp}
Abhirup Datta, Sudipto Banerjee, Andrew~O. Finley, and Alan~E. Gelfand.
\newblock Hierarchical nearest-neighbor {Gaussian} process models for large
  geostatistical datasets.
\newblock \emph{Journal of the American Statistical Association}, 111:\penalty0
  800--812, 2016{\natexlab{a}}.
\newblock \doib{10.1080/01621459.2015.1044091}.

\bibitem[Datta et~al.(2016{\natexlab{b}})Datta, Banerjee, Finley, Hamm, and
  Schaap]{nngp_aoas}
Abhirup Datta, Sudipto Banerjee, Andrew~O. Finley, Nicholas A.~S. Hamm, and
  Martijn Schaap.
\newblock Nonseparable dynamic nearest neighbor {Gaussian} process models for
  large spatio-temporal data with an application to particulate matter
  analysis.
\newblock \emph{The Annals of Applied Statistics}, 10:\penalty0 1286--1316,
  2016{\natexlab{b}}.
\newblock \doib{10.1214/16-AOAS931}.

\bibitem[Dey et~al.(2021)Dey, Datta, and Banerjee]{deyetal20}
Debangan Dey, Abhirup Datta, and Sudipto Banerjee.
\newblock Graphical {Gaussian} process models for highly multivariate spatial
  data.
\newblock \emph{Biometrika}, 2021.
\newblock in press. \doib{doi.org/10.1093/biomet/asab061}.

\bibitem[Dorazio and Royle(2005)]{DorazioRoyle05}
Robert~M Dorazio and J.~Andrew Royle.
\newblock Estimating size and composition of biological communities by modeling
  the occurrence of species.
\newblock \emph{Journal of the American Statistical Association}, 100\penalty0
  (470):\penalty0 389--398, 2005.
\newblock \doib{10.1198/016214505000000015}.

\bibitem[Doser et~al.(2022)Doser, Finley, K\'ery, and Zipkin]{Doseretall2022}
Jeffrey~W. Doser, Andrew~O. Finley, Marc K\'ery, and Elise~F. Zipkin.
\newblock {spOccupancy}: An {R} package for single-species, multi-species, and
  integrated spatial occupancy models.
\newblock \emph{Methods in Ecology and Evolution}, 13\penalty0 (8):\penalty0
  1670--1678, 2022.
\newblock \doib{10.1111/2041-210X.13897}.

\bibitem[Duane et~al.(1987)Duane, A.D., Pendleton, and Roweth]{hmc_duane}
Simon Duane, Kennedy A.D., Brian~J. Pendleton, and Duncan Roweth.
\newblock {Hybrid Monte Carlo}.
\newblock \emph{Physics Letters B}, 195:\penalty0 216--222, 1987.

\bibitem[Dunson and Johndrow(2020)]{hastings50}
David Dunson and James~E. Johndrow.
\newblock The {Hastings} algorithm at fifty.
\newblock \emph{Biometrika}, 107\penalty0 (1):\penalty0 1--23, 2020.
\newblock \doib{10.1093/biomet/asz066}.

\bibitem[Finley et~al.(2008)Finley, Banerjee, Ek, and McRoberts]{finley2008}
Andrew~O. Finley, Sudipto Banerjee, Alan~R. Ek, and Ronald~E. McRoberts.
\newblock Bayesian multivariate process modeling for prediction of forest
  attributes.
\newblock \emph{Journal of Agricultural, Biological, and Environmental
  Statistics}, 13:\penalty0 60, 2008.
\newblock \doib{10.1198/108571108X273160}.

\bibitem[Finley et~al.(2019)Finley, Datta, Cook, Morton, Andersen, and
  Banerjee]{nngp_algos}
Andrew~O. Finley, Abhirup Datta, Bruce~D. Cook, Douglas~C. Morton, Hans~E.
  Andersen, and Sudipto Banerjee.
\newblock Efficient algorithms for {Bayesian} nearest neighbor {Gaussian}
  processes.
\newblock \emph{Journal of Computational and Graphical Statistics},
  28:\penalty0 401--414, 2019.
\newblock \doib{10.1080/10618600.2018.1537924}.

\bibitem[Furrer et~al.(2006)Furrer, Genton, and Nychka]{taper1}
Reinhard Furrer, Marc~G. Genton, and Douglas Nychka.
\newblock Covariance tapering for interpolation of large spatial datasets.
\newblock \emph{Journal of Computational and Graphical Statistics},
  15:\penalty0 502--523, 2006.
\newblock \doib{10.1198/106186006X132178}.

\bibitem[Genton and Kleiber(2015)]{genton_ccov}
Marc~G. Genton and William Kleiber.
\newblock Cross-covariance functions for multivariate geostatistics.
\newblock \emph{Statistical Science}, 30:\penalty0 147--163, 2015.
\newblock \doib{10.1214/14-STS487}.

\bibitem[Girolami and Calderhead(2011)]{girolamicalderhead11}
Mark Girolami and Ben Calderhead.
\newblock Riemann manifold {Langevin} and {Hamiltonian Monte Carlo} methods.
\newblock \emph{Journal of the Royal Statistical Society: Series B},
  73\penalty0 (2):\penalty0 123--214, 2011.
\newblock \doib{10.1111/j.1467-9868.2010.00765.x}.

\bibitem[Gramacy and Apley(2015)]{lagp}
Robert~B. Gramacy and Daniel~W. Apley.
\newblock Local {Gaussian} process approximation for large computer
  experiments.
\newblock \emph{Journal of Computational and Graphical Statistics}, 24\penalty0
  (2):\penalty0 561–578, 2015.
\newblock \doib{10.1080/10618600.2014.914442}.

\bibitem[Guhaniyogi and Banerjee(2018)]{metakriging}
Rajarshi Guhaniyogi and Sudipto Banerjee.
\newblock Meta-kriging: Scalable {Bayesian} modeling and inference for massive
  spatial datasets.
\newblock \emph{Technometrics}, 60\penalty0 (4):\penalty0 430--444, 2018.
\newblock \doib{10.1080/00401706.2018.1437474}.

\bibitem[Haario et~al.(2001)Haario, Saksman, and Tamminen]{haario2001}
Heikki Haario, Eero Saksman, and Johanna Tamminen.
\newblock An adaptive {Metropolis} algorithm.
\newblock \emph{Bernoulli}, 7\penalty0 (2):\penalty0 223--242, 2001.
\newblock \doib{10.2307/3318737}.

\bibitem[Heaton et~al.(2019)Heaton, Datta, Finley, Furrer, Guinness,
  Guhaniyogi, Gerber, Gramacy, Hammerling, Katzfuss, Lindgren, Nychka, Sun, and
  Zammit-Mangion]{Heaton2019}
Matthew~J. Heaton, Abhirup Datta, Andrew~O. Finley, Reinhard Furrer, Joseph
  Guinness, Rajarshi Guhaniyogi, Florian Gerber, Robert~B. Gramacy, Dorit
  Hammerling, Matthias Katzfuss, Finn Lindgren, Douglas~W. Nychka, Furong Sun,
  and Andrew Zammit-Mangion.
\newblock A case study competition among methods for analyzing large spatial
  data.
\newblock \emph{Journal of Agricultural, Biological and Environmental
  Statistics}, 24\penalty0 (3):\penalty0 398--425, Sep 2019.
\newblock \doib{10.1007/s13253-018-00348-w}.

\bibitem[Hoffman and Gelman(2014)]{nuts}
Matthew~D. Hoffman and Andrew Gelman.
\newblock The no-{U}-turn sampler: Adaptively setting path lengths in
  {Hamiltonian Monte Carlo}.
\newblock \emph{Journal of Machine Learning Research}, 15\penalty0
  (47):\penalty0 1593--1623, 2014.
\newblock \url{https://www.jmlr.org/papers/v15/hoffman14a.html}.

\bibitem[Jin et~al.(2021)Jin, Peruzzi, and Dunson]{bags}
Bora Jin, Michele Peruzzi, and David~B. Dunson.
\newblock Bag of {DAGs}: Flexible \& scalable modeling of spatiotemporal
  dependence, 2021.
\newblock \arXiv{2112.11870}.

\bibitem[Jin et~al.(2022)Jin, Herring, and Dunson]{boragp}
Bora Jin, Amy~H. Herring, and David~B. Dunson.
\newblock Spatial predictions on physically constrained domains: Applications
  to arctic sea salinity data, 2022.
\newblock \arXiv{2210.03913}.

\bibitem[Johndrow et~al.(2020)Johndrow, Pillai, and Smith]{johndrowetal20}
James~E. Johndrow, Natesh~S. Pillai, and Aaron Smith.
\newblock No free lunch for approximate {MCMC}, 2020.
\newblock \arXiv{2010.125147}.

\bibitem[J\"onsson et~al.(2010)J\"onsson, Eklundh, Hellstr\"om, B\"arring, and
  J\"onsson]{jonsson2010modis_snow}
A.~M. J\"onsson, L.~Eklundh, M.~Hellstr\"om, L.~B\"arring, and P.~J\"onsson.
\newblock Annual changes in {MODIS} vegetation indices of {S}wedish coniferous
  forests in relation to snow dynamics and tree phenology.
\newblock \emph{Remote Sensing of Environment}, 114:\penalty0 2719–2730,
  2010.
\newblock \doib{10.1016/j.rse.2010.06.005}.

\bibitem[Jurek and Katzfuss(2020)]{jurekkatzfuss2020}
Marcin Jurek and Matthias Katzfuss.
\newblock Hierarchical sparse {Cholesky} decomposition with applications to
  high-dimensional spatio-temporal filtering, 2020.
\newblock \arXiv{2006.16901}.

\bibitem[Katzfuss(2017)]{katzfuss_jasa17}
Matthias Katzfuss.
\newblock A multi-resolution approximation for massive spatial datasets.
\newblock \emph{Journal of the American Statistical Association}, 112:\penalty0
  201--214, 2017.
\newblock \doib{10.1080/01621459.2015.1123632}.

\bibitem[Katzfuss and Guinness(2021)]{katzfuss_vecchia}
Matthias Katzfuss and Joseph Guinness.
\newblock A general framework for {Vecchia} approximations of {Gaussian}
  processes.
\newblock \emph{Statistical Science}, 36\penalty0 (1):\penalty0 124--141, 2021.
\newblock \doib{10.1214/19-STS755}.

\bibitem[Kaufman et~al.(2008)Kaufman, Schervish, and Nychka]{taper2}
Cari~G. Kaufman, Mark~J. Schervish, and Douglas~W. Nychka.
\newblock Covariance tapering for likelihood-based estimation in large spatial
  data sets.
\newblock \emph{Journal of the American Statistical Association}, 103:\penalty0
  1545--1555, 2008.
\newblock \doib{10.1198/016214508000000959}.

\bibitem[Lindgren et~al.(2011)Lindgren, Rue, and Lindström]{spde}
Finn Lindgren, Håvard Rue, and Johan Lindström.
\newblock An explicit link between {Gaussian} fields and {Gaussian} {Markov}
  random fields: the stochastic partial differential equation approach.
\newblock \emph{Journal of the Royal Statistical Society: Series B},
  73:\penalty0 423--498, 2011.
\newblock \doib{10.1111/j.1467-9868.2011.00777.x}.

\bibitem[Marshall and Roberts(2012)]{marshallroberts12}
Tristan Marshall and Gareth Roberts.
\newblock An adaptive approach to {Langevin} {MCMC}.
\newblock \emph{Statistics and Computing}, 22:\penalty0 1041--1057, 2012.
\newblock \doib{10.1007/s11222-011-9276-6}.

\bibitem[Matheron(1982)]{matheron82}
G.~Matheron.
\newblock Pour une analyse krigeante des données régionalisées.
\newblock \emph{Technical report N.732, Centre de Géostatistique}, 1982.

\bibitem[Mesquita et~al.(2020)Mesquita, Blomstedt, and Kaski]{mesquitaetal2020}
Diego Mesquita, Paul Blomstedt, and Samuel Kaski.
\newblock Embarrassingly parallel {MCMC} using deep invertible transformations.
\newblock In Ryan~P. Adams and Vibhav Gogate, editors, \emph{Proceedings of
  Machine Learning Research}, volume 115, pages 1244--1252, Tel Aviv, Israel,
  22--25 Jul 2020. PMLR.
\newblock \url{http://proceedings.mlr.press/v115/mesquita20a.html}.

\bibitem[Mimnagh et~al.(2022)Mimnagh, Parnell, Prado, and
  de~Andrade~Moral]{Mimnaghetal22}
Niamh Mimnagh, Andrew Parnell, Estev\~ao Prado, and Rafael de~Andrade~Moral.
\newblock Bayesian multi-species {N}-mixture models for unmarked animal
  communities.
\newblock \emph{Environmental and Ecological Statistics}, 29:\penalty0
  755--778, 2022.
\newblock \doib{10.1007/s10651-022-00542-7}.

\bibitem[Murray et~al.(2010)Murray, Adams, and MacKay]{elliptss}
Iain Murray, Ryan Adams, and David MacKay.
\newblock Elliptical slice sampling.
\newblock In Yee~Whye Teh and Mike Titterington, editors, \emph{Proceedings of
  the Thirteenth International Conference on Artificial Intelligence and
  Statistics}, volume~9 of \emph{Proceedings of Machine Learning Research},
  pages 541--548, Chia Laguna Resort, Sardinia, Italy, 13--15 May 2010. PMLR.
\newblock \url{https://proceedings.mlr.press/v9/murray10a.html}.

\bibitem[Neal(2011)]{hmc_neal}
R.~M. Neal.
\newblock {MCMC} using {Hamiltonian} dynamics.
\newblock In S.~Brooks, A.~Gelman, G.~L. Jones, and X.-L. Meng, editors,
  \emph{Handbook of {Markov} Chain Monte Carlo}. CRC Press, New York, 2011.
\newblock \doib{10.1201/b10905}.

\bibitem[Neiswanger et~al.(2014)Neiswanger, Wang, and Xing]{neiswangeretal2014}
Willie Neiswanger, Chong Wang, and Eric~P. Xing.
\newblock Asymptotically exact, embarrassingly parallel {MCMC}.
\newblock In \emph{Proceedings of the Thirtieth Conference on Uncertainty in
  Artificial Intelligence}, UAI'14, page 623–632, Arlington, Virginia, USA,
  2014. AUAI Press.
\newblock ISBN 9780974903910.

\bibitem[Nemeth and Sherlock(2018)]{nemeth2018}
Christopher Nemeth and Chris Sherlock.
\newblock Merging {MCMC} subposteriors through {Gaussian}-process
  approximations.
\newblock \emph{Bayesian Analysis}, 13\penalty0 (2):\penalty0 507--530, 06
  2018.
\newblock \doib{10.1214/17-BA1063}.

\bibitem[Peruzzi and Dunson(2022)]{spamtrees}
Michele Peruzzi and David~B. Dunson.
\newblock Spatial multivariate trees for big data {Bayesian} regression.
\newblock \emph{Journal of Machine Learning Research}, 23\penalty0
  (17):\penalty0 1--40, 2022.
\newblock \url{http://jmlr.org/papers/v23/20-1361.html}.

\bibitem[Peruzzi et~al.(2021)Peruzzi, Banerjee, Dunson, and Finley]{grips}
Michele Peruzzi, Sudipto Banerjee, David~B. Dunson, and Andrew~O. Finley.
\newblock {G}rid-{P}arametrize-{S}plit {(GriPS)} for improved scalable
  inference in spatial big data analysis, 2021.
\newblock \arXiv{2101.03579}.

\bibitem[Peruzzi et~al.(2022)Peruzzi, Banerjee, and Finley]{meshedgp}
Michele Peruzzi, Sudipto Banerjee, and Andrew~O. Finley.
\newblock Highly scalable {Bayesian} geostatistical modeling via meshed
  {Gaussian} processes on partitioned domains.
\newblock \emph{Journal of the American Statistical Association}, 117\penalty0
  (538):\penalty0 969--982, 2022.
\newblock \doib{10.1080/01621459.2020.1833889}.

\bibitem[Roberts and Rosenthal(2007)]{robertsrosenthal07}
Gareth~O. Roberts and Jeffrey~S. Rosenthal.
\newblock Coupling and ergodicity of adaptive {Markov} chain {Monte} {Carlo}
  algorithms.
\newblock \emph{Journal of Applied Probability}, 44:\penalty0 458--475, 2007.
\newblock \doib{10.1239/jap/1183667414}.

\bibitem[Roberts and Rosenthal(2009)]{robertsrosenthal09examples}
Gareth~O. Roberts and Jeffrey~S. Rosenthal.
\newblock Examples of adaptive {MCMC}.
\newblock \emph{Journal of Computational and Graphical Statistics}, 18\penalty0
  (2):\penalty0 349--367, 2009.
\newblock \doib{10.1198/jcgs.2009.06134}.

\bibitem[Roberts and Stramer(2002)]{robertsstramer02}
Gareth~O. Roberts and Osnat Stramer.
\newblock Langevin diffusions and {Metropolis-Hastings} algorithms.
\newblock \emph{Methodology And Computing In Applied Probability}, 4:\penalty0
  337–357, 2002.
\newblock \doib{10.1023/A:1023562417138}.

\bibitem[Roberts and Tweedie(1996)]{robertstweedie96}
Gareth~O. Roberts and Richard~L. Tweedie.
\newblock Exponential convergence of {Langevin} distributions and their
  discrete approximations.
\newblock \emph{Bernoulli}, 2\penalty0 (4):\penalty0 341--363, 1996.

\bibitem[Royle(2004)]{Royle2004}
J.A. Royle.
\newblock {N-Mixture Models for Estimating Population Size from Spatially
  Replicated Counts}.
\newblock \emph{Biometrics}, 60\penalty0 (1):\penalty0 108--115, 2004.
\newblock \doib{10.1111/j.0006-341X.2004.00142.x}.

\bibitem[Rue and Held(2005)]{grmfields}
Havard Rue and Leonhard Held.
\newblock \emph{Gaussian {Markov} Random Fields: Theory and Applications}.
\newblock Chapman \& Hall/CRC, 2005.
\newblock \doib{10.1007/978-3-642-20192-9}.

\bibitem[Rue et~al.(2009)Rue, Martino, and Chopin]{inla}
Håvard Rue, Sara Martino, and Nicolas Chopin.
\newblock Approximate {Bayesian} inference for latent {Gaussian} models by
  using integrated nested laplace approximations.
\newblock \emph{Journal of the Royal Statistical Society: Series B},
  71:\penalty0 319--392, 2009.
\newblock \doib{10.1111/j.1467-9868.2008.00700.x}.

\bibitem[Sang and Huang(2012)]{fsa}
Huiyan Sang and Jianhua~Z. Huang.
\newblock A full scale approximation of covariance functions for large spatial
  data sets.
\newblock \emph{Journal of the Royal Statistical Society, Series B},
  74:\penalty0 111--132, 2012.
\newblock \doib{10.1111/j.1467-9868.2011.01007.x}.

\bibitem[Schmidt and Gelfand(2003)]{schmidtgelfand}
Alexandra~M. Schmidt and Alan~E. Gelfand.
\newblock A {Bayesian} coregionalization approach for multivariate pollutant
  data.
\newblock \emph{Journal of Geophysical Research}, 108:\penalty0 D24, 2003.
\newblock \doib{10.1029/2002JD002905}.

\bibitem[Sengupta and Cressie(2013)]{sengupta_cressie}
Aritra Sengupta and Noel Cressie.
\newblock Hierarchical statistical modeling of big spatial datasets using the
  exponential family of distributions.
\newblock \emph{Spatial Statistics}, 2013.
\newblock \doib{10.1016/j.spasta.2013.02.002}.

\bibitem[Shah et~al.(2014)Shah, Wilson, and Ghahramani]{studenttprocess}
Amar Shah, Andrew~G. Wilson, and Zoubin Ghahramani.
\newblock Student-t processes as alternatives to {Gaussian} processes.
\newblock In \emph{Proceedings of the 17th International Conference on
  Artificial Intelligence and Statistics (AISTATS)}, 2014.

\bibitem[Shirota et~al.(2019)Shirota, Finley, Cook, and Banerjee]{conjnngp}
Shinichiro Shirota, Andrew~O. Finley, Bruce~D. Cook, and Sudipto Banerjee.
\newblock Conjugate nearest neighbor {Gaussian} process models for efficient
  statistical interpolation of large spatial data, 2019.
\newblock \arXiv{1907.10109}.

\bibitem[Stein(2002)]{stein_screening}
Michael~L. Stein.
\newblock The screening effect in kriging.
\newblock \emph{The Annals of Statistics}, 30\penalty0 (1):\penalty0 298--323,
  2002.
\newblock \doib{10.1214/aos/1015362194}.

\bibitem[Stein(2014)]{stein2014}
Michael~L. Stein.
\newblock Limitations on low rank approximations for covariance matrices of
  spatial data.
\newblock \emph{Spatial Statistics}, 8:\penalty0 1--19, 2014.
\newblock \doib{doi:10.1016/j.spasta.2013.06.003}.

\bibitem[Stein et~al.(2004)Stein, Chi, and Welty]{steinetal2004}
Michael~L. Stein, Zhiyi Chi, and Leah~J. Welty.
\newblock Approximating likelihoods for large spatial data sets.
\newblock \emph{Journal of the Royal Statistical Society, Series B},
  66:\penalty0 275--296, 2004.
\newblock \doib{10.1046/j.1369-7412.2003.05512.x}.

\bibitem[Sun et~al.(2011)Sun, Li, and Genton]{sunligenton}
Y.~Sun, B.~Li, and M.~Genton.
\newblock Geostatistics for large datasets.
\newblock In J.~Montero, E.~Porcu, and M.~Schlather, editors, \emph{Advances
  and Challenges in Space-time Modelling of Natural Events}, pages 55--77.
  Springer-Verlag, Berlin Heidelberg, 2011.
\newblock \doib{10.1007/978-3-642-17086-7}.

\bibitem[Taylor and Diggle(2014)]{taylordiggle14}
Benjamin~M. Taylor and Peter~J. Diggle.
\newblock {INLA} or {MCMC}? a tutorial and comparative evaluation for spatial
  prediction in log-{Gaussian} {Cox} processes.
\newblock \emph{Journal of Statistical Computation and Simulation}, 84\penalty0
  (10):\penalty0 2266--2284, 2014.
\newblock \doib{10.1080/00949655.2013.788653}.

\bibitem[Tikhonov et~al.(2020)Tikhonov, Opedal, Abrego, Lehikoinen, de~Jonge,
  Oksanen, and Ovaskainen]{hmsc_package}
Gleb Tikhonov, Oystein~H. Opedal, Nerea Abrego, Aleksi Lehikoinen, Melinda
  M.~J. de~Jonge, Jari Oksanen, and Otso Ovaskainen.
\newblock Joint species distribution modelling with the {R}-package {Hmsc}.
\newblock \emph{Methods in Ecology and Evolution}, 11\penalty0 (3):\penalty0
  442--447, 2020.
\newblock \doib{10.1111/2041-210X.13345}.

\bibitem[Vecchia(1988)]{vecchia88}
A.~V. Vecchia.
\newblock Estimation and model identification for continuous spatial processes.
\newblock \emph{Journal of the Royal Statistical Society, Series B},
  50:\penalty0 297--312, 1988.
\newblock \doib{10.1111/j.2517-6161.1988.tb01729.x}.

\bibitem[Vihola(2012)]{vihola2012}
Matti Vihola.
\newblock Robust adaptive {Metropolis} algorithm with coerced acceptance rate.
\newblock \emph{Statistics and Computing}, 22:\penalty0 997--1008, 2012.
\newblock \doib{10.1007/s11222-011-9269-5}.

\bibitem[Wackernagel(2003)]{wackernagel03}
Hans Wackernagel.
\newblock \emph{Multivariate Geostatistics: An Introduction with Applications}.
\newblock Springer, Berlin, 2003.
\newblock \doib{10.1007/978-3-662-05294-5}.

\bibitem[Walker et~al.(1993)Walker, Halfpenny, Walker, and
  Wessman]{longterm_snow}
D.~A. Walker, James~C. Halfpenny, Marilyn~D. Walker, and Carol~A. Wessman.
\newblock Long-term studies of snow-vegetation interactions.
\newblock \emph{BioScience}, 43\penalty0 (5):\penalty0 287--301, 1993.
\newblock \doib{10.2307/1312061}.

\bibitem[Wang et~al.(2015{\natexlab{a}})Wang, Zhang, Qiu, Ji, Tian, Wang, and
  Wang]{snow_effects}
Kun Wang, Li~Zhang, Yubao Qiu, Lei Ji, Feng Tian, Cuizhen Wang, and Zhiyong
  Wang.
\newblock Snow effects on alpine vegetation in the {Qinghai-Tibetan} plateau.
\newblock \emph{International Journal of Digital Earth}, 8\penalty0
  (1):\penalty0 58--75, 2015{\natexlab{a}}.
\newblock \doib{10.1080/17538947.2013.848946}.

\bibitem[Wang and Dunson(2014)]{wangdunson2014}
Xiangyu Wang and David~B. Dunson.
\newblock Parallelizing {MCMC} via {Weierstrass} sampler, 2014.
\newblock \arXiv{1312.4605}.

\bibitem[Wang et~al.(2015{\natexlab{b}})Wang, Guo, Heller, and
  Dunson]{wangetal2015}
Xiangyu Wang, Fangjian Guo, Katherine~A. Heller, and David~B. Dunson.
\newblock Parallelizing {MCMC} with random partition trees.
\newblock In \emph{Proceedings of the 28th International Conference on Neural
  Information Processing Systems - Volume 1}, NIPS'15, page 451–459,
  Cambridge, MA, USA, 2015{\natexlab{b}}. MIT Press.
\newblock \arXiv{1506.03164}.

\bibitem[Xie et~al.(2020)Xie, Jonas, Rixen, {de Jong}, Garonna, Notarnicola,
  Asam, Schaepman, and Kneub\"uhler]{xieetal2020}
Jing Xie, Tobias Jonas, Christian Rixen, Rogier {de Jong}, Irene Garonna,
  Claudia Notarnicola, Sarah Asam, Michael~E. Schaepman, and Mathias
  Kneub\"uhler.
\newblock Land surface phenology and greenness in {Alpine} grasslands driven by
  seasonal snow and meteorological factors.
\newblock \emph{Science of The Total Environment}, 725:\penalty0 138380, 2020.
\newblock \doib{10.1016/j.scitotenv.2020.138380}.

\bibitem[Zanella and Roberts(2021)]{zanellaroberts21}
Giacomo Zanella and Gareth Roberts.
\newblock Multilevel linear models, gibbs samplers and multigrid
  decompositions.
\newblock \emph{Bayesian Analysis}, 2021.
\newblock \doib{10.1214/20-BA1242}.

\bibitem[Zhang and Banerjee(2022)]{zhangbanerjee20}
Lu~Zhang and Sudipto Banerjee.
\newblock Spatial factor modeling: A {Bayesian} matrix-normal approach for
  misaligned data.
\newblock \emph{Biometrics}, 78\penalty0 (2):\penalty0 560--573, 2022.
\newblock \doib{10.1111/biom.13452}.

\bibitem[Zhu et~al.(2022)Zhu, Peruzzi, Li, and Dunson]{radgp}
Yichen Zhu, Michele Peruzzi, Cheng Li, and David~B. Dunson.
\newblock Radial neighbors for provably accurate scalable approximations of
  {Gaussian} processes, 2022.
\newblock \arXiv{2211.14692}.

\bibitem[Zilber and Katzfuss(2020)]{vecchialaplace}
Daniel Zilber and Matthias Katzfuss.
\newblock {Vecchia-Laplace} approximations of generalized {Gaussian} processes
  for big non-{Gaussian} spatial data, 2020.
\newblock \arXiv{1906.07828}.

\end{thebibliography}

\end{document}